\documentclass{IEEEojcsys}

\usepackage[colorlinks,urlcolor=blue,linkcolor=blue,citecolor=blue]{hyperref}

\usepackage[english]{babel}
\usepackage[utf8]{inputenc}
\usepackage{amsmath}
\usepackage{amsfonts}
\usepackage{amsthm}
\usepackage{xcolor}

\usepackage{upgreek}

\usepackage{graphicx}
\ifCLASSOPTIONcompsoc
\usepackage[caption=false,font=normalsize,labelfon
t=sf,textfont=sf]{subfig}
\else
\usepackage[caption=false,font=footnotesize]{subfig}
\fi
\graphicspath{ {images/}{images/clustering/}{images/bisimulation/}{images/stab/}{images/LQR/}{images/newImages/}{images/photos/}}

\usepackage{pifont}

\usepackage[ruled, vlined, linesnumbered]{algorithm2e}

\usepackage{xstring}

\usepackage{cite}
\usepackage[resetlabels, labeled]{multibib} 
\newcites{R}{Response References}

\usepackage{multirow}
\usepackage{calc}
\usepackage{adjustbox}
\usepackage{booktabs}

\usepackage{makecell}

\usepackage{bbm}
\newtheorem{setup}{Setup}
\newtheorem{fact}{Fact}

\newtheorem{problem}{Problem}
\newtheorem*{problem*}{Problem}

\newtheorem{theorem}{Theorem}
\newtheorem{lemma}{Lemma}
\newtheorem{corollary}{Corollary}
\newtheorem{proposition}{Proposition}

\newtheorem{definition}{Definition}
\newtheorem{assumption}{Assumption}

\let\labelindent\relax
\usepackage{enumitem}

\allowdisplaybreaks

\newcommand{\bv}[1]{\mathbf{#1}}		
\newcommand{\bvgrk}[1]{{\boldsymbol{#1}}}	

\newcommand{\jbar}{\bar{j \phantom{\tiny \,}} \kern -0.1em}

\newcommand{\ahat}{\hat{a}}

\newcommand{\ghat}{\hat{g}}

\newcommand{\jhat}{\hat{j \phantom{\tiny \,}} \kern -0.1em}

\newcommand{\phat}{\hat{p}}
\newcommand{\qhat}{\hat{q}}
\newcommand{\rhat}{\hat{r}}

\newcommand{\jtil}{\tilde{j \phantom{\tiny \,}} \kern -0.1em}

\newcommand{\va}{\bv{a}}
\newcommand{\vb}{\bv{b}}
\newcommand{\vc}{\bv{c}}

\newcommand{\vj}{\bv{j}}

\newcommand{\vs}{\bv{s}}

\newcommand{\vu}{\bv{u}}

\newcommand{\vw}{\bv{w}}
\newcommand{\vx}{\bv{x}}
\newcommand{\vy}{\bv{y}}

\newcommand{\vjbar}{\bar{\vj \phantom{\tiny \,}} \kern -0.1em}

\newcommand{\vjhat}{\hat{\vj \phantom{\tiny \,}} \kern -0.1em}

\newcommand{\vuhat}{\hat{\vu}}

\newcommand{\vxhat}{\hat{\vx}}

\newcommand{\vjtil}{\tilde{\vj \phantom{\tiny \,}} \kern -0.1em}

\newcommand{\Abar}{\bar{A}}
\newcommand{\Bbar}{\bar{B}}

\newcommand{\Jbar}{\bar{J}}

\newcommand{\Wbar}{\bar{W}}

\newcommand{\Jhat}{\hat{J}}

\newcommand{\Ttil}{\tilde{T}}

\newcommand{\Acal}{\mathcal{A}}

\newcommand{\Dcal}{\mathcal{D}}

\newcommand{\Gcal}{\mathcal{G}}
\newcommand{\Hcal}{\mathcal{H}}

\newcommand{\Kcal}{\mathcal{K}}
\newcommand{\Lcal}{\mathcal{L}}

\newcommand{\Ocal}{\mathcal{O}}

\newcommand{\Rcal}{\mathcal{R}}
\newcommand{\Scal}{\mathcal{S}}
\newcommand{\Tcal}{\mathcal{T}}

\newcommand{\Xcal}{\mathcal{X}}

\newcommand{\vA}{\bv{A}}
\newcommand{\vB}{\bv{B}}
\newcommand{\vC}{\bv{C}}
\newcommand{\vD}{\bv{D}}
\newcommand{\vE}{\bv{E}}
\newcommand{\vF}{\bv{F}}

\newcommand{\vH}{\bv{H}}
\newcommand{\vI}{\bv{I}}

\newcommand{\vK}{\bv{K}}

\newcommand{\vM}{\bv{M}}

\newcommand{\vO}{\bv{O}}
\newcommand{\vP}{\bv{P}}
\newcommand{\vQ}{\bv{Q}}
\newcommand{\vR}{\bv{R}}
\newcommand{\vS}{\bv{S}}
\newcommand{\vT}{\bv{T}}
\newcommand{\vU}{\bv{U}}
\newcommand{\vV}{\bv{V}}
\newcommand{\vW}{\bv{W}}
\newcommand{\vX}{\bv{X}}

\newcommand{\vAbar}{\bar{\bv{A}}}
\newcommand{\vBbar}{\bar{\bv{B}}}

\newcommand{\vKbar}{\bar{\bv{K}}}

\newcommand{\vPbar}{\bar{\bv{P}}}

\newcommand{\vSbar}{\bar{\bv{S}}}
\newcommand{\vTbar}{\bar{\bv{T}}}
\newcommand{\vUbar}{\bar{\bv{U}}}

\newcommand{\vAhat}{\hat{\bv{A}}}
\newcommand{\vBhat}{\hat{\bv{B}}}

\newcommand{\vKhat}{\hat{\bv{K}}}

\newcommand{\vPhat}{\hat{\bv{P}}}

\newcommand{\vShat}{\hat{\bv{S}}}
\newcommand{\vThat}{\hat{\bv{T}}}

\usepackage{bm}
\newcommand{\vAcal}{\bm{\mathcal{A}}}

\newcommand{\vAcalhat}{\hat{\bm{\Acal} \phantom{a}} \kern-0.5em}
\newcommand{\vAcalcheck}{\check{\bm{\Acal} \phantom{a}} \kern-0.5em}
\newcommand{\vAcalbar}{\bar{\bm{\Acal} \phantom{a}} \kern-0.5em}

\newcommand{\kappabar}{\bar{\kappa}}

\newcommand{\taubar}{\bar{\tau}}

\newcommand{\xihat}{\hat{\xi}}

\newcommand{\rhohat}{\hat{\rho}}

\newcommand{\omegahat}{\hat{\omega}}

\newcommand{\vmu}{\bvgrk{\upmu}}

\newcommand{\vpi}{\bvgrk{\uppi}}

\newcommand{\vmuhat}{\hat{\bvgrk{\upmu}}}

\newcommand{\vDelta}{\bvgrk{\Delta}}

\newcommand{\vLambda}{\bvgrk{\Lambda}}

\newcommand{\vSigma}{\bvgrk{\Sigma}}
\newcommand{\vPhi}{\bvgrk{\Phi}}

\newcommand{\norm}[1]{\|{#1} \|}

\newcommand{\curlybrackets}[1]{\{ #1 \}}
\newcommand{\squarebracketsbig}[1]{\left[ #1 \right]}

\newcommand{\parenthesesbig}[1]{\left( #1 \right)}

\newcommand{\indicator}[1]{\mathbf{1}_{\curlybrackets{#1}}}

\newcommand{\R}{\mathbb{R}}				

\newcommand{\dm}[2]
{
	\IfStrEq{#2}{1}{\R^{#1}}{\R^{#1 \x #2}}
}

\newcommand{\N}{\mathcal{N}}			
\newcommand{\onevec}{\mathbf{1}} 		

\newcommand{\rank}{\textup{{rank}}} 			
\newcommand{\tr}{\textup{\textbf{tr}}} 			
\newcommand{\diag}{\textup{\textbf{diag}}} 		
\newcommand{\vek}{\textup{\textbf{vec}}}			
\newcommand{\expctn}{\mathbb{E}} 	 	
\newcommand{\inv}{{\mskip-1mu \scalebox{1.5}[0.75]{-} \mskip-1mu 1}}	
\newcommand{\invv}[1]{{\mskip-1mu \scalebox{1.5}[0.75]{-} \mskip-1mu #1}} %
\newcommand{\prob}{\mathbb{P}} 				
\newcommand{\fro}{\textup{F}}			

\usepackage{mathtools}
\makeatletter
\newcommand{\raisemath}[1]{\mathpalette{\raisem@th{#1}}}
\newcommand{\raisem@th}[3]{\raisebox{#1}{$#2#3$}}
\makeatother

\usepackage{amssymb}
\makeatletter
\newcommand*{\T}{{\mathpalette\@transpose{}}}
\newcommand*{\@transpose}[2]{\raisebox{\depth}{$\m@th#1\intercal$}}
\makeatother

\newcommand*{\x}{\mathsf{x}\mskip1mu} 	

\newcommand*{\minus}{\mathsf{-}\mskip1mu} 

\newcommand{\splitatcommas}[1]{%
	\begingroup
	\begingroup\lccode`~=`, \lowercase{\endgroup
		\edef~{\mathchar\the\mathcode`, \penalty0 \noexpand\hspace{0pt plus .1em}}%
	}\mathcode`,="8000 #1%
	\endgroup
}

\newcommand{\Item}[1]{%
	\ifx\relax#1\relax  \item \else \item[#1] \fi
	\abovedisplayskip=0pt\abovedisplayshortskip=0pt~\vspace*{-\baselineskip}
} 

\newcounter{reviewer}
\newcounter{point}[reviewer]
\makeatletter
\newcommand{\customlabel}[2]{%
   \protected@write \@auxout {}{\string \newlabel {#1}{{#2}{\thepage}{#2}{#1}{}} }%
   \hypertarget{#1}{#2}
}
\makeatother

\let\debugmode\undefined 

\ifdefined\debugmode
\newcommand{\red}[1]{{\color{red} #1}}

\definecolor{revColor}{RGB}{0, 0, 255}
\newcommand{\zheinline}[1]{{\color{red} (Zhe: #1)}}
\newcommand{\zhe}[1]{\marginpar{\hspace{-10pt}\color{blue}\tiny\ttfamily Zhe: #1}}
\newcommand{\ysatt}[1]{\marginpar{\color{red}\tiny\ttfamily YS: #1}}
\newcommand{\nec}[1]{\marginpar{\hspace{-10pt}\color{green}\tiny\ttfamily Nec: #1}}
\newcommand{\davv}[1]{\marginpar{\hspace{-10pt}\color{purple}\tiny\ttfamily DA: #1}}
\newcommand{\laura}[1]{{\color{cyan} #1}}
\newcommand{\laurac}[1]{\marginpar{\hspace{-10pt}\color{cyan}\tiny\ttfamily Laura: #1}}
\newcommand{\parentchild}[4]{}
\else
\newcommand{\red}[1]{}

\definecolor{revColor}{RGB}{0, 0, 0}
\newcommand{\zheinline}[1]{}
\newcommand{\zhe}[1]{}
\newcommand{\ysatt}[1]{}
\newcommand{\nec}[1]{}
\newcommand{\davv}[1]{}
\newcommand{\laura}[1]{}
\newcommand{\laurac}[1]{}

\newcommand{\parentchild}[4]{}
\fi

\newcommand{\numSys}{s}
\newcommand{\dimSt}{n}
\newcommand{\dimInput}{p}

\newcommand{\numCls}{r}

\usepackage{graphicx}

\begin{document}

\sptitle{Article Category}

\title{Mode Reduction for Markov Jump Systems} 

\editor{This paper was recommended by Associate Editor F. A. Author.}

\author{Zhe Du}

\author{\phantom{}Laura Balzano}

\author{Necmiye Ozay}

\affil{Department of Electrical Engineering and Computer Science, University of Michigan Ann Arbor, MI 48109, USA}

\corresp{CORRESPONDING AUTHOR: Zhe Du (e-mail: \href{mailto:zhedu@umich.edu}{zhedu@umich.edu})}
\authornote{N. Ozay and Z. Du were supported by ONR grants N00014-18-1-2501 and N00014-21-1-2431, L. Balzano and Z. Du were supported by AFOSR YIP award FA9550-19-1-0026 and NSF CAREER award CCF-1845076, and L. Balzano was supported by NSF BIGDATA award IIS-1838179.}

\markboth{MODE REDUCTION FOR MARKOV JUMP SYSTEMS}{DU {\itshape ET AL}.}

\begin{abstract}
Switched systems are capable of modeling processes with underlying dynamics that may change abruptly over time. To achieve accurate modeling in practice, one may need a large number of modes, but this may in turn increase the model complexity drastically. Existing work on reducing system complexity mainly considers state space reduction, whereas reducing the number of modes is less studied. In this work, we consider Markov jump linear systems (MJSs), a special class of switched systems where the active mode switches according to a Markov chain, and several issues associated with its mode complexity. Specifically, inspired by clustering techniques from unsupervised learning, we are able to construct a reduced MJS with fewer modes that approximates the original MJS well under various metrics. Furthermore, both theoretically and empirically, we show how one can use the reduced MJS to analyze stability and design controllers with significant reduction in computational cost while achieving guaranteed accuracy.
\end{abstract}


\maketitle

\section{Introduction}
As the control and machine learning communities build tools to model ever more complex dynamical systems, it will become increasingly important to identify redundant aspects of a model and remove them using various unsupervised learning techniques. State dimensionality reduction has long been common in control systems, using principal component analysis and similar techniques. In this paper we consider the setting where switched systems have redundant modes, and we apply clustering -- another fundamental unsupervised learning technique -- to remove redundancies. 

Switched systems generalize time-invariant systems by allowing the dynamics to switch over time. They have been used to model abrupt changes in the environment (e.g. weather and road surfaces), controlled plants (e.g. functioning statuses of different components), disturbances, or even control goals (e.g. cost functions in the optimal control). 
Switched system models are used in a variety of applications including controlling a Mars rover exploring an unknown heterogeneous terrain, solar power generation, investments in financial markets, and communications with packet losses \cite{blackmore2005combining,cajueiro2002stochastic,loparo1990probabilistic,svensson2008optimal,ugrinovskii2005decentralized, sinopoli2005lqg, truong2021analysis}.
However, these benefits are accompanied by new complexity challenges: the number of modes that is needed to model systems accurately and thoroughly may grow undesirably large. For example, for controlled plants composed of multiple components, if we model each combination of health statuses, e.g. working and faulty, of all components as a mode, then the number of modes grows exponentially with the number of components. Given this rate, there can be an huge amount of modes even with a moderate number of components. Given a system with this many modes, analysis can become computationally intractable. For example, in finite horizon linear quadratic regulator (LQR) problems with horizon $T$, the total number of controllers to be computed is $\numSys^T$ where $\numSys$ denotes the number of modes.
This lack of scalability calls for systematic and theoretically guaranteed ways to reduce the number of modes.

Existing work on (switched) system reduction mainly focuses on reducing the state dimension \cite{zhang2003h} or constructing finite abstractions for the continuous state space \cite{zamani2014approximately}. Reducing the mode complexity, however, is still mainly an uncharted territory. 

In this work, we study how one can perform mode reduction for Markov jump linear systems (MJS), a class of switched systems with individual modes given by state space linear time-invariant (LTI)  models and the mode switching governed by Markov chains. Our main contributions are the following:
\begin{itemize}
	\item We propose a clustering-based method that takes mode dynamics as features and use the estimated clusters to construct a mode-reduced MJS.
	\item The reduced MJS is proved to well approximate the original MJS under several approximation metrics.
	\item We show the reduced MJS can be used as a surrogate for the original MJS to analyze stability and design controllers with guaranteed performance and significant reduction in computational cost.
	\item Compared to the preliminary conference version \cite{zhe2021clustering}, this paper establishes stronger approximation guarantees (Section \ref{sec_MJSReduction_approxMetrics}), provides novel stability analysis (Section \ref{sec_MJSReduction_stabAnalysis}), includes all the proofs that were omitted in \cite{zhe2021clustering}, and includes a more thorough discussion and survey of related work, including a table that lays out relationships among existing work and our work.
\end{itemize}
Our work adds a new dimension, i.e., reduction of modes, to the research of switched system reduction.
This framework can be generalized to other problems such as robust and optimal control, invariance analysis, partially observed systems, etc.
Other than constructing and analyzing the reduced MJS, the technical tools we develop in this work regarding perturbations can be applied to cases when there are model mismatches, e.g. system estimation errors incurred when dynamics are learned in identification or data-driven adaptive control as in \cite{sattar2021identification}. 

\begin{table*}[ht]
\centering
\caption{Related Work on Reduction for Stochastic Switched Systems}\label{Table_MJSReduction_RelatedWork}
{\renewcommand{\arraystretch}{1.2}
\begin{tabular}{||l|c|c|c|c|c||} 
		\hline
		\phantom{\emph{}}
		\makecell{\textbf{Reference}} & \makecell{\textbf{Model}} & \makecell{\textbf{Reduction} \\ \textbf{Target}} & \makecell{\textbf{Exact} \\ \textbf{Bisimulation} \\ \textbf{Condition}} & \makecell{\textbf{Reduction} \\ \textbf{Method}} & \makecell{\textbf{Approximation}\\ \textbf{Metric}}  \\ \hline \hline
		\cite{desharnais2004metrics} & Controlled Markov Process & \multirow{7}{*}{\makecell{State Space \\ Cardinality}} & \multirow{5}{*}{Yes} &  \multirow{2}{*}{\textbf{N.A.}} & Formula Metric\\
		\cline{1-2} \cline{6-6}
		\cite{tkachev2014approximation, bian2017relationship} & \multirow{2}{*}{\makecell{Labelled Markov Chains \\ (autonomous)}} & &  &  & \multirow{2}{*}{Trajectory}\\
		\cline{1-1} \cline{5-5}
		\cite{lun2018approximate} &  & &  &  \multirow{2}{*}{Clustering} & \\
		\cline{1-2} \cline{6-6}
		\cite{zhang2018state,du2019mode,bittracher2021probabilistic} & Markov Chains &  & & & \textbf{N.A.}\\
		\cline{1-2} \cline{5-6}
		\cite{zamani2014approximately} & Switching Stochastic Sys. & & & \multirow{3}{*}{\makecell{State  Space \\ Discretization}}  & Trajectory \\
		\cline{1-2} \cline{4-4} \cline{6-6}
		\cite{abate2011approximate} & \multirow{2}{*}{\makecell{Stochastic Hybrid \\ System (autonomous)} }  & & \multirow{8}{*}{\textbf{N.A.}} & & Transition Kernel \\
		\cline{1-1} \cline{6-6}
		\cite{soudjani2011adaptive} & & & & & Invariance Probability \\
		\cline{1-3} \cline{5-6}
		\cite{julius2006approximate,julius2009approximations,zamani2016approximations} & \makecell{Jump Linear \\ Stochastic System} & \multirow{4}{*}{\makecell{State  Space \\ Dimension \\ (Order)}} & & \makecell{(Bi)simulation \\ Function}  & Trajectory \\
		\cline{1-2} \cline{5-6}
		\cite{zhang2003h,shen2019model} & \multirow{3}{*}{MJS} & & & \multirow{1}{*}{$\Hcal_\infty$ Reduction}  & \multirow{2}{*}{$\Hcal_\infty$ norm} \\	
		\cline{1-1} \cline{5-5} 
		\cite{kotsalis2010balanced} &  &  & & Balanced Truncation &  \\
		\cline{1-1} \cline{5-6} 
		\cite{sun2016model} &  &  & & $\Hcal_2$ Reduction  & $\Hcal_2$ norm \\
		\hline
		This Work & MJS & Modes & Yes & Clustering  & Trajectory \\
		\hline
\end{tabular}}
\end{table*}

\section{Related Work}
Depending on the problems of interest and methodologies, the work on reduction for stochastic (switched) systems can be roughly divided into three categories: bisimulation, symbolic abstraction, and order reduction.

\noindent \textbf{Bisimulation:} 
To evaluate the equivalency between two stochastic switched systems, notions of (approximate) probabilistic bisimulation are proposed in \cite{larsen1991bisimulation, desharnais2002bisimulation, desharnais2004metrics}. Approximation metrics \cite{abate2013approximation} from different perspectives are developed to compare two systems, e.g. one(multi)-step transition kernels \cite{abate2011approximate} and trajectories \cite{girard2007approximation, tkachev2014approximation, julius2009approximations}. 
Based on the approximate bisimulation notion in \cite{bian2017relationship}, a technique for reducing the state space of labeled Markov chains through state aggregation is proposed in \cite{lun2018approximate}.
Unlike existing work that typically defines the notions of (approximate) bisimulation on the state space, we provide an algorithm that constructs a reduced system by aggregating the mode space, which provably approximates the original one. Our work shares the idea of aggregation of Markov chains with \cite{lun2018approximate}, but we also seek to recover the best aggregation partition which is otherwise assumed as prior knowledge in \cite{lun2018approximate}.

\noindent \textbf{Symbolic Abstraction:} 
Given a system with continuous state space, abstraction \cite{alur2000discrete} considers discretizing the state space and then constructing a finite state symbolic model, which can be used as a surrogate for model verification \cite{clarke2018model, kurshan2014computer} or controller synthesis \cite{maler1995synthesis}. 
The work on abstraction for stochastic hybrid systems starts with the autonomous cases. Under uniform discretization, \cite{abate2010approximate} and \cite{abate2011approximate} provide approximation guarantees that depend on the discretization width. An adaptive partition scheme is proposed in \cite{soudjani2011adaptive}, which mitigates the curse of dimensionality suffered by uniform sampling. Since the systems under consideration are autonomous, these work mainly serves verification purposes, but falls short toward controller synthesis goals. \cite{zamani2014approximately} addressed this by allowing inputs in the systems. The idea of partitioning the continuous state space is similar to our work except that our partition is performed on the mode space, i.e., the discrete state space in hybrid systems, which provides a new yet closely related dimension to the existing abstraction work.

\noindent \textbf{Order Reduction:} 
Another important line of research on system reduction is order reduction \cite{gugercin2004survey}, where one seeks to reduce the dimension of the state space to satisfy certain criteria. With the help of linear matrix inequalities (LMIs), various methods have been applied for MJS, including $\Hcal_\infty$ reduction \cite{zhang2003h}, balanced truncation \cite{kotsalis2010balanced}, and $\Hcal_2$ reduction \cite{sun2016model}, etc. Order reduction is also applied to more complex models with time-varying transition probabilities \cite{shen2019model}.

The reduction of Markov chains, a class of simplified yet fundamental stochastic switched models, has also attracted the learning and statistics communities. Several notions of lumpability are proposed in \cite{buchholz1994exact}, which coincide with the notion of bisimulation in \cite{larsen1991bisimulation}. Lumpability allows one to reduce the original Markov chain to a smaller scale yet equivalent Markov chain by lumping the Markovian states. 
Similar research focusing on the equivalence metrics and bounding the difference of transition kernel, can be found in \cite{hoffmann2009bounding, schulman2001coarse, gaveau2005dynamical} under the name of coarse graining.
Compared with the bisimulation work for general stochastic systems, which is mostly conceptual, the restriction to Markov chains allows for practical ``low-rank + clustering'' methods \cite{meila2001random} to uncover the lumpability structure. \cite{zhang2018state} considers the case when the Markov matrix is estimated from a trajectory, and the approximate lumpability case are studied in \cite{du2019mode,bittracher2021probabilistic}. Furthermore \cite{du2019mode} studies the reduction of Markov chains that are embedded in switched autoregressive exogenous models, but the overall dynamical models are not reduced. Based on the ideas in \cite{du2019mode}, our work further extends the reduction to the overall MJS.

A comprehensive comparison of the related work together with our work is listed in Table \ref{Table_MJSReduction_RelatedWork}. 
The entry ``exact bisimulation condition" tells whether ideal case sufficient conditions are provided under which a system can be reduced without introducing any model inaccuracy, i.e., they are bisimilar. In practice, when the reduced system is constructed, these principled conditions can help gain more insight into the original system.
In practice, these system reduction methods developed under different perspectives can be combined to achieve overall better performance. For example, for the continuous state space, one can apply order reduction followed by finite abstraction (the former can help remove the curse of dimensionality for the latter), and meanwhile our work can further help reduce the discrete mode space.

This work is organized as follows: we present the preliminaries and mode reduction problem setup in Section \ref{sec_MJSReduction_prelim}; an clustering-based reduction approach is proposed in Section \ref{sec_MJSReduction_AlgandTheory}; in Section \ref{sec_MJSReduction_approxMetrics}, we discuss how the reduced MJS approximates the original one under setups; Section \ref{sec_MJSReduction_stabAnalysis} and \ref{sec_MJSReduction_LQRControl} respectively show that one can use the reduced MJS as a surrogate to evaluate stability and design LQR controllers for the original MJS; simulation experiments are presented in Section \ref{sec_MJSReduction_Experiment}. 

\section{Preliminaries and Problem Setup} \label{sec_MJSReduction_prelim}
For a matrix $\vE$, 
$\vE(i,:)$ denotes the $i$th row of $\vE$, and $\vE(i,j{:}k)$ denotes the $i$th row preserving only the $j$th to $k$th columns. For any index set $A$,  $\vE(i,A)$ denotes the $i$th row preserving columns given by $A$. 
Let $\sigma_i(\vE)$ ($\lambda_i(\vE)$) denote its $i$th largest singular (eigen) value.
For any $s \in \mathbb{N}$, we let $[s]:=\curlybrackets{1,2,\dots,s}$.
We say $\Omega_{1:\numCls}:=\curlybrackets{\Omega_1, \dots, \Omega_\numCls}$ is a $\numCls$-cluster partition of $[\numSys]$ if $\bigcup_{i=1}^\numCls \Omega_i = [\numSys]$, $\Omega_i \bigcap \Omega_j = \phi$ for any $i \neq j$, and $\Omega_i \neq \phi$. We let $\Omega_{(i)}$ denote the cluster with $i$th largest cardinality. For a sequence of variables $X_0, X_1, \dots, X_N$, let $X_{0:N}:=\curlybrackets{X_i}_{i=0}^N$. 
Notation $\otimes$ denotes the Kronecker product.
Notation $\vI_n$ denotes the $n$-dimensional identity matrix, and $\onevec_{n}$ denotes the $n$-dimensional all-ones vector.

\subsection{Preliminaries} \label{subsec_MJSReduction_prelim}
In this work, we consider Markov jump systems (MJSs) with dynamics given by
\begin{equation}\label{eq_MJSReduction_MJS}
	\Sigma {:=} \curlybrackets{\vx_{t+1} {=} \vA_{\omega_t}\vx_t + \vB_{\omega_t}\vu_t, \ \omega_t \sim \text{MarkovChain}(\vT)}
\end{equation}
where $\vx_t \in \dm{\dimSt}{1}$ and $\vu_t \in \dm{\dimInput}{1}$ denote the state and input at time $t$. 
There are $\numSys$ modes parameterized by $\curlybrackets{\vA_i, \vB_i}_{i=1}^\numSys$ where $\vA_i \in \dm{\dimSt}{\dimSt}$ and $\vB_i \in \dm{\dimSt}{\dimInput}$ are state and input matrices for mode $i$. The active mode at time $t$ is indexed by $\omega_t \in [\numSys]$, and the mode switching sequence $\omega_{0:t}$ follows a Markov chain with Markov matrix $\vT \in \dm{\numSys}{\numSys}$, i.e., $\prob(\omega_{t+1}=j \mid \omega_t=i) = \vT(i,j)$. We assume the Markov chain $\vT$ is ergodic. By properties of ergodicity, $\vT$ has a unique stationary distribution $\vpi \in \dm{\numSys}{1}$, and we let $\pi_{\max}$ and $\pi_{\min}$ denote the largest and smallest element in $\vpi$. In the remainder of the paper, we use $\Sigma:= \text{MJS}(\vA_{1:\numSys}, \vB_{1:\numSys}, \vT)$ to denote the groundtruth MJS in \eqref{eq_MJSReduction_MJS} that we want to study, and similarly use notation MJS($\cdot, \cdot, \cdot$) to parameterize any MJS with expressions given in \eqref{eq_MJSReduction_MJS}. We introduce the following two special types of Markov chains, which are closely tied to the main focus of this work.

\begin{definition}[Lumpability and Aggregatability \cite{buchholz1994exact}]\label{def_MJSReduction_lumpable}
	Markov matrix $\vT \in \dm{\numSys}{\numSys}$ is \emph{lumpable} w.r.t. partition $\Omega_{1:\numCls}$ on $[\numSys]$ if for any $k,l \in [\numCls]$, and $i,i' \in \Omega_k$, we have
	$
	\sum_{j\in \Omega_l} \vT(i,j) = \sum_{j\in \Omega_l} \vT(i',j).
	$
	As a special case, it is further \emph{aggregatable} if 
	$\vT(i,:)=\vT(i',:)$.
\end{definition}
Lumpability of a Markov chain coincides with the definition of probabilistic bisimulation in \cite{desharnais2002bisimulation}, which describes an equivalence relation on $[\numSys]$, i.e., two members are equivalent if they belong to the same cluster. For a Markov chain $\vT$ that is lumpable with respect to partition $\Omega_{1:\numSys}$, we use $\zeta_t \in [\numCls]$ to index the active cluster at time $t$, i.e., $\zeta_t = k$ if and only if $\omega_t \in \Omega_k$, and use $\zeta_{0:t}$ to denote the active cluster sequence. 

\begin{figure*}
	\centering
	\includegraphics[width=0.7\linewidth]{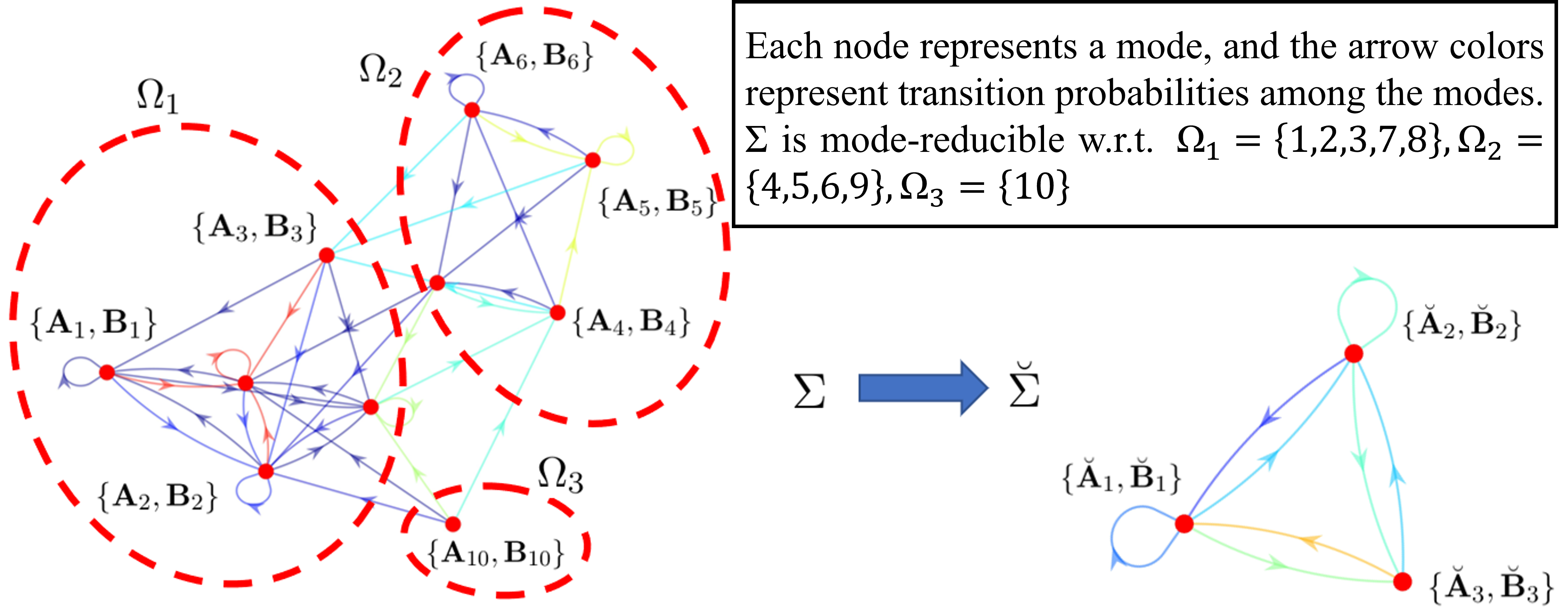}
	\caption{Illustration of reduction under mode-reducibility condition. 
	}\label{fig_MJSReduction_ReductionIllustration}
\end{figure*}

\subsection{Problem Formulation}
With the notions of Markov chain lumpability and aggregatability, in this work, we consider reducing the number of modes for $\Sigma$ under the following two problem settings.
\begin{problem}[Lumpable Case]\label{problem_lumpableCase}
    Assume the dynamics of $\Sigma {=}
	\textup{MJS}(\vA_{1:\numSys}, \vB_{1:\numSys}, \vT)$ is known. 
	Suppose there exists a hidden partition $\Omega_{1:\numCls}$ on $[\numSys]$ and $\epsilon_\vA, \epsilon_\vB, \epsilon_\vT \geq 0$ such that
	\begin{equation} \label{eq_MJSReduction_approxDyn}
	\begin{split}
	    \sum_{k {\in} [\numCls]} \sum_{i, i' {\in} \Omega_k} \norm{\vA_i - \vA_{i'}}_\fro \leq \epsilon_\vA, \\ \sum_{k {\in} [\numCls]} \sum_{i, i' {\in} \Omega_k} \norm{\vB_i - \vB_{i'}}_\fro \leq \epsilon_\vB, \\
	\end{split}    
	\end{equation}
	\begin{equation}
	    \sum_{k,l {\in} [\numCls]} \sum_{i, i' {\in} \Omega_k}  \Big| \sum_{j \in \Omega_l} \vT(i,j) - \sum_{j \in \Omega_l} \vT(i',j) \Big| \leq \epsilon_\vT. \label{eq_MJSReduction_approxLump} 
	\end{equation}
	Then, given $\curlybrackets{\vA_i, \vB_i}_{i=1}^\numSys, \vT$ and $\numCls$, we seek to estimate the partition $\Omega_{1:\numCls}$ by clustering the modes, construct a reduced MJS for $\Sigma$, and provide guarantees on the behavior difference incurred by the reduction.	
\end{problem}

Throughout this work, we refer to $\epsilon_\vA, \epsilon_\vB$, and $\epsilon_\vT$ as \emph{perturbations}. 
When $\epsilon_\vT = 0$, by definition $\vT$ is lumpable. Thus, in condition \eqref{eq_MJSReduction_approxLump}, one can view $\vT$ as \emph{approximately} lumpable. 
One can think of $\curlybrackets{\omega_t, \vx_t} \in [\numSys] \times \dm{\dimSt}{1}$ as a hybrid state \cite{abate2011approximate}. Then, condition \eqref{eq_MJSReduction_approxLump} guarantees the existence of an approximate equivalence relation in the discrete domain $[\numSys]$, while condition \eqref{eq_MJSReduction_approxDyn} guarantees this in the continuous domain $\dm{\dimSt}{1}$. 
More discussions for the special case when $\epsilon_\vA =  \epsilon_\vB = \epsilon_\vT = 0$ follow in the next subsection.

For the aggregatable case, we separately formulate a similar problem in Problem \ref{problem_aggregatableCase}. 
\begin{problem}[Aggregatable Case] \label{problem_aggregatableCase}
	In Problem \ref{problem_lumpableCase}, replace \eqref{eq_MJSReduction_approxLump} with
	$
	\sum_{k {\in} [\numCls]} \sum_{i, i' {\in} \Omega_k} \norm{\vT(i,:)^\T - \vT(i',:)^\T}_1 \leq \epsilon_\vT.
	$	
\end{problem}
\ref{problem_aggregatableCase} differs from \ref{problem_lumpableCase} in terms of $\epsilon_\vT$, which quantifies the violation of lumpability and aggregatability respectively for the matrix $\vT$. Consider a 3-state Markov matrix $\vT$ with rows $\vT(1,:) {=} [0.2, 0.4, 0.4]$, $\vT(2,:) {=} [0.7, 0.1, 0.2]$, and $\vT(3,:) {=} [0.7, 0, 0.3]$. Then for the partition $\Omega_1 = \{1\}, \Omega_2 = \{2,3\}$, we obtain $\epsilon_\vT = 0$ in \ref{problem_lumpableCase} but $\epsilon_\vT = 0.4$ in \ref{problem_aggregatableCase}. In other words, $\vT$ is exactly lumpable but non-aggregatable with a violation level of $0.4$. In \ref{problem_aggregatableCase}, $\epsilon_\vT = 0$ only when $\vT(i,:) = \vT(i',:)$ for all $i,i' \in \Omega_k$, i.e. the rows are equal. On the other hand, in \ref{problem_lumpableCase}, $\epsilon_\vT = 0$ is possible even if no row equalities exist. Hence, being $\epsilon_\vT$-aggregatable in \ref{problem_aggregatableCase} is a stronger assumption than being $\epsilon_\vT$-lumpable in \ref{problem_lumpableCase}. As a result, in Section \ref{sec_MJSReduction_AlgandTheory}, the clustering guarantee for \ref{problem_aggregatableCase} is stronger and more interpretable than that of \ref{problem_lumpableCase}.


\subsection{Equivalency between MJSs}\label{subsec_MJSReduction_ModereducibleCond}
To compare the original and mode-reduced MJSs as mentioned in \ref{problem_lumpableCase}, we need a notion of equivalency between two MJSs with different numbers of modes. This is provided below via a surjection from modes of the larger MJS to the smaller one, which extends the bijection idea in \cite{julius2009approximations, zhang2003h} that can only compare two MJSs with equal amounts of modes.

\begin{definition}[Equivalency between MJSs]\label{def_MJSReduction_EquivalenceMJS}
	Consider two \textup{MJSs} $\Sigma_1$ and $\Sigma_2$ with the same state and input dimensions $\dimSt$, $\dimInput$, but different number of modes $s_1$ and $s_2$ respectively. WLOG, assume $s_1 > s_2$. Let $\curlybrackets{\vx_t^{(1)}, \vu_t^{(1)}, \omega_t^{(1)}}$ and $\curlybrackets{\vx_t^{(2)}, \vu_t^{(2)}, \omega_t^{(2)}}$ denote their respective state, input, and mode index. $\Sigma_1$ and $\Sigma_2$ are equivalent if there exists a partition $\Omega_{1:s_2}$ on $[s_1]$ such that $\Sigma_1$ and $\Sigma_2$ have the same transition kernels, i.e. for any time $t$, any mode $k, k' \in [s_2]$, any $\vx, \vx' \in \dm{\dimSt}{1}$, and any $\vu \in \dm{\dimInput}{1}$
	\begin{multline}
		\prob \big(\omega_{t+1}^{(1)} {\in} \Omega_{k'}, \vx_{t+1}^{(1)} {=} \vx' \mid \ \omega_t^{(1)} {\in} \Omega_{k}, \vx_t^{(1)} {=} \vx, \vu_t^{(1)} {=} \vu \big) \\
		=
		\prob \big(\omega_{t+1}^{(2)} {=} k', \vx_{t+1}^{(2)} {=} \vx'  \mid \ \omega_t^{(2)} {=} k, \vx_t^{(2)} {=} \vx, \vu_t^{(2)} {=} \vu \big).
	\end{multline}
\end{definition}
The trivial perturbation-free case, i.e., $\epsilon_\vA, \epsilon_\vB, \epsilon_\vT {=} 0$, provides a sufficient condition that guarantees that an MJS can be reduced to a smaller MJS with equivalency between them.

\begin{definition}[Mode-reducibility Condition]\label{def_MJSReduction_modeReducible}
    If in \ref{problem_lumpableCase}, $\epsilon_\vA, \epsilon_\vB, $
    $\epsilon_\vT {=} 0$, we say
    $\Sigma$ is \textit{mode-reducible} with respect to $\Omega_{1:\numCls}$.
\end{definition}

If this condition holds for $\Sigma$, we can construct a mode-reduced MJS $\breve{\Sigma} {:=} \textup{MJS}(\breve{\vA}_{1:\numCls}, \breve{\vB}_{1:\numCls}, \breve{\vT})$ such that for any $k,l \in [\numCls]$, any $i\in \Omega_k$, $\breve{\vA}_k = \vA_i$, $\breve{\vB}_k = \vB_i$, and $\breve{\vT} \in \dm{\numCls}{\numCls}$ with $\breve{\vT}(k,l) = \sum_{j \in \Omega_l} \vT(i, j)$, which is illustrated in Fig. \ref{fig_MJSReduction_ReductionIllustration}.
Let $\curlybrackets{\breve{\vx}_t, \breve{\vu}_t, \breve{\omega}_t}$ denote the state, input, and mode index for the reduced $\breve{\Sigma}$. Then, the following fact shows that $\breve{\Sigma}$ and $\Sigma$ are equivalent according to Definition \ref{def_MJSReduction_EquivalenceMJS}. 

\begin{fact}\label{fact_MJSReduction_existReducedSys}
	Suppose $\Sigma$ is mode-reducible and $\breve{\Sigma}$ is constructed as above.
	Consider the case when $\Sigma$ and $\hat{\Sigma}$ have (i) initial mode distributions satisfy $\prob(\omega_0 \in \Omega_{k}) = \prob(\breve{\omega}_0 = k)$ for all $k \in [\numCls]$, (ii)
	the same initial states $(\vx_0 = \breve{\vx}_0)$, and (iii) the same input sequences $(\vu_{0:t-1} = \breve{\vu}_{0:t-1})$. Then, these two MJSs have the same mode and state transition kernels, i.e. $\prob(\omega_t \in \Omega_k, \vx_{t} {=} \vx ) = \prob(\breve{\omega}_t{=}k, \breve{\vx}_{t} {=} \vx)$ for all $t$, all $k \in [\numCls]$ and $\vx \in \dm{\dimSt}{1}$. Particularly, there exists a special type of reduced $\breve{\Sigma}$ such that the modes are synchronized: for all $t$, $\breve{\omega}_t=\zeta_t$. In this case, $\breve{\vx}_t = \vx_t$ for all $t$.
\end{fact}
Fact \ref{fact_MJSReduction_existReducedSys} first shows the equivalency between $\Sigma$ and $\breve{\Sigma}$ in terms of the transition kernels, which is then extended to trajectory realizations if certain synchrony exists between $\zeta_{0:t}$ and $\breve{\omega}_{0:t}$. The condition $\breve{\omega}_t=\zeta_t$ in Fact \ref{fact_MJSReduction_existReducedSys} essentially establishes a coupling between the Markov chains $\omega_{0:t}$ and $\breve{\omega}_{0:t}$ such that $\prob(\omega_t \in \Omega_k, \breve{\omega}_t=k) = \prob(\omega_t \in \Omega_k) = \prob(\breve{\omega}_t=k)$. Establishing coupling between the stochastic systems usually allows for stronger equivalency and approximation result. Similar coupling scheme is implicitly used in \cite{julius2009approximations, zhang2003h}; an optimal coupling by minimizing Wasserstein distance is discussed in \cite{tkachev2014approximation}; and a weaker coupling using the idea of HMM is discussed in \cite{shen2019model}.


In Definition \ref{def_MJSReduction_EquivalenceMJS} and Fact \ref{fact_MJSReduction_existReducedSys}, one can view  $\curlybrackets{\omega_t, \vx_t} \in [\numSys] \times \dm{\dimSt}{1}$ as a hybrid state \cite{abate2011approximate}. $\vT$ being lumpable guarantees the existence of an equivalence relation in the discrete domain $[\numSys]$ as in Definition \ref{def_MJSReduction_lumpable}, while state/input matrices being the same guarantees this in the continuous domain $\dm{\dimSt}{1}$. 

\section{Clustering-based Mode Reduction for MJS}\label{sec_MJSReduction_AlgandTheory}
In this section, we first propose Algorithm \ref{Alg_MJSReduction} to estimate the latent partition $\Omega_{1:\numCls}$ and construct the reduced MJS for Problem \ref{problem_lumpableCase} and \ref{problem_aggregatableCase}, and then provide its theoretical guarantees for partition estimation in Section \ref{sec_MJSReduction_AlgandTheory}-\ref{subsec_MJSReduction_theoryClustering}.

{
\SetAlgoNoLine%
\begin{algorithm}[ht]
	\LinesNumbered
	\KwIn{$\vA_{1:\numSys}, \vB_{1:\numSys}, \vT, \vpi$, $\numCls$, and non-negative tuning weights $\alpha_\vA, \alpha_\vB, \alpha_\vT$ that sum to $1$}
	Construct feature matrix $\vPhi$: $\forall i \in [s]$,\\
	\uCase{Problem \ref{problem_aggregatableCase}\label{algline_MJSReduction_11}}
	{
		$\vPhi(i,:) {=} [\alpha_\vA \vek(\vA_i)^\T, \alpha_\vB \vek(\vB_i)^\T, \alpha_\vT \vT(i,:)]$.  \label{algline_MJSReduction_7}\\
	}
	\uCase{Problem \ref{problem_lumpableCase}} 
	{
		$\vH = \diag(\vpi)^\frac{1}{2} \vT \diag(\vpi)^{\invv{\frac{1}{2}}}$ \label{algline_MJSReduction_12}\\
		$\vW_\numCls \leftarrow \text{ top } \numCls \text{ left singular vectors of } \vH$ \\			
		$\vS_\numCls = \diag(\vpi)^{\invv{\frac{1}{2}}} \vW_\numCls$ \label{algline_MJSReduction_1}\\
		$\vPhi(i,:) = [\alpha_\vA \vek(\vA_i)^\T, \alpha_\vB \vek(\vB_i)^\T, \alpha_\vT \vS_\numCls(i,:)]$ \label{algline_MJSReduction_13}\\
	}
	$\vU_{\numCls} \leftarrow \text{ top } \numCls \text{ left singular vectors of }\vPhi$ \\
	Solve k-means problem:
	$
	\hat{\Omega}_{1:\numCls}, \hat{\vc}_{1:\numCls} = 
	\arg \min_{\hat{\Omega}_{1:\numCls}, \hat{\vc}_{1:\numCls}}
	\sum_{k \in [\numCls]} \sum_{i \in \hat{\Omega}_k} 
	\norm{\vU_{\numCls}(i,:) - \hat{\vc}_k}^2
	$\\	
	\label{algline_MJSReduction_10}
	Construct $\hat{\Sigma}$, $\forall k,l \in [r]$
	$
	\hat{\vA}_k = \frac{1}{|\hat{\Omega}_{k}|} \sum_{i\in \hat{\Omega}_{k}} \vA_i, \quad
	\hat{\vB}_k = \frac{1}{|\hat{\Omega}_{k}|} \sum_{i\in \hat{\Omega}_{k}} \vB_i, $ $\quad
	\hat{\vT}(k,l) = \frac{1}{|\hat{\Omega}_{k}|} \sum_{i\in \hat{\Omega}_{k}, j \in \hat{\Omega}_{l}} \vT(i, j)
	$\\	\label{algline_MJSReduction_20}
	\KwOut{$\hat{\Sigma}: \textup{MJS}(\hat{\vA}_{1:\numCls}, \hat{\vB}_{1:\numCls}, \hat{\vT})$}
	\caption{System Reduction for MJS} \label{Alg_MJSReduction}
\end{algorithm}}

We treat the estimation of partition $\Omega_{1:\numCls}$ essentially as a mode clustering problem with the dynamics matrices $\vA_i$, $\vB_i$ and transition distribution $\vT(i,:)$ serving as features for mode $i$. In Algorithm \ref{Alg_MJSReduction}, we first construct the feature matrix $\vPhi$ from Line \ref{algline_MJSReduction_11} to Line \ref{algline_MJSReduction_13}, with $\vPhi(i,:)$ denoting the features of mode $i$. For the aggregatable case in Problem \ref{problem_aggregatableCase}, we simply stack the vectorized $\vA_i$, $\vB_i$ and $\vT(i,:)$, and use $\alpha_\vA, \alpha_\vB, \alpha_\vT$ to denote their weights respectively. One way to choose these weights is as a normalization, e.g. $\alpha_\vA \propto 1 / \max_i \norm{\vA_i}$, so that these three features would have the same scales. Though in the aggregatable case \ref{problem_aggregatableCase}, similarities among the rows of $\vT$  shed light on the groundtruth partition $\Omega_{1:\numCls}$, this is no longer valid in the lumpable case \ref{problem_lumpableCase} as  two modes belonging to the same cluster can still have different transition probabilities $\vT(i,:)$, even if $\epsilon_\vT = 0$. According to \eqref{eq_MJSReduction_approxLump}, the groundtruth partition $\Omega_{1:\numCls}$ is only embodied in the mode-to-cluster transition probabilities $\sum_{j \in \Omega_l}\vT(i, j)$ constructed using the groundtruth partition itself. This leaves us in a ``chicken-and-egg" dilemma. To deal with this, from Line \ref{algline_MJSReduction_12} to \ref{algline_MJSReduction_13}, we compute the first $\numCls$ left singular vectors $\vW_\numCls$ of matrix $\diag(\vpi)^\frac{1}{2} \vT \diag(\vpi)^{\invv{\frac{1}{2}}}$, and then weight it by $\diag(\vpi)^{\invv{\frac{1}{2}}}$ to obtain matrix $\vS_\numCls \in \dm{\numSys}{\numCls}$, which is used to construct features in $\vPhi$ for the lumpable case \ref{problem_lumpableCase}. We will later justify using $\vS_\numCls$ as features by showing row similarities in $\vS_\numCls$ reflect the partition under certain assumptions. 

With the feature matrix $\vPhi$, to recover the partition, we resort to k-means: in Line \ref{algline_MJSReduction_10}, k-means is applied to the first $\numCls$ left singular vector $\vU_\numCls$ of $\vPhi$. 
The typical algorithm for k-means is Lloyd's algorithm, where the cluster centers and partition membership are updated alternately. 
Based on the solution $\hat{\Omega}_{1:\numCls}$ obtained via k-means, we construct the reduced $\hat{\Sigma}$ by averaging modes within the same estimated cluster. 
A more subtle averaging scheme is through the weights provided in the stationary distribution $\vpi$ which describes the frequency of each mode being active in the long run. $\hat{\Sigma}$ generated by this scheme (or any weighted averaging) would have the same performance guarantees as the uniform averaging, which is provided in Section \ref{sec_MJSReduction_AlgandTheory}-\ref{subsec_MJSReduction_theoryClustering}.


In practice, if we have no good prior knowledge which model of Problem \ref{problem_lumpableCase} and \ref{problem_aggregatableCase} would yield the best system reduction performance, we can first obtain the partitions for both cases and then pick the one yields smaller hindsight perturbations $\epsilon_\vA, \epsilon_\vB, \epsilon_\vT$ in Problem \ref{problem_lumpableCase} and \ref{problem_aggregatableCase}.
When one picks $\alpha_\vA = \alpha_\vB = 0$, i.e only the Markov matrix $\vT$ is used to cluster the modes, then our clustering scheme under the aggregatable case \ref{problem_aggregatableCase} is equivalent to \cite{zhang2018state} which studies clustering for Markov matrices that is estimated from a single trajectory. The lumpable case \ref{problem_lumpableCase}, on the other hand, is based on preliminary analysis in \cite{meila2001random}.

We note that several aspects of this algorithm we have guarantees for do not directly consider the metrics important to this problem; for example averaging dynamics matrices within the same cluster may not yield the dynamics that gives an optimal fit for prediction or controller design. That said, even for this straightforward approach, Section \ref{sec_MJSReduction_approxMetrics} provides several strong approximation guarantees. We are hopeful that future generalizations will be able to build on this theory and further improve the control performance of our mode-reduction approach.



\subsection{Theoretical Guarantees for Clustering}\label{subsec_MJSReduction_theoryClustering}
In this section, we discuss the clustering performance by comparing the estimated partition $\hat{\Omega}_{1:\numCls}$ and the true $\Omega_{1:\numCls}$.
As k-means algorithms are known to have local convergence properties \cite{bottou1994convergence},
we instead assume for the k-means problem in Algorithm \ref{Alg_MJSReduction}, a $(1+\epsilon)$ approximate solution can be obtained, i.e., 
$
\sum_{k \in [\numCls], i \in \hat{\Omega}_k} \norm{\vU_{\numCls}(i,:)  - \hat{\vc}_k}^2
\leq
(1+\epsilon)\min_{\Omega_{1:\numCls}', \vc_{1:\numCls}'} $ $
\sum_{k \in [\numCls], i \in \Omega_k'} 
\norm{\vU_{\numCls}(i,:)  - \vc_k'}^2.
$
Many efficient algorithms have been developed that can provide  $(1+\epsilon)$ approximate solutions. For $\epsilon=1$, a linear time (in terms of $\numCls$ and $\numSys$) algorithm is provided in \cite{gonzalez1985clustering}. For smaller $\epsilon$, \cite{kumar2004simple} proposes a linear time algorithm using random sampling; \cite{song2010fast} gives a polynomial time algorithm with computational complexity independent of $\numSys$.
We later show how $\epsilon$ affects the overall clustering performance. 
	
To evaluate the performance of partition estimation, we define misclustering rate (MR) as 
$
\textup{MR}(\hat{\Omega}_{1:\numCls}) = \min_{h \in \mathcal{H}} \sum_{k \in [\numCls]} \frac{|\curlybrackets{i: i \in \Omega_k, i \notin \hat{\Omega}_{h(k)}}|}{|\Omega_k|},
$
where $\mathcal{H}$ is the set of all bijections from $[\numCls]$ to $[\numCls]$ so that the comparison finds the best cluster label matching.
The error metric $\textup{MR}$ counts the total misclustered modes normalized by the cluster sizes, which implies clustering errors occurring in smaller clusters would yield larger $\textup{MR}$.

We define the following averaged feature matrix $\bar{\vPhi}$ based on the underlying partition $\Omega_{1:\numCls}$: for all $i \in [\numSys]$ (suppose $i \in \Omega_k$ for some $k \in [\numCls]$), 
$
	\bar{\vPhi}(i,:) = \frac{1}{|{\Omega}_{k}|} \sum_{i' \in {\Omega}_{k}} \vPhi(i',:).
$
By construction, there are up to $\numCls$ unique rows in $\bar{\vPhi}$, hence $\rank(\bar{\vPhi}) \leq \numCls$. 
We first present the clustering guarantee for Problem \ref{problem_aggregatableCase}, i.e., the aggregatable case.

\begin{theorem}\label{thrm_MJSReduction_AggreCase}
	Consider Problem \ref{problem_aggregatableCase} and Algorithm \ref{Alg_MJSReduction}. Suppose $\hat{\Omega}_{1:\numCls}$ is a $(1+\epsilon)$ k-means solution. Let $\epsilon_{Agg}^2 := \alpha_\vA^2 \epsilon_\vA^2 + \alpha_\vB^2 \epsilon_\vB^2 + \alpha_\vT^2 \epsilon_\vT^2$. Then, if $\rank(\bar{\vPhi}) = \numCls$ and $\epsilon_{Agg} \leq \frac{\sigma_\numCls(\bar{\vPhi}) \sqrt{|\Omega_{(\numCls)}| + |\Omega_{(1)}|}}{8 \sqrt{ (2+\epsilon)  |\Omega_{(1)}|}}$, we have
	\vspace{-0em} 
	\begin{equation}
			\label{eq_MJSReduction_aggregatableMR}
		\textup{MR}(\hat{\Omega}_{1:\numCls})
		\leq 
		64(2+\epsilon) \sigma_\numCls(\bar{\vPhi})^{\invv{2}}
		\epsilon_{Agg}^2.
	\end{equation}
	Additionally, if $\epsilon_{Agg} \leq \frac{\sigma_\numCls(\bar{\vPhi})}{8 \sqrt{ (2+\epsilon) |\Omega_{(1)}|}}$, then $\textup{MR}(\hat{\Omega}_{1:\numCls}) = 0$.
\end{theorem}
The key term $\epsilon_{Agg}$ measures how modes within the same cluster differ from each other, i.e., inner-cluster distance. On the other hand, the singular value $\sigma_\numCls(\bar{\vPhi})$ measures the differences of modes from different clusters, i.e., inter-cluster distance. This is because when modes belonging to different clusters have similar features, their corresponding rows in the averaged feature matrix $\bar{\vPhi}$ will also be similar, which could give small $\sigma_\numCls(\bar{\vPhi})$. Particularly, if two different clusters share the same features, then $\rank(\bar{\vPhi})=\numCls-1$ and $\sigma_\numCls(\bar{\vPhi})=0$. In the theorem, when the inner-cluster distance is small compared to the inter-cluster distance, the misclustering rate can be bounded by their ratio $\epsilon_{Agg}/\sigma_\numCls(\bar{\vPhi})$.
By definition of misclustering rate, the smallest nonzero value it can take is given by $\frac{1}{|\Omega_{(1)}|}$. Therefore, whenever the upper bound in \eqref{eq_MJSReduction_aggregatableMR} is smaller than $\frac{1}{|\Omega_{(1)}|}$, one can guarantee $\textup{MR}(\hat{\Omega}_{1:\numCls})=0$, which yields the final claim in Theorem \ref{thrm_MJSReduction_AggreCase}.

The clustering guarantee for the lumpable case in Problem \ref{problem_lumpableCase} is more involved than the aggregatable case. We first provide a few more notions and definitions that can help the exposition. We say a Markov matrix $\vT$ is \emph{reversible} if there exists a distribution $\vpi \in \dm{\numSys}{1}$ such that $\vpi(i)\vT(i,j) = \vpi(j)\vT(j,i)$ for all $i,j \in [\numSys]$. This condition translates to $\diag(\vpi) \vT = \vT^\T \diag(\vpi)$ when $\vT$ is ergodic with stationary distribution $\vpi$. For a reversible Markov matrix that is also lumpable, we have the following property.
\begin{lemma}[Appendix A in \cite{meila2001random}]\label{lemma_MJSReduction_lumpableInformativeSpectrum_shorVer}
	For a reversible Markov matrix $\vT$ that is also lumpable with respect to partition $\Omega_{1:\numCls}$, it is diagonalizable with real eigenvalues. Let $\vS \in \dm{\numSys}{\numSys}$ denote an arbitrary eigenvector matrix of $\vT$. Then, there exists an index set $\Acal \subseteq [\numSys]$ with $|\Acal| = \numCls$ such that for all $k\in [\numCls]$, for all $i, i' \in \Omega_k$, we have $\vS(i, \Acal) = \vS(i', \Acal)$.
\end{lemma}

We say $\vT$ in Lemma \ref{lemma_MJSReduction_lumpableInformativeSpectrum_shorVer} has \emph{informative spectrum} if $\Acal = [\numCls]$ and $|\lambda_{\numCls}(\vT)| > |\lambda_{\numCls+1}(\vT)|$, which implies that the $\numCls$ eigenvectors that carry partition information in Lemma \ref{lemma_MJSReduction_lumpableInformativeSpectrum_shorVer} correspond to the $\numCls$ leading eigenvalues. 
For lumpable Markov matrices, we define the $\epsilon_\vT$-neighborhood of $\vT$:
\begin{equation}\label{eq_NeighborhoodT}
\begin{split}
\Lcal(\vT, \Omega_{1:\numCls}, \epsilon_\vT) :=  \Big\{\vT_0{\in} \dm{\numSys}{\numSys} :  \vT_0 \text{ is Markovian}, \hspace{2em} \\ 
\norm{\vT_0 - \vT}_\infty \leq \epsilon_\vT, \norm{\vT_0 - \vT}_\fro \leq \epsilon_\vT,   \hspace{2em}\\  \hspace{-0.5em}
\forall k,l {\in} [\numCls], \forall i {\in} \Omega_k, \  \sum_{j \in \Omega_l} \vT_0(i,j) = \frac{1}{|\Omega_k|} \sum_{\substack{i' \in \Omega_k\\j \in \Omega_l}} \vT(i', j) \Big\}.
\end{split}
\end{equation}
Under the approximate lumpability condition in \eqref{eq_MJSReduction_approxLump}, one can show this neighborhood set is non-empty. Related discussions are provided in Appendix \ref{appendix_lumpable}. To find such a $\vT_0 \in \Lcal(\vT, \Omega_{1:\numCls}, \epsilon_\vT)$, one only needs to solve a feasibility linear programming problem.
Then we provide the clustering guarantee for the lumpable case.

\begin{theorem}\label{thrm_MJSReduction_LumpCase}
	Consider Problem \ref{problem_lumpableCase} and Algorithm \ref{Alg_MJSReduction}. 
	Let 
	$\gamma_1 {:=} \sum_{i=2}^{\numSys} \frac{1}{1 - \lambda_i(\vT)}$, 
	$\gamma_2 {:=} \min \curlybrackets{\sigma_{\numCls}(\vH) - \sigma_{\numCls+1}(\vH), 1}, 
	\gamma_3 {:=}$		
	$ \frac{16 \gamma_1 \sqrt{\numCls \pi_{\max}} \norm{\vT}_\fro}{\gamma_2 \pi_{\min}^2} $, and 
	$\epsilon_{Lmp}^2 {:=} \alpha_\vA^2 \epsilon_\vA^2 + \alpha_\vB^2 \epsilon_\vB^2 + \alpha_\vT^2 \gamma_3^2 \epsilon_\vT^2$.
	Assume there exists an ergodic and reversible $\vT_0 \in \Lcal(\vT, \Omega_{1:\numCls}, \epsilon_\vT)$ with informative spectrum.
	Suppose $\hat{\Omega}_{1:\numCls}$ is a $(1+\epsilon)$ k-means solution. Then, if $\rank(\bar{\vPhi}) {=} \numCls$,
	$\epsilon_\vT {\leq} \frac{\pi_{\min}}{\gamma_1}$,  	
	$\epsilon_{Lmp} {\leq} \hspace{-0.1em} \frac{\sigma_{\numCls}\hspace{-0.1em}(\hspace{-0.1em}\bar{\vPhi}\hspace{-0.1em}) \sqrt{|\Omega_{(\numCls)}| {+} |\Omega_{(1)}|}}{8 \sqrt{\numSys (2+\epsilon) |\Omega_{(1)}| }}$, we have
	\vspace{-0em}   
	\begin{equation}
    	\textup{MR}(\hspace{-0.1em} \hat{\Omega}_{1:\numCls} \hspace{-0.1em})
    	{\leq}
    	64 (2 {+} \epsilon) \sigma_{\numCls}(\bar{\vPhi})^{\invv{2}}  \epsilon_{Lmp}^2.
	\end{equation}
	Additionally, if $\epsilon_{Lmp} \leq \frac{\sigma_{\numCls}(\bar{\vPhi})}{8 \sqrt{(2+\epsilon) |\Omega_{(1)}|}}$, then $\textup{MR}(\hat{\Omega}_{1:\numCls}) = 0$.
\end{theorem}
Theorem \ref{thrm_MJSReduction_LumpCase} for the lumpable case is similar to Theorem \ref{thrm_MJSReduction_AggreCase} for the aggregatable case with an additional $\gamma_3$ term. This is a result of using $\vS_\numCls$ and $\vT$ to construct features in Algorithm \ref{Alg_MJSReduction} for these two cases. $\gamma_3$ describes how much the lumpability perturbation $\epsilon_\vT$ on $\vT$ affects the row equalities of its spectrum-related matrix $\vS_\numCls$ in Lemma \ref{lemma_MJSReduction_lumpableInformativeSpectrum_shorVer}. The assumption on the existence of $\vT_0$ with informative spectrum guarantees (i) the partition $\Omega_{1:\numCls}$ information is carried by the leading eigenvectors of $\vT_0$ as introduced in Lemma \ref{lemma_MJSReduction_lumpableInformativeSpectrum_shorVer}, and (ii) this information can still be preserved in $\vS_\numCls$ as long as $\vT$ is close to $\vT_0$. Because of this, Theorem \ref{thrm_MJSReduction_LumpCase} may not hold for arbitrary lumpable $\vT$, but only those close to Markov matrices with informative spectra. 


\section{Approximation Guarantees} \label{sec_MJSReduction_approxMetrics}
With perturbations $\epsilon_\vA, \epsilon_\vB, \epsilon_\vT$, the reduced $\hat{\Sigma}$ may not be equivalent to the original $\Sigma$ as in Fact \ref{fact_MJSReduction_existReducedSys}. In this case, if certain approximation guarantees can be established, they can be used in verification tasks such as safety \cite{julius2009approximations} and invariance \cite{soudjani2011adaptive} evaluations.
In this section, we show that the reduced system $\hat{\Sigma}$ can be guaranteed to well approximate the original system $\Sigma$ under metrics such as transition kernels (distributions) and trajectory realizations. Particularly, these metrics reach $0$ when $\epsilon_\vA, \epsilon_\vB, \epsilon_\vT=0$, i.e., the mode-reducibility condition in Definition \ref{def_MJSReduction_modeReducible} holds.
We have shown in Theorem \ref{thrm_MJSReduction_AggreCase} and \ref{thrm_MJSReduction_LumpCase} that $\textup{MR}(\hat{\Omega}_{1:\numCls})=0$ when perturbations $\epsilon_\vA, \epsilon_\vB, \epsilon_\vT$ are small. Hence, in this section together with Section \ref{sec_MJSReduction_stabAnalysis} and \ref{sec_MJSReduction_LQRControl}, we assume $\Omega_{1:\numCls} {=} \hat{\Omega}_{1:\numCls}$ for simplicity. In these sections, the theory holds for perturbations $\epsilon_\vA, \epsilon_\vB, \epsilon_\vT$ introduced in either \ref{problem_lumpableCase} or \ref{problem_aggregatableCase}.

We study the approximation in a setup where $\Sigma$ and $\hat{\Sigma}$ start with the same initial condition and are driven by the same input.
\begin{setup}[Initialization-Excitation Setup]\label{setup_InitializationExcitation}
	Systems $\Sigma$ and $\hat{\Sigma}$ have (i) initial mode distributions satisfy $\prob(\omega_0 \in \hat{\Omega}_k) = \prob(\omegahat_0 = k)$ for all $k \in [\numCls]$; (ii) the same initial states, i.e., $\vx_0 = \hat{\vx}_0$, and (iii) the same inputs $\vu_t = \hat{\vu}_t$ for all $t$.
\end{setup}
Note that when $\Sigma$ and $\hat{\Sigma}$ have fixed and shared initial conditions and inputs as in setup \ref{setup_InitializationExcitation}, we can at most evaluate the difference between $\vx_t$ and $\vxhat_t$ in terms of their distributions (or, the transition kernels of $\Sigma$ and $\hat{\Sigma}$).
However, the actual realizations of $\vx_t$ and $\vxhat_t$, i.e., when we only generate a single sample for each, can be very different. This is because $\vx_t$ and $\vxhat_t$ are driven not only by the input excitation, but also the mode switching sequences $\omega_{0:t-1}$ and $\omegahat_{0:t-1}$; thus, $\vx_t$ and $\vxhat_t$ are likely to be far away from each other if the realizations of $\omega_{0:t-1}$ and $\omegahat_{0:t-1}$ are different. On the other hand, if the reduced model $\hat{\Sigma}$ is to be used online to predict the future behavior of $\Sigma$, and if the mode $\omega_t$ can be observed at run-time, we can assume the following and derive stronger relations on the state realization difference $\norm{\vx_t - \vxhat_t}$. 
\begin{setup}[Mode Synchrony]\label{setup_modeSynchrony}
	Mode $\omegahat_t$ of $\hat{\Sigma}$ is synchronous to $\omega_t$ of $\Sigma$, i.e., for all $t$, if $\omega_t \in \hat{\Omega}_k$ then $\omegahat_t=k$.
\end{setup} 
Mode synchrony setup essentially establishes the strongest possible coupling between $\omega_{0:t}$ and $\omegahat_{0:t}$ as discussed in Section \ref{sec_MJSReduction_prelim}-\ref{subsec_MJSReduction_ModereducibleCond}. When the mode sequence $\omegahat_{0:t}$ of $\hat{\Sigma}$ is synchronized with that of $\Sigma$, this amounts to having $\hat{\Sigma}$ being driven by an external switching signal $\omega_{0:t}$. 


In the following, we provide
bounds on how close $\hat{\Sigma}$ is to $\Sigma$ in terms of the following approximation metrics: (i)  under the mean-square stability of $\Sigma$, the difference $\norm{\vx_t - \vxhat_t}$ in trajectories (Theorem \ref{thrm_MJSReduction_stateDistSyncMSS}); (ii) under uniform stability, the difference in trajectories (Theorem \ref{thrm_MJSReduction_bisimUS} \ref{enum_bisimUS_stateDistSync}) and the difference of transition kernels (Theorem \ref{thrm_MJSReduction_bisimUS} \ref{enum_bisimUS_WasDisAutonomousSys}).

\subsection{Result with Mean-square Stability}\label{subsec_MJSReduction_MSS}
Due to the stochasticity of mode switching, the stability for MJS is typically studied in mean-square sense. We say $\Sigma$ is mean-square stable (MSS), if there exists $\vSigma_{\infty}$ such that $\lim_{t \rightarrow \infty}\expctn[\vx_t \vx_t^\T] = \vSigma_{\infty}$. Define the augmented state matrix $\vAcal \in \dm{\numSys \dimSt^2}{\numSys \dimSt^2}$ with its $ij$-th $\dimSt^2 \times \dimSt^2$ block given by $[\vAcal]_{ij}:= \vT(j,i) \cdot \vA_j \otimes \vA_j$, and let $\rho(\vAcal)$ denote the spectral radius of $\vAcal$. Then, for the autonomous case, i.e., $\vu_t = 0$, $\Sigma$ being MSS is equivalent to $\rho(\vAcal)<1$. It is known that MSS does not imply each individual mode is stable, and vice versa \cite{costa2006discrete}.
For any $\rho \geq \rho(\vAcal)$, let us define
$
\tau(\vAcal, \rho):= \sup_{k\in \mathbb{N}} \norm{\vAcal^k}/\rho^k
$.
This term compares the convergence of $\norm{\vAcal^k}$ and $\rho^k$. In the remainder of this paper, we use $\tau$ as a shorthand notation for $\tau(\vAcal, \rho)$, but keep in mind that $\tau$ depends on the choice of the free parameter $\rho$.
It is easy to see (i) $\tau \geq 1$, (ii) $\norm{\vAcal^k} \leq \tau \rho^k$, and (iii) by Gelfand's formula, $\tau$ is finite for any $\rho > \rho(\vAcal)$. 
When $\vAcal$ is diagonalizable with decomposition $\vAcal = \vV \vLambda \vV^\inv$, a generic property of matrices in $\dm{\numSys \dimSt^2}{\numSys \dimSt^2}$, $\tau$ is also finite for $\rho = \rho(\vAcal)$. In this case, we have $\tau \leq \norm{\vV} \norm{\vV^\inv}$.
We let $\Abar := \max_i \norm{\vA_i}$, $\Bbar := \max_i \norm{\vB_i}$.
The following theorem provides an upper bound for $\norm{\vx_t - \vxhat_t}$ under MSS.

\begin{theorem}\label{thrm_MJSReduction_stateDistSyncMSS}
	Consider setup \ref{setup_InitializationExcitation} and \ref{setup_modeSynchrony} where the shared initial state $\vx_0$ and inputs $\vu_{0:t}$ can be arbitrary as long as for all $t$, $\vu_t$ is bounded, i.e., $\norm{\vu_t} \leq \bar{u}$.
	Assume $\Sigma$ is MSS and $\hat{\Omega}_{1:\numCls} = \Omega_{1:\numCls}$ in Algorithm \ref{Alg_MJSReduction}. 
	For any $\rho \in [\rho(\vAcal),1)$ and its corresponding $\tau$, let $\rho_0 := \frac{1+\rho}{2}$.
	For perturbation, assume $\epsilon_\vA \leq \min \curlybrackets{\Abar, \frac{1-\rho}{6\tau \overset{{} }{\Abar} \norm{\vT} }}$ and $\epsilon_\vB \leq \Bbar$.
	Then, $\expctn[\norm{\vx_t - \vxhat_t}] \leq 4 \sqrt{\dimSt \sqrt{\numSys}} \tau \epsilon^{mss}_t$ where
	$	
		\epsilon^{mss}_t:= \rho_0^{\frac{t-1}{2}} \sqrt{t \Abar \norm{\vT} \epsilon_\vA} \norm{\vx_0} + \sqrt{\Bbar}  \bar{u} \big( \frac{\sqrt{\rho_0}}{(1-\sqrt{\rho_0})^2} \sqrt{\Abar \norm{\vT} \epsilon_\vA} + \frac{\sqrt{2}}{1-\sqrt{\rho_0}} \sqrt{\epsilon_\vB} \big).
	$
\end{theorem}
In this theorem, $\epsilon_t^{mss}$ is the key element in the upper bound. In its definition, the first term describes the effect of $\epsilon_\vA$ through initial state $\vx_0$.
Since $\rho<1$ due to $\Sigma$ being MSS, we know $\rho_0<1$, which implies exponential decay. The rest of the terms in $\epsilon_t^{mss}$ characterize the effects of $\epsilon_\vA$ and $\epsilon_\vB$ through the inputs. And if there is no input, the trajectory difference $\norm{\vx_t - \vxhat_t}$ converges to 0 exponentially with $t$. 
The condition  $\epsilon_\vA \leq \frac{1-\rho}{6\sqrt{\numSys} \tau \overset{{} }{\Abar}}$ is used to guarantee perturbation $\epsilon_\vA$ is small such that $\hat{\Sigma}$ is still MSS, as otherwise the difference will grow exponentially, and no meaningful results can be established in this case.
Conditions $\epsilon_\vA < \Abar$ and $\epsilon_\vB \leq \Bbar$ are only used to simplify the expressions, and similar bounds can be established without them.

Fact \ref{fact_MJSReduction_existReducedSys} provides a sanity check for Theorem \ref{thrm_MJSReduction_stateDistSyncMSS}: when $\epsilon_\vA =\epsilon_\vB= 0$, we have $\vx_t = \vxhat_t$.
In the autonomous case, i.e., $\vu_t=0$, as a direct corollary of Theorem \ref{thrm_MJSReduction_stateDistSyncMSS}, we can further obtain a probabilistic bound on the difference over an entire trajectory using Markov inequality: with  probability at least $1-\delta$, $\sum_{t=0}^\infty \norm{\vx_t - \vxhat_t} \leq \frac{4 \sqrt{\dimSt \dimInput} \tau \norm{\vx_0} \sqrt{\Abar \epsilon_\vA}}{\delta(1-\sqrt{\rho_0})^2}$.

\subsection{Results with Uniform Stability} \label{subsec_MJSReduction_US}
MSS in Section \ref{sec_MJSReduction_approxMetrics}-\ref{subsec_MJSReduction_MSS} is a weak notion of stability in that it only requires stability in expectation while still allowing a set of mode switching sequences that result in explosive $\vx_t$, even a set with nonzero probability.
In this section, we consider uniform stability, which guarantees stable $\vx_t$ even with an arbitrary switching sequence. Uniform stability allows us to further build approximation results without enforcing mode synchrony as in \ref{setup_modeSynchrony}.

We let $\xi(\vA_{1:\numSys})$ denote the joint spectral radius of the set of state matrices $\vA_{1:\numSys}$ of $\Sigma$ , i.e.,
$
	\xi(\vA_{1:\numSys}) {:=} \lim_{k \rightarrow \infty} \max_{\omega_{1:k} \in [\numSys]^{k}} $ $ \norm{\vA_{\omega_1} \cdots \vA_{\omega_k}}^{\frac{1}{k}}.
$
We say $\Sigma$ is uniformly (and exponentially) stable if $\xi(\vA_{1:\numSys}){<}1$.
For any $\xi {\geq} \xi(\vA_{1:\numSys})$, we define $\kappa(\vA_{1:\numSys}, \xi)$ to compare the convergence of $\norm{\vA_{\omega_1} \cdots \vA_{\omega_k}}$ and  $\xi^k$ for any mode switching sequence $\omega_{1:k}$:
$
	\kappa(\vA_{1:\numSys}, \xi) := \sup_{k\in \mathbb{N}} $ $ \max_{\omega_{1:k} \in [\numSys]^{k}}\norm{\vA_{\omega_1} \cdots \vA_{\omega_k}} / \xi^k.
$
In the remainder of this paper, we use $\kappa$ as a shorthand notation for $\kappa(\vA_{1:\numSys}, \xi)$.
Note that the pair $\curlybrackets{\xi, \kappa}$ for uniform stability is just the counterpart of $\curlybrackets{\rho, \tau}$ defined earlier for MSS. And similarly, we have (i) $\kappa \geq 1$, (ii) $\norm{\vA_{\omega_1} \cdots \vA_{\omega_k}} \leq \kappa \xi^k$, and (iii) $\kappa$ is finite for any $\xi > \xi(\vA_{1:\numSys})$. 
Furthermore, we let $\bar{\Tcal}:= \max_{i,j} \vT(i,j)$.

We first formally define the transition kernels for $\Sigma$ and $\hat{\Sigma}$ and their distance. Under fixed initial state $\vx_0$ and input sequence $\vu_{0:t-1}$, we define the reachable set $\Xcal_t := \curlybrackets{\vx_t: \vx_{0:t} \text{ is a solution to \eqref{eq_MJSReduction_MJS} } \forall \ \omega_{0:t-1} \in [\numSys]^t}$. 
Then we define the $t$-step transition kernel as $p_t(\vx) := \prob(\vx_t = \vx)$ for all $\vx \in \Xcal_t$. Note that both $\Xcal_t$ and $p_t(\vx)$ depend on the choice of the initial state and input sequence as well. We omit this dependency in the notation not only for simplicity but also because the approximation results we provide hold for arbitrary initial state and input sequence. Similarly, for the reduced $\hat{\Sigma}$, we use $\hat{\Xcal}_t$ to denote the reachable set at time $t$, and for $\vxhat \in \hat{\Xcal}_t$, we let $\phat_t(\vxhat):= \prob(\vxhat_t = \vxhat)$. Then, for $\ell \geq 1$ the $\ell$-Wasserstein distance $W_\ell(p_t, \phat_t)$, between distributions $p_t$ and $\phat_t$ is defined as the optimal objective value of the following mass transportation problem:
\begin{equation}
	\label{eq_MJSReduction_WassDistDef}
	\begin{split}
		\min_{f \geq 0}  \quad& \big( {\textstyle\sum_{\vx \in \Xcal_t, \vxhat \in \hat{\Xcal}_t}} f(\vx, \vxhat) \norm{\vx-\vxhat}^\ell \big)^{1/\ell} \\
		\text{s.t.} \quad & {\textstyle\sum_{\vx \in \Xcal_t}} f(\vx, \vxhat) = \phat_t(\vxhat), \forall \ \vxhat \\
		& {\textstyle\sum_{\vxhat \in \hat{\Xcal}_t}} f(\vx, \vxhat) = p_t(\vx), \forall \ \vx. \\
	\end{split}
\end{equation}
The constraints describe the transportation of probability mass distributed according $p_t$ to the support of $\phat_t$ so that the mass after transportation distributes the same as $\phat_t$. We can view $f(\vx, \vxhat)$ as the mass that is transported from point $\vx$ to $\vxhat$ and $\norm{\vx-\vxhat}$ as the distance it travels. 
When $\ell=1$, the goal is to minimize the total weighted travel distance, and the resulting $W_1$ is also known as the earth mover's distance.
Now we are ready to present our results for the uniform stability assumption. 

\begin{theorem}\label{thrm_MJSReduction_bisimUS}
	Consider setup \ref{setup_InitializationExcitation} where the shared initial state $\vx_0$ and inputs $\vu_{0:t}$ can be arbitrary as long as for all $t$, $\vu_t$ is bounded, i.e., $\norm{\vu_t} \leq \bar{u}$. 
	Assume $\Sigma$ is uniformly stable and $\hat{\Omega}_{1:\numCls} = \Omega_{1:\numCls}$ in Algorithm \ref{Alg_MJSReduction}. 
	For any $\xi \in [\xi(\vA_{1:\numSys}),1)$ and its corresponding $\kappa$, let $\xi_0 := \frac{1+\xi}{2}$.
	For perturbation, we assume $\epsilon_\vA \leq \frac{1-\xi}{2 \kappa}$ and $\epsilon_\vB \leq \Bbar$.	
	Then, we have the following results.
	\begin{enumerate}[label=\textup{(T\arabic*)}]
		\item Under \ref{setup_modeSynchrony}, $\norm{\vx_t - \vxhat_t} \leq \epsilon^{us}_{t}:= t \xi_0^{t-1} \kappa^2 \norm{\vx_0} \epsilon_\vA
			+ \frac{2(1 + t \xi_0^t) \kappa^2 \Bbar \bar{u}}{1 - \xi_0}  \epsilon_\vA
			+ \frac{\kappa \bar{u}}{1 - \xi} \epsilon_\vB$ almost surely.
		\label{enum_bisimUS_stateDistSync}
		\item Consider the autonomous case, i.e., $\vB_{1:\numSys}=0$. (\ref{setup_modeSynchrony} is not mandatory.) Then, 
		$
			W_\ell(p_{t}, \phat_{t})
			\leq 
            t \xi_0^{t \minus 1} \kappa^2 \norm{\vx_0} \epsilon_\vA 
        	 +
        	2 \numCls^2 t \kappa \norm{\vx_0} \numCls^t (\kappa \epsilon_\vA {+} \xi)^t 
        	(\bar{\Tcal} {+} \epsilon_\vT)^{(t-2)/\ell} \epsilon_\vT^{1/\ell}.
		$
		\label{enum_bisimUS_WasDisAutonomousSys}
	\end{enumerate}
\end{theorem}
In Theorem \ref{thrm_MJSReduction_bisimUS}, condition $\epsilon_\vA \leq \frac{1-\xi}{2 \kappa}$ guarantees the reduced $\hat{\Sigma}$ is uniformly stable with joint spectral radius upper bounded by $\xi_0$. The condition $\epsilon_\vB \leq \Bbar$ simplifies the expression, and similar results can be obtained when it is relaxed. 
\ref{enum_bisimUS_stateDistSync} upper bounds the realization difference with the mode synchrony setup. We can see the similarity between the upper bounds $\epsilon_{t}^{us}$ and $\epsilon_{t}^{mss}$ of Theorem \ref{thrm_MJSReduction_stateDistSyncMSS} under MSS assumption.
The fact that uniform stability and MSS upper bound $\norm{\vx_t - \vxhat_t}$ deterministically and in expectation respectively is a manifestation of the difference between these two stability notions for MJS.

In \ref{enum_bisimUS_WasDisAutonomousSys}, we bound the Wasserstein distance between $p_t$ and $\phat_t$. This bound depends on both perturbations $\epsilon_\vA$ and $\epsilon_\vT$. Let $\vmu$ and $\vS$ denote the mean and covariance for $\vx_t$; and similarly define $\vmuhat$ and $\vShat$ for $\vxhat_t$.  From \cite[Theorem 4]{kuhn2019wasserstein}, we obtain $\norm{\vmu - \vmuhat}^2 + d(\vS, \vShat) \leq W_2(p_{t}, \phat_{t})^2$, where $d(\vS, \vShat):= \tr(\vS + \vShat - 2 (\vS^{\frac{1}{2}} \vShat \vS^{\frac{1}{2}})^{\frac{1}{2}})$ is a metric between $\vS$ and $\vShat$. Hence, by setting $\ell = 2$ in \ref{enum_bisimUS_WasDisAutonomousSys}, we also obtain upper bounds for the differences between $p_t$ and $\phat_t$ in terms of their first and second order moments, i.e., $\norm{\vmu - \vmuhat}$ and $d(\vS, \vShat)$. These metrics can be used to obtain performance bounds in other control problems such as covariance steering \cite{chen2015optimal, goldshtein2017finite, okamoto2018optimal} and ensemble control \cite{li2015ensemble}.

\section{Stability Analysis}\label{sec_MJSReduction_stabAnalysis}
In this section, we study whether the stability properties of $\Sigma$ can be deduced from those of $\hat{\Sigma}$.
Recall that MSS of $\Sigma$ depends on $\rho(\vAcal)$, the spectral radius of its augmented state matrix $\vAcal$, and its uniform stability depends on $\xi(\vA_{1:\numSys})$, the joint spectral radius of state matrices $\vA_{1: \numSys}$. Similarly, for $\hat{\Sigma}$, we define its augmented state matrix $\vAcalhat \in \dm{\numCls \dimSt^2}{\numCls \dimSt^2}$ with its $ij$-th $\dimSt^2 \times \dimSt^2$ block given by $[\vAcalhat]_{ij} := \vThat(j,i) \cdot \vAhat_j \otimes \vAhat_j$ and let $\rho(\vAcalhat)$ denote its spectral radius; we let $\xi(\vAhat_{1: \numCls})$ denote the joint spectral radius of state matrices $\vAhat_{1: \numCls}$. With these notations, we want to analyze when $\rho(\vAcalhat)$ (or $\xi(\vAhat_{1: \numCls})$) can be taken as an approximation for $\rho(\vAcal)$ (or $\xi(\vA_{1:\numSys})$) since computing or approximating $\rho(\vAcalhat)$ and $\xi(\vAhat_{1: \numCls})$ may require much less computation compared with $\rho(\vAcal)$ and $\xi(\vA_{1:\numSys})$ as $\hat{\Sigma}$ has much fewer number of modes than $\Sigma$. 

To begin with, we first construct an intermediate MJS by \textit{expanding} the reduced $\hat{\Sigma}$: we let $\bar{\Sigma}:=\text{MJS}(\vAbar_{1:\numSys}, \vBbar_{1:\numSys}, \vTbar)$ such that $\vTbar \in \Lcal(\vT, \hat{\Omega}_{1:\numCls}, \epsilon_\vT)$, and for all $i\in [\numSys]$ (suppose $i \in \hat{\Omega}_k$), $\vAbar_i = \vAhat_k$, $\vBbar_i = \vBhat_k$. 
By definition of $\Lcal(\vT, \hat{\Omega}_{1:\numCls}, \epsilon_\vT)$, we can solve for $\vTbar$ through a linear programming feasibility problem with constraints given by the definition of $\Lcal(\cdot, \cdot, \cdot)$ in \eqref{eq_NeighborhoodT}. Particularly, if it is the aggregatable case \ref{problem_aggregatableCase}, it suffices to let $\vTbar(i,:) :=|\hat{\Omega}_k|^\inv \sum_{i \in \hat{\Omega}_k} \vT(i,:)$ if $i \in \hat{\Omega}_k$.
Note that by construction, $\bar{\Sigma}$ is mode-reducible with respect to $\hat{\Omega}_{1:\numCls}$ and can be reduced to $\hat{\Sigma}$. According to Fact \ref{fact_MJSReduction_existReducedSys}, $\bar{\Sigma}$ has the same dynamics as $\hat{\Sigma}$. Since $\bar{\Sigma}$ has the same number of modes as $\Sigma$, we can use $\bar{\Sigma}$ as a bridge to compare $\Sigma$ and $\hat{\Sigma}$. 
We let $\rho(\vAcalbar)$ denote the spectral radius of $\vAcalbar \in \dm{\numSys \dimSt^2}{\numSys \dimSt^2}$ whose $ij$-th $\dimSt^2 \times \dimSt^2$ block is given by $[\vAcalbar]_{ij} := \vTbar(j,i) \cdot \vAbar_j \otimes \vAbar_j$ and let $\xi(\bar{\vA}_{1:\numSys})$ denote the joint spectral radius of $\bar{\vA}_{1:\numSys}$. The following preliminary result (proof omitted due to its simplicity) says $\bar{\Sigma}$ and $\hat{\Sigma}$ have the same stability properties.
\begin{lemma}\label{lemma_MJSReduction_expandedMJSSpectralRadius}
 	For $\hat{\Sigma}$ and $\bar{\Sigma}$, we have $\rho(\vAcalhat) = \rho(\vAcalbar)$ and $\xi(\hat{\vA}_{1:\numSys}) = \xi(\bar{\vA}_{1:\numSys})$.
\end{lemma}
One implication of Lemma \ref{lemma_MJSReduction_expandedMJSSpectralRadius} is that if an MJS is mode-reducible, the reduced MJS has the same MSS and uniform stability as the original MJS in terms of (joint) spectral radius. 
When $\Sigma$ is not exactly mode-reducible, Lemma \ref{lemma_MJSReduction_expandedMJSSpectralRadius} allows us to compare the stability properties of $\hat{\Sigma}$ and $\Sigma$ via the intermediate expanded $\bar{\Sigma}$ as presented in Theorem \ref{thrm_MJSReduction_stabilityAnalysis}. For analysis purposes, similar to $\tau$ and $\kappa$ defined for $\Sigma$, 
we define $\bar{\tau} := \sup_{k\in \mathbb{N}} \norm{\vAcalbar^k}/\rhohat^k$ for any $\rhohat \geq \rho(\vAcalhat)$ and $\bar{\kappa}:= \sup_{k\in \mathbb{N}} \max_{\omega_{1:k} \in [\numCls]^{k}}\norm{\vAbar_{\omega_1} \cdots \vAbar_{\omega_k}} / \xihat^k$ for any $\xihat \geq \xi(\vAhat_{1:\numSys})$

\begin{theorem}[Stability Analysis] \label{thrm_MJSReduction_stabilityAnalysis}
	Assume $\hat{\Omega}_{1:\numCls} = \Omega_{1:\numCls}$ in Algorithm \ref{Alg_MJSReduction}, then $\Sigma$ and $\hat{\Sigma}$ have the following relations.
 	\begin{enumerate}[label=\textup{(T\arabic*)}]
        \item (MSS) For any $\rho \geq \rho(\vAcal)$ and its corresponding $\tau$, any $\rhohat \geq \rho(\vAcalhat)$ and its corresponding $\bar{\tau}$,
        we have
        \begin{equation}
        \begin{split}
            \rho(\vAcalhat) - \rho(\vAcal) &\leq \tau \epsilon_\rho + (\rho - \rho(\vAcal)) \\
            \rho(\vAcal) - \rho(\vAcalhat) & \leq \bar{\tau} \epsilon_\rho + (\rhohat - \rho(\vAcalhat))
        \end{split}
        \end{equation}
        where $\epsilon_\rho:= \sqrt{\numSys} ((2 \Abar + \epsilon_\vA) \epsilon_\vA + \Abar^2 \epsilon_\vT)$.
        \label{enum_MSSStabilityComparison}
        \item (Uniform stability) For any $\xi \geq \xi(\vA_{1:\numSys})$ and its corresponding $\kappa$, any $\xihat \geq \xi(\vAhat_{1:\numCls})$ and its corresponding $\bar{\tau}$,
        we have
        \begin{equation}
        \begin{split}
            \xi(\vAhat_{1:\numCls}) - \xi(\vA_{1:\numSys}) &\leq \kappa \epsilon_\vA + (\xi - \xi(\vA_{1:\numSys})) \\
            \xi(\vA_{1:\numSys}) - \xi(\vAhat_{1:\numCls}) & \leq \bar{\kappa} \epsilon_\vA + (\hat{\xi} - \xi(\vAhat_{1:\numCls})).
        \end{split}
        \end{equation}
        \label{enum_USStabilityComparison}
    \end{enumerate}
\end{theorem}

\begin{proof}
From Lemma \ref{lemma_MJSReduction_expandedMJSSpectralRadius}, it suffices to prove
\begin{gather}
\rho(\vAcalbar) \leq \tau \epsilon_\rho + \rho, \quad
\rho(\vAcal) \leq \taubar \epsilon_\rho + \rhohat. \label{eq_MJSReduction_33}\\
\xi(\vAbar_{1:\numSys}) \leq \kappa \epsilon_\vA + \xi, \quad
\xi(\vA_{1:\numSys}) \leq \kappabar \epsilon_\vA + \xihat. \label{eq_MJSReduction_34}
\end{gather}
Since we assume $\hat{\Omega}_{1:\numCls} = \Omega_{1:\numCls}$, then for $\Sigma$ and $\bar{\Sigma}$, we have $\norm{\vAbar_i - \vA_i} \leq \epsilon_\vA$,  $\norm{\vBbar_i - \vB_i} \leq \epsilon_\vB$, and $\norm{\vTbar - \vT}_{\infty} \leq \epsilon_\vT$. Consider matrix $\vAcalbar$ and $\vAcal$, we have 
$
[\vAcalbar]_{ij} - [\vAcal]_{ij}
= \vTbar(j,i) \vAbar_j \otimes \vAbar_j - \vT(j,i) \vA_j \otimes \vA_j
= \vTbar(j,i) (\vAbar_j \otimes \vAbar_j - \vA_j \otimes \vA_j) + (\vTbar(j,i) - \vT(j,i)) \vA_j \otimes \vA_j. 
$
Note that
$
\vAbar_j \otimes \vAbar_j - \vA_j \otimes \vA_j
= (\vAbar_j - \vA_j) \otimes \vA_j + \vA_j \otimes (\vAbar_j - \vA_j) + (\vAbar_j - \vA_j) \otimes (\vAbar_j - \vA_j),
$
which gives
$
\norm{\vAbar_j \otimes \vAbar_j - \vA_j \otimes \vA_j}
\leq (2 \Abar + \epsilon_\vA) \epsilon_\vA.
$
Then, we have
$
		\norm{[\vAcalbar]_{ij} - [\vAcal]_{ij}} 
		\leq \vTbar(j,i) (2 \Abar + \epsilon_\vA) \epsilon_\vA + | \vTbar(j,i) - \vT(j,i) | \Abar^2.
$
To simplify the notation, we let $c_1 := (2 \Abar + \epsilon_\vA) \epsilon_\vA$ and $c_2 := \Abar^2$. By Cauchy-Schwarz inequality, we have $\sum_i \norm{[\vAcalbar]_{ij} - [\vAcal]_{ij}}^2 \leq (c_1 \norm{\vTbar(j,:)} + c_2 \norm{ \vTbar(j,:) - \vT(j,:)})^2$. Thus, $\norm{\vAcalbar - \vAcal} \leq \sqrt{\numSys} \max_j (\sum_i \norm{[\vAcalbar]_{ij} - [\vAcal]_{ij}})^{0.5} \leq \sqrt{\numSys} (c_1+ c_2 \epsilon_\vT) =: \epsilon_\rho$.

With Corollary \ref{corollary_MJSReduction_singleMatrixPerturbation} in the appendix, we have $\norm{\vAcalbar^k} \leq \tau (\tau \epsilon_\rho + \rho)^k$. By Gelfand's formula, $\rho(\vAcalbar) {=} \limsup_{k \rightarrow \infty} \norm{\vAcalbar^k}^{\frac{1}{k}} {\leq} \tau \epsilon_\rho + \rho$, which shows the left inequality of \eqref{eq_MJSReduction_33}. If we use Corollary \ref{corollary_MJSReduction_singleMatrixPerturbation} the other way, we have $\norm{\vAcal^k} \leq \bar{\tau} (\bar{\tau} \epsilon_\rho + \hat{\rho})^k$, which similarly implies $\rho(\vAcal) \leq \bar{\tau} \epsilon_\rho + \hat{\rho}$. With these results, \eqref{eq_MJSReduction_33} is proved. 
\eqref{eq_MJSReduction_34} can be shown similarly by noticing $\norm{\vAbar_i - \vA_i} \leq \epsilon_\vA$ and then using Lemma \ref{lemma_MJSReduction_JSRPerturb} in the appendix.
\end{proof}

Theorem \ref{thrm_MJSReduction_stabilityAnalysis} provides upper bounds on $|\rho(\vAcalhat) - \rho(\vAcal)|$ and $|\xi(\vAhat_{1:\numCls}) - \xi(\vA_{1:\numSys})|$.
By definition, $\tau$ decreases when $\rho$ increases, and the same applies to the pairs $\{\taubar, \rhohat\}$, $\{\kappa, \xi\}$, and $\{\kappabar, \xihat\}$. Hence, for fixed $\epsilon_\rho$ and $\epsilon_\vA$, by tuning the free parameters $\rho$, $\rhohat$, $\xi$, and $\xihat$, one may obtain tighter upper bounds in Theorem \ref{thrm_MJSReduction_stabilityAnalysis}. 
When $\rho = \rho(\vAcal)$, $\rhohat = \rho(\vAcalhat)$, the bound in \ref{enum_MSSStabilityComparison} becomes tight at $\epsilon_\rho = 0$ as the upper and lower bounds meet at $0$.
Note that these results hold for both stable and unstable $\Sigma$, and does not require perturbation $\epsilon_\vA, \epsilon_\vB, \epsilon_\vT$ to be small, which is in contrast to approximation results in Theorem \ref{thrm_MJSReduction_stateDistSyncMSS} and \ref{thrm_MJSReduction_bisimUS}.

Now we briefly compare the complexities for computing or approximating $\rho(\vAcal), \rho(\vAcalhat), \xi(\vA_{1:\numSys})$, and $\xi(\vAhat_{1:\numCls})$. Since $\vAcal$ has dimension $\numSys \dimSt^2 \times \numSys \dimSt^2$, the complexity to compute its spectral radius $\rho(\vAcal)$ is $\Ocal(\numSys^3 \dimSt^6)$, but it only requires $\Ocal(\numCls^3 \dimSt^6)$ for $\rho(\vAcalhat)$. Computation of the joint spectral radius is in general undecidable \cite{jungers2009joint}. An iterative approach \cite{parrilo2008approximation} provides an approximation for $\xi(\vA_{1:\numSys})$ with computational complexity $\mathcal{O}(\numSys)$, whereas it only requires $\mathcal{O}(\numCls)$ for $\xi(\vAhat_{1:\numCls})$. 

\section{Controller Design with Case Study on LQR}\label{sec_MJSReduction_LQRControl}
When the mode of an MJS can be measured at run-time, one can use \textit{mode-dependent} controllers. A mode-dependent controller is essentially a collection of individual controllers, one per mode, and the deployed controller switches with corresponding modes. Therefore, if we can reduce the modes, that would also reduce the number of controllers in a mode-dependent control. That is, with the reduced $\hat{\Sigma}$, we can design mode-dependent controller $\vKhat_{1:\numCls}$ for $\hat{\Sigma}$ and then associate every mode $i$ in $\Sigma$ with $\vKhat_k$ if $i \in \hat{\Omega}_k$.
Since $\hat{\Sigma}$ has a smaller scale than $\Sigma$, the computational cost may be reduced but the question is how this simplified controller performs on the original system $\Sigma$. 
In this section, we show how this idea can be used for linear quadratic regulator (LQR) for MJS and provide suboptimality guarantees for the reduced controller.

In the infinite horizon MJS LQR problems, given positive definite cost matrices $\vQ$ and $\vR$, we define quadratic cumulative cost $J_T = \expctn \squarebracketsbig{ \sum_{t=0}^{T-1}  \parenthesesbig{\vx_t^\T \vQ \vx_t + \vu_t^\T \vR \vu_t} + \vx_T^\T \vQ \vx_T}$. The goal is to design inputs to minimize the infinite time average cost $\limsup_{T\rightarrow \infty} \frac{1}{T} J_T$ under $\Sigma$. To ease the exposition, we let $\mathbb{S}_\numSys^+ := \curlybrackets{\vX_{1:\numSys}: \forall i \in [\numSys], \vX_i \in \dm{\dimSt}{\dimSt}, \vX_i \succeq 0}$. For $\vX_{1:\numSys} \in \mathbb{S}_\numSys^+$, for all $i \in [\numSys]$, define three operators $\varphi_i(\vX_{1:\numSys}) {:=} \sum_{j \in [\numSys]} \vT(i,j) \vX_j$, $\Kcal_i(\vX_{1:\numSys}) {:=} {-} \hspace{-0.3em} \parenthesesbig{\vR + \vB_i^\T \varphi_i(\vX_{1:\numSys}) \vB_i}^{\inv} \hspace{-0.3em} \parenthesesbig{\vB_i^\T \varphi_i(\vX_{1:\numSys}) \vA_i}$, and 
\begin{multline}
    \Rcal_i(\vX_{1:\numSys}) {:=} \vQ + \vA_i^\T \varphi_i(\vX_{1:\numSys}) \vA_i - \vA_i^\T \varphi_i(\vX_{1:\numSys})^\T \vB_i \\
	\qquad \cdot \parenthesesbig{\vR + \vB_i^\T \varphi_i(\vX_{1:\numSys}) \vB_i}^{\inv} \vB_i^\T \varphi_i(\vX_{1:\numSys}) \vA_i. \label{eq_MJSReduction_RiccatiOperator}
\end{multline}
Then, the solution to the infinite horizon LQR is the following: we first solve for the coupled Riccati equations
$
\vP_i = \Rcal_i(\vP_{1:\numSys}), \forall i \in [\numSys],
$
and then if $\omega_t = i$ at time $t$, we let input $\vu_t = \vK_i \vx_t$ where
$
\vK_i = \Kcal_i(\vP_{1:\numSys}).
$
The solution existence, uniqueness, and optimality can be guaranteed by the following assumption according to \cite{costa2006discrete}.
\begin{assumption}\label{assmp_LQR}
	$\Sigma$ is mean-square stabilizable. Cost matrices $\vQ \succ 0$ and $\vR \succ 0$.
\end{assumption}

To design controllers with the reduced $\hat{\Sigma}$, we can first compute controller $\vKhat_{1:\numCls}$ by solving LQR problem with $\hat{\Sigma}$ as the MJS dynamics. This requires solving $\numCls$ coupled Riccati equations, each of which is parameterized by $\vAhat_i, \vBhat_i,$ $\vThat(i,:), \vQ, \vR$. To apply $\vKhat_{1:\numCls}$ to the original $\Sigma$, we simply let $\vu_t = \vKhat_k \vx_t$ if $\omega_t = k$. Since the number of coupled Riccati equations is the same as the number of modes, the computational cost for $\Sigma$ is $\Ocal(\numSys)$ while only $\Ocal(\numCls)$ for $\hat{\Sigma}$, thus the saving is prominent when $\numCls \ll \numSys$.

Next, we analyze the suboptimality when applying controllers computed with $\hat{\Sigma}$. 
To begin with, similar to the notations for $\Sigma$, for $\hat{\Sigma}$ we define $\hat{\varphi}_{1:\numCls}, \hat{\Rcal}_{1:\numCls}$, $\hat{\Kcal}_{1:\numCls}$, $\vPhat_{1:\numCls}$, and $\vKhat_{1:\numCls}$. Particularly, $\vPhat_{1:\numCls}$ denotes the Riccati solution such that $\vPhat_i = \hat{\Rcal}_i(\vPhat_{1:\numCls})$, and $\vKhat_{1:\numCls}$ is computed such that $\vKhat_i = \hat{\Kcal}_i(\vPhat_{1:\numCls})$.
We will take the expanded and mode-reducible MJS $\bar{\Sigma}$ constructed with $\hat{\Sigma}$ and $\hat{\Omega}_{1:\numCls}$ in Section \ref{sec_MJSReduction_stabAnalysis} as a bridge. For $\bar{\Sigma}$, we similarly define $\bar{\varphi}_{1:\numSys}, \bar{\Rcal}_{1:\numSys}$, $\bar{\Kcal}_{1:\numSys}$, $\vPbar_{1:\numSys}$, and $\vKbar_{1:\numSys}$.
In terms of LQR solutions, the relation between $\hat{\Sigma}$ and $\bar{\Sigma}$ is given below.
\begin{lemma}\label{lemma_MJSReduction_LQR}
	Assume the Riccati solution $\vPbar_{1:\numSys}$ exists and $\vPbar_i \succ 0$ for all $i$. Then, (i) there exists a unique Riccati solution $\vPhat_{1:\numCls}$ in $\mathbb{S}_\numCls^+$; (ii) $\vPhat_k {=} \vPbar_i, \vKhat_k {=} \vKbar_i$ for any $i \in \hat{\Omega}_k$ for any $k$.
\end{lemma}
\begin{proof}
    We consider the Riccati operator iteration defined as follows: $\vPbar_i^{(0)} = \vQ$, $\vPbar_i^{(h+1)} = \bar{\Rcal}_i (\vPbar_i^{(h)})$ for all $i \in [\numSys], h \in \mathbb{N}$ and $\vPhat_k^{(0)} = \vQ$, $\vPhat_i^{(h+1)} = \hat{\Rcal}_i (\vPhat_i^{(h)})$ for all $k \in [\numCls], h \in \mathbb{N}$. Then, note that by construction, for all $i \in \hat{\Omega}_k$ and all $l \in [\numCls]$, we have $\sum_{j \in \hat{\Omega}_l}\vTbar(i,j) = \vThat(k,l)$. Through induction and algebra, it is easy to show that for all $h \in \mathbb{N}$, and for any $i, i' \in \hat{\Omega}_k$ for any k, we have $\vPbar_i^{(h)} = \vPbar_{i'}^{(h)} = \vPhat_k^{(h)}$.
    
    Since $\vPbar_i {\succ} 0$, by \cite[Fact 4]{du2021certainty}, we know $\vPbar_{1:\numSys}$ is the unique solution among $\mathbb{S}_\numSys^+$, and $\vKbar_{1:\numSys}$ stabilizes $\bar{\Sigma}$.
    According to \cite[Proposition A.23]{costa2006discrete}, the stabilizability of $\bar{\Sigma}$ and the fact $\vQ, \vR {\succ} 0$ imply $\lim_{h \rightarrow \infty} \vPbar_i^{(h)} {=} \vPbar_i$. Combining this convergence result with the Riccati iteration results we just showed, we further have, for any $i, i' {\in} \hat{\Omega}_k$ and any k, we have $\vPbar_i {=} \vPbar_{i'} {=} \vPhat_k$. Then, it is easy to show that $\vKbar_i {=} \vKbar_{i'} {=} \vKhat_k$.
    The uniqueness of $\vPhat_{1:\numSys}$ can be shown by contradiction.
\end{proof}
With this lemma, we have the following suboptimality guarantees in terms of applying controller $\vKhat_{1:\numCls}$ to $\Sigma$.

\begin{theorem}[LQR Suboptimality]\label{thrm_MJSReduction_LQRSuboptimality}
	Assume \ref{assmp_LQR} holds for $\Sigma$, and $\Sigma$ has additive Gaussian noise $\N(0, \sigma_\vw^2 \vI_\dimSt)$ that is independent of the mode switching. Let $J^\star$ and $\Jhat$ respectively denote the infinite time average cost incurred by the optimal controller $\vK_{1:\numSys}$ and controller $\vKhat_{1:\numCls}$ (at time $t$, $\vu_t = \vKhat_k \vx_t$ if $\omega_t \in \hat{\Omega}_k$). Then, there exists constants $\bar{\epsilon}_{\vA,\vB}$, $ \bar{\epsilon}_\vT$, $C_{\vA, \vB}$, and $C_{\vT}$,
	such that when $\max\curlybrackets{\epsilon_\vA, \epsilon_\vB} \leq \bar{\epsilon}_{\vA,\vB}$ and $\epsilon_\vT \leq \bar{\epsilon}_\vT$,
	\begin{equation}\label{eq_lqrSubopt_avgCost}
	    \Jhat - J^\star \leq \sigma_\vw^2 (C_{\vA, \vB} \max\curlybrackets{\epsilon_\vA, \epsilon_\vB} + C_{\vT} \epsilon_\vT)^2
	\end{equation}
	Let $J^\star_\infty$ and $\Jhat_\infty$ denote the infinite time cumulative cost incurred by $\vK_{1:\numSys}$ and $\vKhat_{1:\numCls}$ respectively. Then, when $\sigma_\vw = 0$, $\epsilon_\vT \leq \bar{\epsilon}_\vT$, and $\max\curlybrackets{\epsilon_\vA, \epsilon_\vB} \leq \bar{\epsilon}_{\vA,\vB}$,
	\begin{equation}\label{eq_lqrSubopt_cumCost}
	    \Jhat_\infty - J^\star_\infty
	    \leq (C'_{\vA, \vB} \max\curlybrackets{\epsilon_\vA, \epsilon_\vB} + C'_{\vT} \epsilon_\vT) \norm{\vx_0}^2,
	\end{equation}
	for some constants $C'_{\vA, \vB}$ and $C'_{\vT}$.
\end{theorem}
\begin{proof} We will use Lemma \ref{lemma_MJSReduction_LQR} and $\bar{\Sigma}$ and $\vKbar_{1:\numSys}$ as a bridge to compare $\vKhat_{1:\numCls}$ and $\vK_{1:\numSys}$.	
First, we prove \eqref{eq_lqrSubopt_avgCost}.
Comparing $\bar{\Sigma}$ and $\Sigma$, one can see $\norm{\vAbar_i - \vA_i} \leq \epsilon_\vA$, $\norm{\vBbar_i - \vB_i} \leq \epsilon_\vB$, and $\norm{\vTbar - \vT}_\infty \leq \epsilon_\vT$. Then, from \cite[Theorem 6]{du2021certainty} we know when $\max\curlybrackets{\epsilon_\vA, \epsilon_\vB} \leq \bar{\epsilon}_{\vA,\vB}$ and $\epsilon_\vT \leq \bar{\epsilon}_\vT$ for some constants $\bar{\epsilon}_{\vA,\vB}$ and $\bar{\epsilon}_\vT$, the Riccati solution $\vPbar_{1:\numSys}$ uniquely exists among $\mathbb{S}_\numSys^+$ and are positive definite, and the cost $\Jbar$ when applying $\vKbar_{1:\numSys}$ to $\Sigma$ has suboptimality $\Jbar - J^\star \leq \sigma_\vw^2 (C_{\vA, \vB} \max\curlybrackets{\epsilon_\vA, \epsilon_\vB} + C_{\vT} \epsilon_\vT)$ for some constants $C_{\vA, \vB}$ and $C_{\vT}$. Using Lemma \ref{lemma_MJSReduction_LQR}, we know $\vPhat_{1:\numCls}$ uniquely exists $\mathbb{S}_\numCls^+$, and $\vKhat_k = \vKbar_i$ for any $i$ belonging to any $\hat{\Omega}_k$, which implies applying $\vKbar_{1:\numSys}$ is equivalent to applying $\vKhat_{1:\numCls}$ as in the theorem statement. Thus $\Jbar {=} \Jhat$, and $\Jhat - J^\star {=} \Jbar - J^\star \leq \sigma_\vw^2 (C_{\vA, \vB} \max\curlybrackets{\epsilon_\vA, \epsilon_\vB} + C_{\vT} \epsilon_\vT)$.

Next, we prove \eqref{eq_lqrSubopt_cumCost}. Similar as above, we let $\Jbar_\infty$ denote the cumulative cost when applying $\vKbar_{1:\numSys}$ to $\Sigma$, then we have $\Jhat_\infty = \Jbar_\infty$. From the proof of \cite[Theorem 4.5]{costa2006discrete}, we have $\Jbar_\infty - J^\star_\infty = \sum_{t=0}^\infty \expctn[\norm{\vM_{\omega_t} (\vKbar_{\omega_t} - \vK^\star_{\omega_t}) \vx_t}^2 ]$ where $\vM_i = \vR + \vB_i^\T \varphi_i(\vP_{1:\numSys}) \vB_i$ and $\vx_t$ is driven by controller $\vKbar_{1:\numSys}$. From \cite[Theorem 6]{du2021certainty}, we know when $\max\curlybrackets{\epsilon_\vA, \epsilon_\vB} \leq \bar{\epsilon}_{\vA,\vB}$ and $\epsilon_\vT \leq \bar{\epsilon}_\vT$, then $\vKbar_{1:\numSys}$ is a stabilizing controller and $\norm{\vKbar_{1:\numSys} - \vK^\star_{1:\numSys}} \leq C^{\vK}_{\vA,\vB} \max\curlybrackets{\epsilon_\vA, \epsilon_\vB} + C^{\vK}_{\vT} \epsilon_{\vT}$ for some constants $C^{\vK}_{\vA,\vB}$ and $C^{\vK}_{\vT}$. Following \cite[Lemma 10]{sattar2021identification}, we know $\sum_{t=0}^\infty \expctn[\norm{\vx_t}^2] \leq C_\vx \norm{\vx_0}^2$ for some constant $C_\vx$. Combining these results, we have $\Jhat_\infty - J^\star_\infty \leq \norm{\vM_{1:\numSys}} \norm{\vKbar_{1:\numSys} - \vK^\star_{1:\numSys}} \sum_{t=0}^\infty \expctn[\norm{\vx_t}^2] \leq \norm{\vM_{1:\numSys}} C_\vx (C^{\vK}_{\vA,\vB} \max\curlybrackets{\epsilon_\vA, \epsilon_\vB} + C^{\vK}_{\vT} \epsilon_{\vT}) \norm{\vx_0}^2$.
\end{proof}
In Theorem \ref{thrm_MJSReduction_LQRSuboptimality}, constants $\bar{\epsilon}_{\vA,\vB}$, $ \bar{\epsilon}_\vT$, $C_{\vA, \vB}$, $C_{\vT}$, $C'_{\vA, \vB}$, and $C'_{\vT}$ only depend on the original MJS $\Sigma$ and cost matrices $\vQ$ and $\vR$, and their exact expressions can be obtained following the proof and corresponding references.
As a sanity check, when there is no perturbation, i.e., mode-reducible case, then we have $\Jhat = J^\star$ and $
\Jhat_\infty = J^\star_\infty$, which can also be implied from Lemma \ref{lemma_MJSReduction_LQR}.
For the reduced MJS $\hat{\Sigma}$, its Riccati solution $\vPhat_{1:\numCls}$ and thus controllers $\vKhat_{1:\numCls}$ are guaranteed to exist when perturbation $\epsilon_\vA, \epsilon_\vB$, and $\epsilon_\vT$ are small enough as required in Theorem \ref{thrm_MJSReduction_LQRSuboptimality}. 
The additive noise in Theorem \ref{thrm_MJSReduction_LQRSuboptimality} means the MJS dynamics is given by $\vx_{t+1} = \vA_{\omega_t}\vx_t + \vB_{\omega_t}\vu_t + \vw_t$ where $\vw_t \sim \N(0, \sigma_\vw^2 \vI_\dimSt)$. 
In the noisy case, both $\Jhat_\infty$ and $J^\star_\infty$ are infinite, so the cumulative suboptimality $\Jhat_\infty - J^\star_\infty$ is only studied for the noise-free case as in \eqref{eq_lqrSubopt_cumCost}. 
On the other hand, in the noise-free case, we have not only $J^\star = \Jhat$ as implied by \eqref{eq_lqrSubopt_avgCost}, but also $J^\star = \Jhat = 0$ as long as $\vKhat_{1:\numCls}$ is stabilizing.

\section{Numerical Experiments}\label{sec_MJSReduction_Experiment}
In this section, we present synthetic experiments to evaluate the main results in the paper.
We evaluate the clustering performance of Algorithm \ref{Alg_MJSReduction} and the LQR controller designed with the reduced MJS $\hat{\Sigma}$ as discussed in Section \ref{sec_MJSReduction_LQRControl}.
All the experiments are performed using MATLAB R2020a on a laptop with Xeon E3-1505M CPU. We use the kmeans() function from the Statistics and Machine Learning Toolbox in MATLAB for the k-means problem in Algorithm \ref{Alg_MJSReduction}.

\subsection{Clustering Evaluation}
We consider the uniform partition $\Omega_{1:\numCls}$, i.e. $|\Omega_i| {=} \bar{\numSys} {:=} \numSys / \numCls$ for any $i$. 
The system $\Sigma$ is randomly generated according to \ref{problem_lumpableCase} or \ref{problem_aggregatableCase} with desired levels of perturbation $\epsilon_\vA, \epsilon_\vB, \epsilon_\vT$ so that in \eqref{eq_MJSReduction_approxDyn} each summand $\norm{\vA_i - \vA_{i'}} \leq \epsilon_\vA/(\numCls \bar{\numSys}^2)$. The same applies to $\vB_{1:\numSys}$ and $\vT$.
Specifically, we first randomly generate a small scale MJS $\breve{\Sigma} = \textup{MJS}(\breve{\vA}_{1:\numCls}, \breve{\vB}_{1:\numCls}, \breve{\vT})$: we sample each matrix element in $\breve{\vA}_k$ and $\breve{\vB}_k$ from standard Gaussian distributions and then scale the matrices so that each $\norm{\breve{\vA}_k} = 0.5$ and $\norm{\breve{\vB}_k}=1$ unless otherwise mentioned; and each $\breve{\vT}(i,:)$ is sampled from the flat Dirichlet distribution. 
Then, we generate $\Sigma$ by augmenting $\breve{\Sigma}$. For every mode $i \in \Omega_k$, we let $\vA_i = \breve{\vA}_k + \vE_i$ and $\vB_i = \breve{\vB}_k + \vF_i$ where we sample elements in $\vE_i$ and $\vE_i$ from standard Gaussian and then scale them so that $\norm{\vE_i}_\fro = \frac{\epsilon_\vA}{2 \numCls \bar{\numSys}^2} $ and $\norm{\vF_i}_\fro = \frac{\epsilon_\vB}{2 \numCls \bar{\numSys}^2}$. The generation of $\vT$ is a bit involved.
For the aggregatable case \ref{problem_aggregatableCase}, we first generate a Markov matrix $\vTbar \in \dm{\numSys}{\numSys}$ such that for every $i \in \Omega_k$, $\vTbar(i, \Omega_l) = \va_{k,l} \breve{\vT}(k,l)$ where $\va_{k,l} \in \dm{1}{|\Omega_l|}$ is sampled from the flat Dirichlet distribution; then we let $\vT(i,:) = (1-\frac{\epsilon_\vT}{2 \numCls \bar{\numSys}^2})\vTbar(i,:) + \frac{\epsilon_\vT}{2 \numCls \bar{\numSys}^2} \vb_i$ where $\vb_i \in \dm{1}{\numSys}$ is again sampled from the flat Dirichlet distribution.
The same steps are used to generate $\vT$ for the lumpable case \ref{problem_lumpableCase} except that 
$\vTbar(i, \Omega_l) = \va_{i,l} \breve{\vT}(k,l)$.
Following these steps, $\Sigma$ satisfies the perturbation conditions in \ref{problem_lumpableCase} and \ref{problem_aggregatableCase}.

To evaluate Algorithm \ref{Alg_MJSReduction}, we fix $\dimSt{=}5$, $\dimInput{=}3$, $\numCls{=}4$ and record the misclustering rate (MR) defined in Section \ref{subsec_MJSReduction_theoryClustering}
over 100 runs. 
Fig. \ref{fig_ClusteringResult} presents the clustering performances under different number of modes $\numSys$ and perturbations $\epsilon_\vA, \epsilon_\vB$ and $\epsilon_\vT$.
In the plots, we normalized the perturbation on the x-axis by $\numSys^2$ so that the trends under different $\numSys$ can be better visualized. This also follows from the experiment setup: each summand in \eqref{eq_MJSReduction_approxDyn} has $\norm{\vA_i - \vA_{i'}} \leq \Ocal(\epsilon_\vA/\numSys^2)$.
It is clear that the clustering performance degrades with increasing $\numSys$ and perturbations. We can also observe that when the perturbation is small, there are no misclustered modes. 
\begin{figure} 
	\centering
	\captionsetup[subfloat]{captionskip=0pt, farskip=0pt}
	\subfloat[]{%
		\includegraphics[width=1.6in]{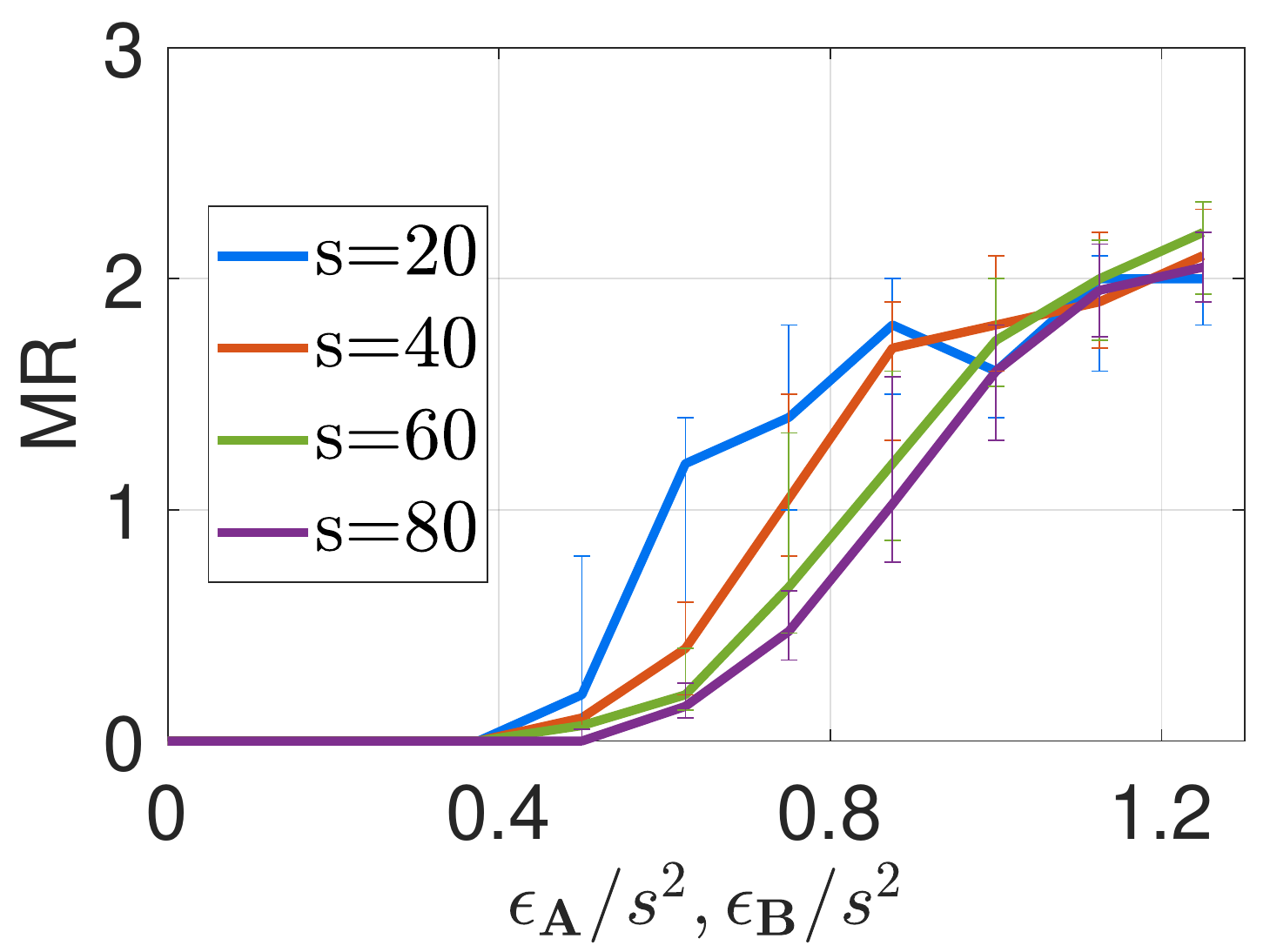}}
	\subfloat[]{%
		\includegraphics[width=1.6in]{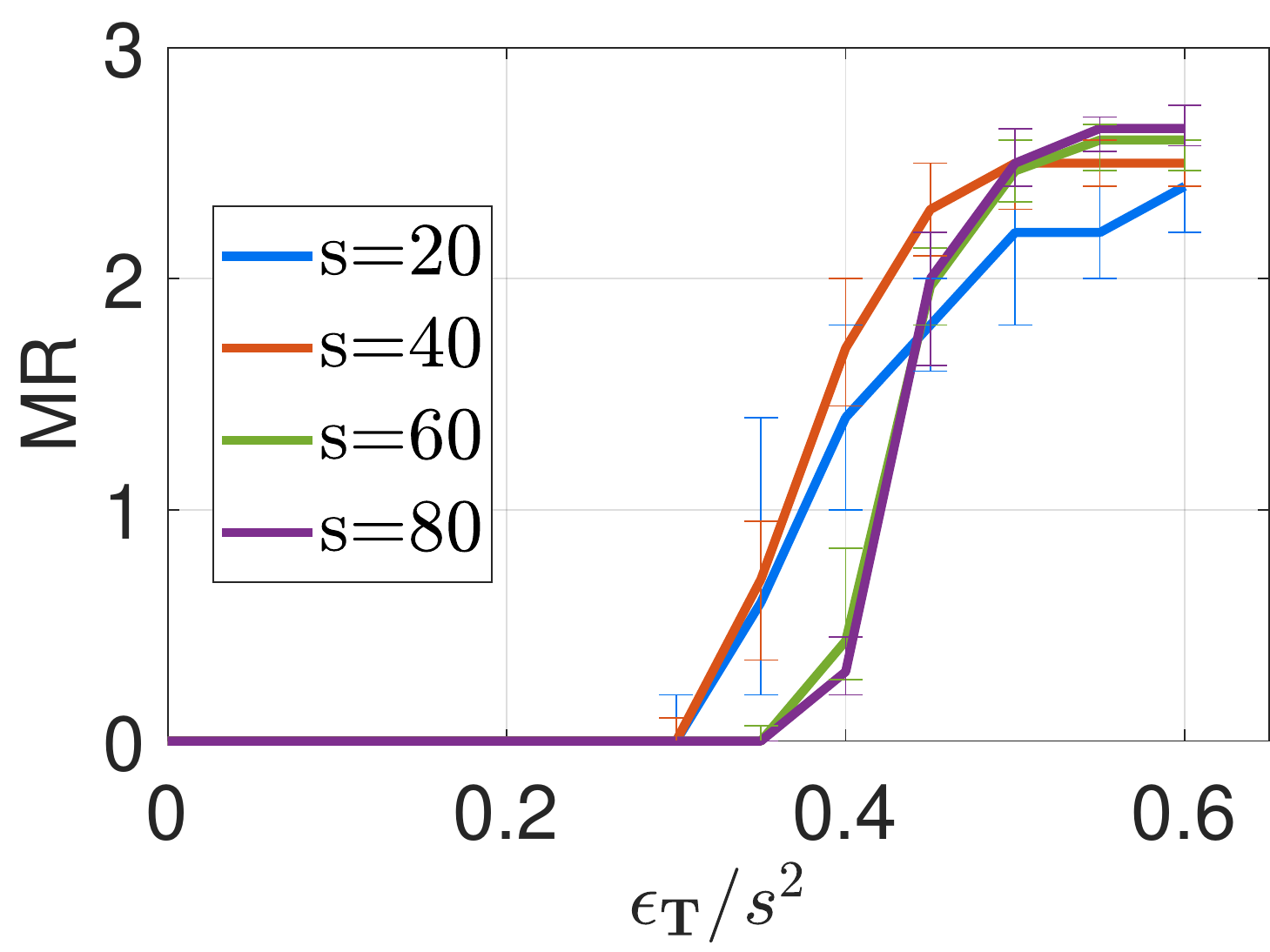}}
	\\
	\subfloat[]{%
		\includegraphics[width=1.6in]{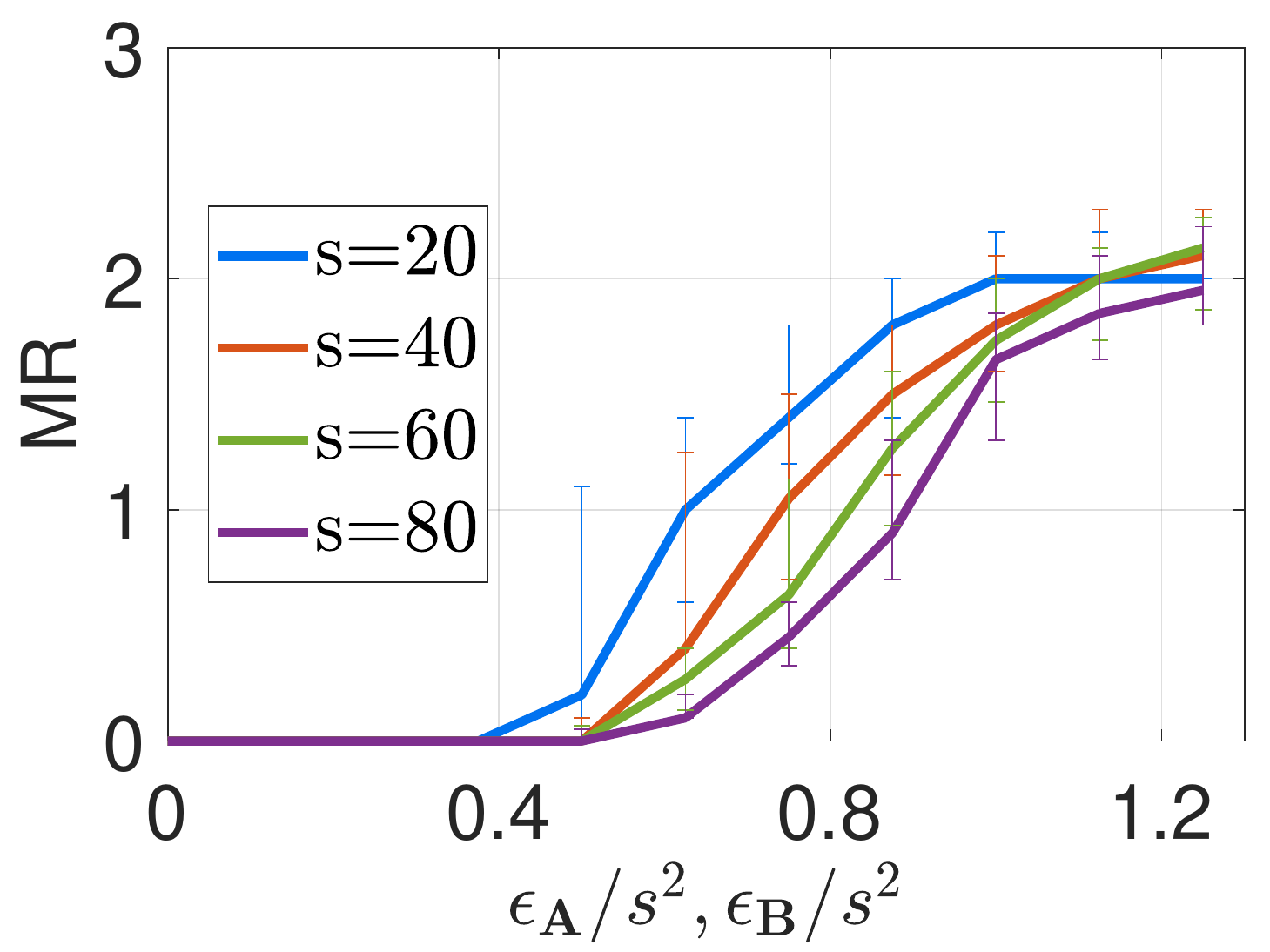}}
	\subfloat[]{%
		\includegraphics[width=1.6in]{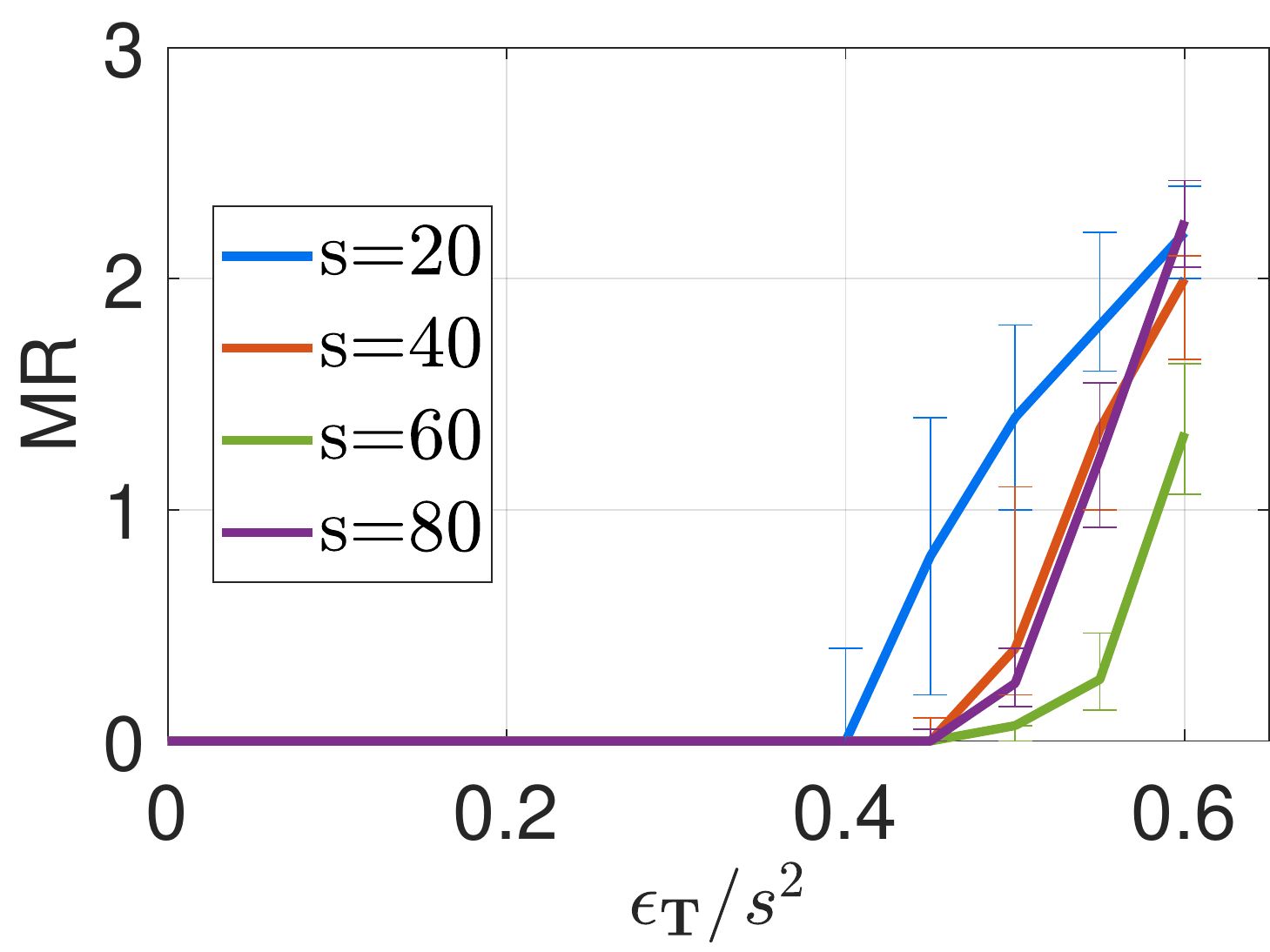}}
	\\	
	\caption{MR (median, first and third quartiles) vs number of modes $s$ vs perturbation level. First row: aggregatable case \ref{problem_aggregatableCase}. Second row: lumpable case \ref{problem_lumpableCase}. First column: $\epsilon_\vT$ ${=}0.5 \numSys^2$, $\alpha_\vA \hspace{-0.1ex} {\propto} 1 \hspace{-0.3ex} {/} \hspace{-0.4ex} \max_i \hspace{-0.5ex} \norm{\vA_i}$, $\alpha_\vB \hspace{-0.1ex} {\propto}1 \hspace{-0.3ex} {/} \hspace{-0.4ex} \max_i \hspace{-0.5ex} \norm{\vB_i}$, $\alpha_\vT \hspace{-0.1ex} {\propto} $ $ 0.01 \hspace{-0.3ex} {/} \hspace{-0.3ex} \norm{\vT}$. Second column: $\epsilon_\vA,\epsilon_\vB {=} 0.5 \numSys^2$, $\alpha_\vA {\propto} 0.01 \hspace{-0.3ex} {/} \hspace{-0.3ex} \max_i \hspace{-0.3ex} \norm{\vA_i}$, $\alpha_\vB {\propto} 0.01 \hspace{-0.3ex} {/} \hspace{-0.3ex} \max_i \hspace{-0.3ex} \norm{\vB_i}$, $\alpha_\vT {\propto} 1 \hspace{-0.3ex} {/} \hspace{-0.3ex} \norm{\vT}$.}
	\label{fig_ClusteringResult} 
\end{figure}

\subsection{LQR Controller Design}
Then, we implement the idea of designing LQR controllers for the original $\Sigma$ through the reduced $\hat{\Sigma}$ as discussed in Section \ref{sec_MJSReduction_LQRControl}. 
We let $\numCls=4$, $\dimSt=10$, $\dimInput=5$, and the system dynamics is generated the same way as the previous section. The process noise variance is $\sigma_\vw^2 = 0.1$, and the initial state is  $\vx_0 = \onevec$. Fig. \ref{subfig_LQR_suboptimality} shows the suboptimality against perturbations for $s=100$. As one would expect, the suboptimality increases with the perturbation levels and is $0$ when there is no perturbation. The trend on $\epsilon_\vT$ is evident when $\epsilon_\vA$ and $\epsilon_\vB$ are small but imperceptible 
for larger values of $\epsilon_\vA$ and $\epsilon_\vB$. Fig. \ref{subfig_LQR_compTime} shows the time to compute controllers via Riccati iterations using $\Sigma$ and $\hat{\Sigma}$ as a function of $s$. The computation terminates when the controller difference between two consecutive iterations falls below $10^{-12}$. We see when $\numSys$ is large, $\Sigma$ needs significantly more time than $\hat{\Sigma}$.

\begin{figure} 
	\centering
	\captionsetup[subfloat]{captionskip=0pt, farskip=0pt}
	\subfloat[]{%
		\includegraphics[width=1.6in]{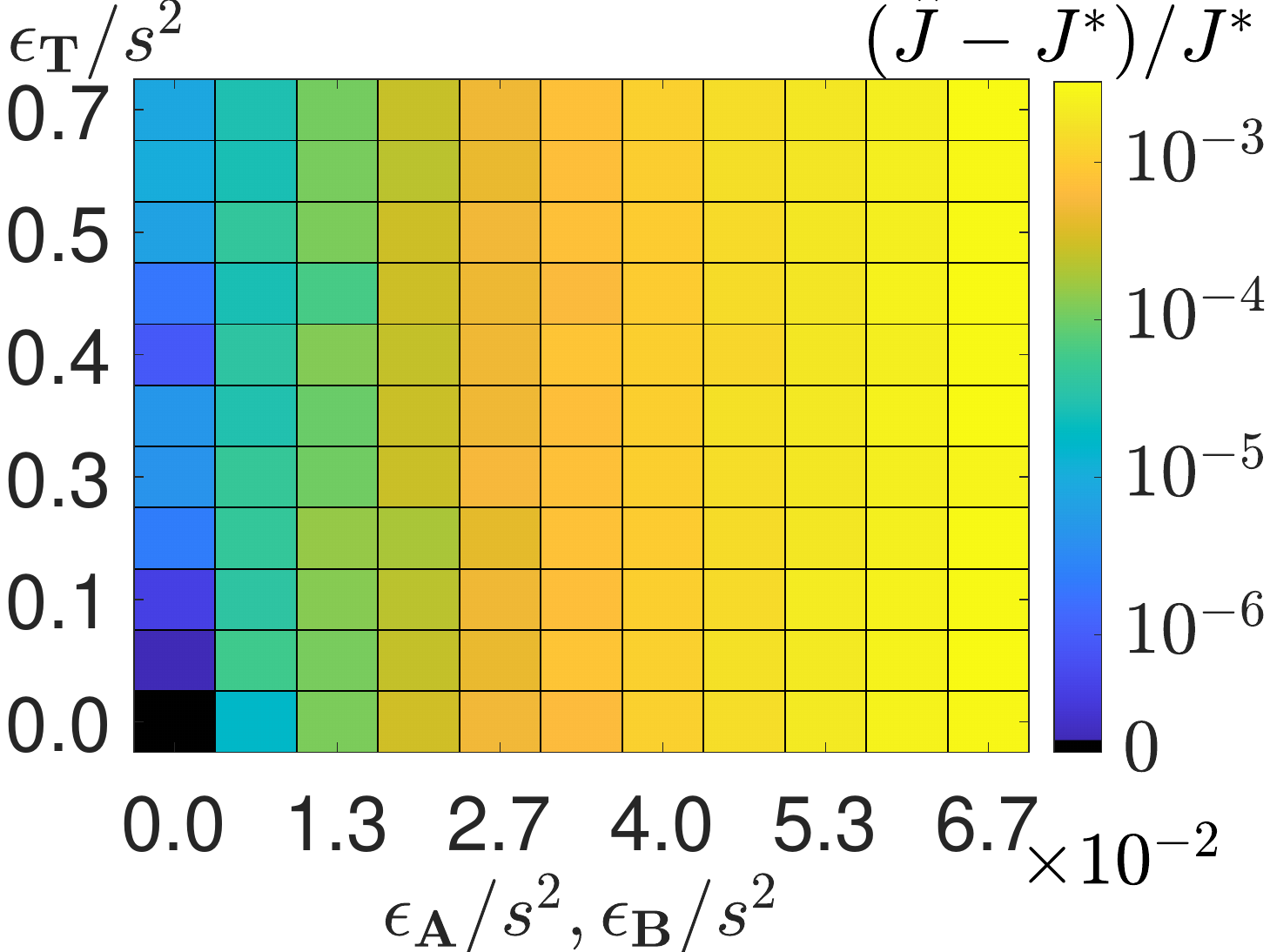}\label{subfig_LQR_suboptimality}}
	\hfill
	\subfloat[]{%
		\includegraphics[width=1.6in]{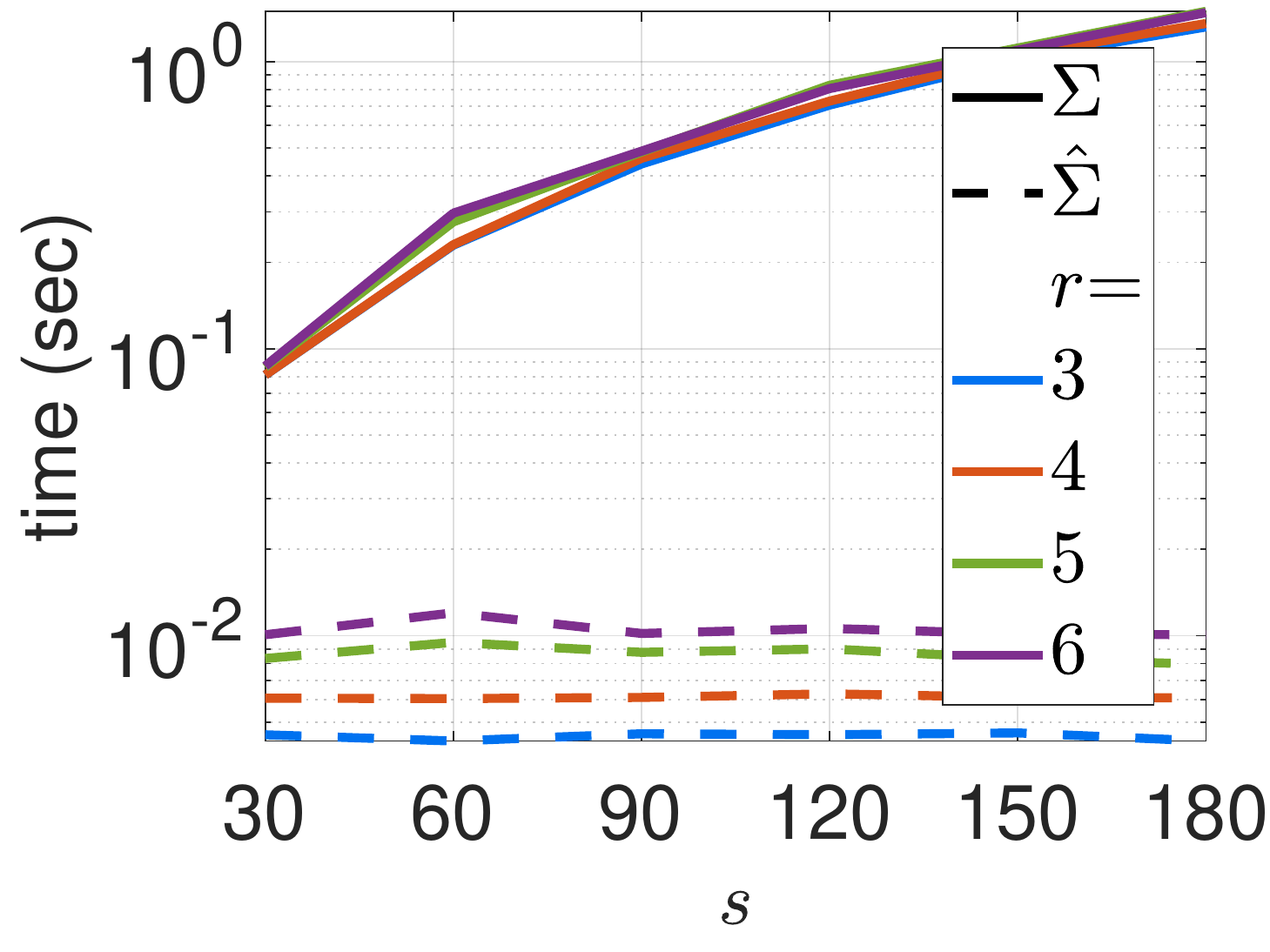}\label{subfig_LQR_compTime}}
	\\	
	\caption{(a) LQR suboptimality (median) vs perturbations; (b) LQR computation time (median) for the original MJS $\Sigma$ and the reduced MJS $\hat{\Sigma}$ with different number of modes and clusters. (We omit the quartiles as they are very close to the median.)}
	\label{fig_LQR} 
\end{figure}

Next, we consider a more practical scenario where one has no knowledge of the true number of cluster $\numCls$, and replace it in Algorithm \ref{Alg_MJSReduction} with a hyper-parameter $\hat{\numCls}$ as the number of modes in $\hat{\Sigma}$. We fix $\numSys = 100$ and $\numCls=30$, and the rest of the experiment setup is the same as Fig. \ref{fig_LQR}. We record the suboptimality and computation time under different choices of $\hat{\numCls}$ in Table \ref{Table_SuboptimalityandTimevsPickedCluster}. 
When increases $\hat{\numCls}$, the suboptimality achieves the minimum when $\hat{\numCls} = \numCls = 30$ and then gradually increases until $\hat{\numCls} = \numSys$, i.e. no reduction is performed at all. This comes as a bit of surprise as one would expect no worse performance when using more clusters than needed. Further investigation suggests that when $\hat{\numCls}>\numCls$, misclustering occurs more frequently than the case of $\hat{\numCls}=\numCls$, which is likely to account for the performance degradation. In practice, to find the best $\hat{\numCls}$, one could try multiple $\hat{\numCls}$ in Algorithm \ref{Alg_MJSReduction}, plug in the resulting partitions into \ref{problem_lumpableCase} and \ref{problem_aggregatableCase}, and select the one that gives the smallest perturbation $\epsilon_\vA, \epsilon_\vB, \epsilon_\vT$.

\begin{table}[ht]
\centering
\setlength\tabcolsep{4pt}
\caption{Suboptimality vs computation time vs selected number of modes in $\hat{\Sigma}$}\label{Table_SuboptimalityandTimevsPickedCluster}
{\renewcommand{\arraystretch}{1.2}
\begin{tabular}{||c | c | c | c | c | c ||} 
 \hline
 $\rhat$ & 10 & 20 & 30 & 40 & 50  \\ 
 \hline
 $\frac{\Jhat - J^\star}{J^\star}$ & $1.2e{\minus}1$ & $3.8e{\minus}2$ & $4.1e{\minus}7$ & $6.9e{\minus}5$ & $8.9e{\minus}4$  \\ \hline
 Time (sec) & $3.7e{\minus}2$ & $6.8e{\minus}2$ & $1.4e{\minus}1$ & $2.2e{\minus}1$ & $3.2e{\minus}1$ \\  \hline \hline
 $\rhat$ & 60 & 70 & 80 & 90 & 100 \\ \hline
 $\frac{\Jhat - J^\star}{J^\star}$ & $7.1e{\minus}3$ & $1.6e{\minus}2$ & $2.1e{\minus}2$ & $1.5e{\minus}2$ & 0 \\ \hline
 Time (sec) & $6.8e{\minus}1$ & $8.4e{\minus}1$ & $1.0e0$ & $1.1e0$ & $1.2e0$ \\  
 \hline
\end{tabular}}
\end{table}

\subsection{Trajectory Approximation}\label{subsec_experiments_trajDiff}
In this section, we evaluate the trajectory approximation results from Section \ref{sec_MJSReduction_approxMetrics}. Let $\theta = \pi/16$, $\breve{\vA}_1 = [[\cos(\theta), \sin(\theta)]^\T, [-\sin(\theta), \cos(\theta)]^\T]^\T$, 
$\breve{\vA}_2 = [[0.8, 0]^\T$
$, [0, 0.8]^\T]^\T$, and $\breve{\vA}_3 = [[1.2, 0]^\T, [0, 1.2]^\T]^\T$. Then, we construct an autonomous MJS $\Sigma$ with 6 modes: for $k = \curlybrackets{1,2,3}$, $\vA_{2k-1} = \breve{\vA}_k + [[0.1, 0]^\T, [0, 0.1]^\T]^\T$ and $\vA_{2k} = \breve{\vA}_k - [[0.1, 0]^\T, [0, 0.1]^\T]^\T$. The uniform partition $\{\{1,2\}, \{3,4\}, \{5,6\}\}$ gives $\epsilon_\vA = 0.6 \sqrt{2}$ according to \ref{problem_lumpableCase}. Define $\vT$ such that for all $i$, $\vT(i,j) = 0.2$ if $j \in \{1,2,3,4\}$ and $\vT(i,j) = 0.1$ if $j \in \{5, 6\}$. By relevant definitions in Section \ref{sec_MJSReduction_approxMetrics}, the constructed $\Sigma$ is MSS but not uniformly stable.

We fix the initial state $\vx_0 = [1,1]^\T$, generate $500$ independent trajectories for states $\vx_t$ and $\vxhat_t$, and record the difference $\norm{\vx_t - \vxhat_t}$. In Fig. \ref{fig_TrajDiff}, each thin solid line represents the difference, in log-scale, for each trajectory, the yellow dashed line shows their average, and the blue dashed line depicts the upper bound in Theorem \ref{thrm_MJSReduction_stateDistSyncMSS}. Throughout the time horizon, though not very tight, the theoretical upper bound stays above the averaged difference. Note that, for a given $\delta$, by Markov inequality, shifting the upper bound in the plot upward by $\log(\delta)$ would give a bound on the individual error trajectories with probability $1-\delta$. As seen in the figure, even the non-shifted version serves a good bound for individual error trajectories.

\begin{figure}[ht]
	\centering
	\includegraphics[width=2.2in]{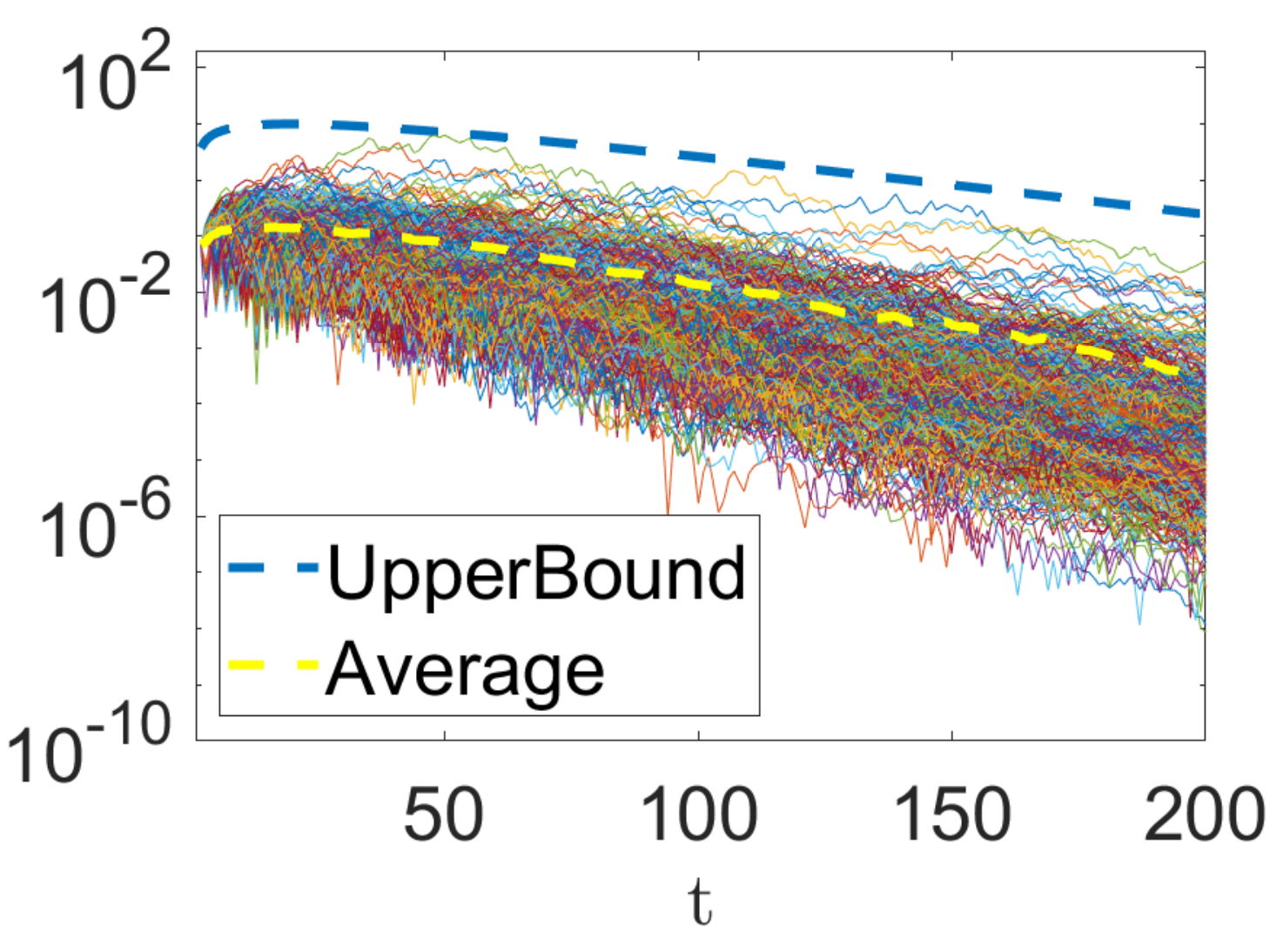}
	\caption{Trajectory difference and the upper bound}
	\label{fig_TrajDiff} 
\end{figure}

\section{Conclusion and Future Work}
In this work, we propose a clustering-based method to reduce the number of modes in an MJS. The reduced MJS provably well approximates the original MJS in terms of trajectory, transition kernels, stability, and controller optimality.
One future direction could be the generalization of the offline controller design scheme in Section \ref{sec_MJSReduction_LQRControl} to settings where controllers need to be computed in the runtime, such as model predictive control and adaptive control. In these problems, the savings of computation time would be even more prominent.
Another potential future direction could be the extension of the fully observed MJS in this work to partially observed MJS, i.e., the state $\vx_t$ is observed through $\vy_t = \vC_{\omega_t} \vx_t$ for some mode-dependent output matrices $\vC_{1:\numSys}$. As a side note, the similarity between MJS and Markov decision processes (MDP) hints that the framework and principles developed for MJS in this work may also help the complexity reduction of MDP and reinforcement learning problems. 

\appendices
\section{Aggregatable Clustering --- Proof for Theorem \ref{thrm_MJSReduction_AggreCase}}
We first provide several supporting lemmas. The first one is regarding the perturbation of the left singular vector space.
\begin{lemma}[Singular Vectors Perturbation Bound]\label{lemma_MJSReduction_singvecPerturb}
	Consider two arbitrary matrices $\bar{\vPhi}, \vPhi \in \dm{\numSys}{\numCls}$. Let $\vUbar, \vU \in \dm{\numSys}{\numCls}$ respectively denote the top-$\numCls$ left singular vectors of $\bar{\vPhi}$ and $\vPhi$ with $\vUbar^\T \vUbar = \vU^\T \vU = \vI_{\numCls}$. Then
	\begin{equation} \label{eq_MJSReduction_singvalPerturb}
		\min_{\vO \in \mathcal{O}(\numCls)} \norm{\vUbar \vO - \vU }_\fro 
		\leq 
		\frac{2 \sqrt{2} \norm{\bar{\vPhi} - \vPhi}_\fro}{\sigma_\numCls(\bar{\vPhi}) - \sigma_{\numCls+1}(\bar{\vPhi})},
	\end{equation}
	where $\mathcal{O}(\numCls)$ denotes the set of all $\numCls \x \numCls$ orthonormal matrices.
\end{lemma}
This result can be seen simply by combining Lemma 10 and Lemma 11 in \cite{du2019mode}, where Lemma 10 requires a trivial generalization from spectral norm to the Frobenius norm.
The next result says if a matrix has certain rows being identical, its singular vectors share the same identity pattern.
\begin{lemma}[Lemma 12 in \cite{du2019mode}]\label{lemma_MJSReduction_singvecIdentity}
	Consider a matrix $\bar{\vPhi} \in \dm{\numSys}{\numCls}$ and a partition $\Omega_{1:\numCls}$ on $[\numSys]$ such that for any $i, i' \in \Omega_k$, $\bar{\vPhi}(i,:) = \bar{\vPhi}(i',:)$. Assume $\rank(\bar{\vPhi}) = \numCls$. Let $\vUbar \in \dm{\numSys}{\numCls}$ denote the top-$\numCls$ left singular vectors of $\bar{\vPhi}$ with $\vUbar^\T \vUbar = \vI_{\numCls}$. Then for any $i \in \Omega_k$ and $j \in \Omega_l$, $\norm{\vUbar(i,:) - \vUbar(j,:)}=(\frac{1}{|\Omega_k|} + \frac{1}{|\Omega_l|})^{0.5}$ if $k \neq l$ and $0$ if $k=l$.
\end{lemma}

The next lemma provides a preliminary result on the performance of k-means when it is applied to a data matrix with feature dimension same as the number of clusters. 

\begin{lemma}[Lemma 5.3 in \cite{lei2015consistency}] \label{lemma_MJSReduction_kmeans1}
	Consider two arbitrary matrices $\vUbar, \vU \in \dm{\numSys}{\numCls}$ with $\Delta_\vU := \norm{\vUbar - \vU}_\fro$. Suppose there exists a partition $\Omega_{1:\numCls}$ on $[\numSys]$ such that for any $i,i' \in \Omega_k$, $\vUbar(i,:) = \vUbar(i',:)$.
	Define the inter-cluster distance for cluster $k$ as 
	$\delta_k:= \min_{l \in [\numCls]\backslash k}\min_{i \in \Omega_k, j \in \Omega_l} $ $\norm{\vUbar(i,:) - \vUbar(j,:)}$.
	Let $\curlybrackets{\hat{\Omega}_{1:\numCls}, \hat{c}_{1:\numCls}}$ be a $(1+\epsilon)$ solution to the k-means problem on the rows of $\vU$. Then, when $\Delta_\vU \leq \frac{\min_k \sqrt{|\Omega_{k}|} \delta_k}{\sqrt{8(2+\epsilon)}}$, we have
	\begin{equation}
		\min_{h \in \mathcal{H}} \sum_{k \in [\numCls]} |\curlybrackets{i: i \in \Omega_k, i \notin \hat{\Omega}_{h(k)}}| \cdot \delta_k^2 \leq 8(2+\epsilon) \Delta_\vU^2,
	\end{equation} 
	where $\mathcal{H}$ is the set of all bijections from $[\numCls]$ to $[\numCls]$.
\end{lemma}

By combining Lemma \ref{lemma_MJSReduction_singvecPerturb}, \ref{lemma_MJSReduction_singvecIdentity}, and \ref{lemma_MJSReduction_kmeans1}, we obtain guarantee on the performance of k-means when it is applied to the left singular vectors of the data matrix, which is the key lemma we will use to show Theorem \ref{thrm_MJSReduction_AggreCase} and Theorem \ref{thrm_MJSReduction_LumpCase}.
\begin{lemma}[Approximate k-means error bound] \label{lemma_MJSReduction_kmeans2}
	Consider two arbitrary matrices $\bar{\vPhi}, \vPhi \in \dm{\numSys}{\numCls}$ with $\Delta_{\vPhi} := \norm{\bar{\vPhi} - \vPhi}_\fro$. Suppose there exists a partition $\Omega_{1:\numCls}$ on $[\numSys]$ such that for any $i,i' \in \Omega_k$, $\bar{\vPhi}(i,:) = \bar{\vPhi}(i',:)$. Assume $\rank(\bar{\vPhi}) = \numCls$. Let $\vU \in \dm{\numSys}{\numCls}$ denote the top-$\numCls$ left singular vectors of $\vPhi$ with $\vU^\T \vU = \vI_{\numCls}$. Let $\curlybrackets{\hat{\Omega}_{1:\numCls}, \hat{c}_{1:\numCls}}$ be a $(1+\epsilon)$ solution to the k-means problem on clustering the rows of $\vU$. Then, when $\Delta_{\vPhi} \leq \frac{\sigma_\numCls(\bar{\vPhi}) \sqrt{|\Omega_{(\numCls)}| + |\Omega_{(1)}|}}{8 \sqrt{(2+\epsilon) |\Omega_{(1)}|}}$, we have
	$
		\textup{MR}(\hat{\Omega}_{1:\numCls}) \leq \frac{64(2+\epsilon)}{\sigma_{\numCls}(\bar{\vPhi})^2} \Delta_{\vPhi}^2.
	$
\end{lemma}
\vspace{0em}     
\begin{proof}
	Let $\vUbar {\in} \dm{\numSys}{\numCls}$ denote the top-$\numCls$ left singular vectors of $\bar{\vPhi}$ with $\vUbar^\T \vUbar $ ${=} \vI_{\numCls}$. Then, Lemma \ref{lemma_MJSReduction_singvecPerturb} implies that exists $\vO^\star {\in} $
	$ \mathcal{O}(\numCls)$ 
	such that 
	$
		\norm{\vUbar \vO^\star - \vU }_\fro 
		{\leq} 
		\frac{2 \sqrt{2} \Delta_\Phi}{\sigma_\numCls(\bar{\vPhi})}.
	$
	Note that 
	$
	\| [\vUbar \vO^\star] (i,:) - $ 	
	$ [\vUbar\vO^\star] (j,:) \|
	{=}
	\norm{(\vUbar(i,:) {-} \vUbar(j,:)) \vO^\star} 
	=
	\norm{\vUbar(i,:) {-}\vUbar(j,:)}$.
By Lemma \ref{lemma_MJSReduction_singvecIdentity}, we know for any $i {\in} \Omega_k$, $j {\in} \Omega_l$, 
$\|{[\vUbar \vO^\star] (i,:) -} [\vUbar \vO^\star] (j,:) \| {=} \sqrt{\frac{1}{|\Omega_k|} {+} \frac{1}{|\Omega_l|}}$ if $k {\neq} l$ and $0$ if $k {=} l$.
	Then, for any $k {\in} [\numCls]$, let $\delta_k{:=} \min_{l \in [\numCls]\backslash k}\min_{i \in \Omega_k, j \in \Omega_l} $ $ \norm{[\vUbar \vO^\star](i,:) - [\vUbar\vO^\star] (j,:)} $, we see 
	$\delta_k$ $\geq \sqrt{\frac{1}{|\Omega_k|} + \frac{1}{|\Omega_{(1)}|}}$.

    Note that when $\Delta_\vPhi {\leq} \frac{\sigma_\numCls(\bar{\vPhi}) \sqrt{|\Omega_{(\numCls)}| + |\Omega_{(1)}|}}{8 \sqrt{(2+\epsilon) |\Omega_{(1)}|}}$, one can check $\norm{\vUbar \vO^\star - \vU }_\fro \leq \frac{\min_k \sqrt{|\Omega_{k}|} \delta_k}{\sqrt{8(2+\epsilon)}}$. Then, by Lemma \ref{lemma_MJSReduction_kmeans1}, we have $\textup{MR}(\hat{\Omega}_{1:\numCls}) 
	= \min_{h \in \mathcal{H}} \sum_{k=1}^{\numCls} |\curlybrackets{i: i \in \Omega_k, i \notin \hat{\Omega}_{h(k)}}| \frac{1}{|\Omega_k|}
	\leq \min_{h \in \mathcal{H}} \sum_{k=1}^{\numCls} |\curlybrackets{i: i \in \Omega_k, i \notin \hat{\Omega}_{h(k)}}| \delta_k^2 \leq 64(2+\epsilon) \sigma_\numCls(\bar{\vPhi})^{\invv{2}} \Delta_{\vPhi}^2$.
\end{proof}

\begin{proof}[Main Proof for Theorem \ref{thrm_MJSReduction_AggreCase}]
	Consider $\vPhi$ in Algorithm \ref{Alg_MJSReduction} Line \ref{algline_MJSReduction_7} and its averaged version $\bar{\vPhi}$ defined in Section \ref{subsec_MJSReduction_theoryClustering}. 	
    Then, by definition, we have   
	$
		\norm{\bar{\vPhi} {-} \vPhi}_\fro^2
		=
		\sum_{k \in [\numCls]} \sum_{i \in \Omega_k}
		\norm{\bar{\vPhi}(i,:) {-} \vPhi(i,:) }_\fro^2
		=
		\alpha_\vT^2 \cdot \sum_{k \in [\numCls]} \sum_{i \in \Omega_k}  \| \vT(i,:) - |\Omega_k|^\inv \cdot \sum_{i' \in \Omega_k} \vT(i',:) \|^2 $
		$ 
		+ \alpha_\vA^2 \cdot \sum_{k \in [\numCls]} \sum_{i \in \Omega_k}  \|\vA_{i} - |\Omega_k|^\inv
		\cdot \sum_{i' \in \Omega_k} \vA_{i'} \|_\fro^2 
		$
		$
		+ \alpha_\vB^2 \cdot \sum_{k \in [\numCls]} \sum_{i \in \Omega_k}  \| \vB_{i} - |\Omega_k|^\inv \sum_{i' \in \Omega_k} \vB_{i'} \|_\fro^2 .
	$
    By the definitions of $\epsilon_\vA, \epsilon_\vB, \epsilon_\vT$ in Problem \ref{problem_aggregatableCase}, triangle inequality ,and Cauchy-Schwarz inequality, we have
	$
		\norm{\bar{\vPhi} - \vPhi}_\fro 
		\leq \epsilon_{Agg}
	$
	where $\epsilon_{Agg} := \sqrt{ \alpha_\vA^2 \epsilon_\vA^2 + \alpha_\vB^2 \epsilon_\vB^2 + \alpha_\vT^2 \epsilon_\vT^2 }$.
	By construction, in matrix $\bar{\vPhi}$, rows that belong to the same cluster are identical, thus we can apply Lemma \ref{lemma_MJSReduction_kmeans2} to $\curlybrackets{\bar{\vPhi}, \vPhi}$ and obtain that when $\epsilon_{Agg} \leq \frac{\sigma_\numCls(\bar{\vPhi}) \sqrt{|\Omega_{(\numCls)}| + |\Omega_{(1)}|}}{8 \sqrt{(2+\epsilon) |\Omega_{(1)}|}}$, we have
	$\textup{MR}(\hat{\Omega}_{1:\numCls}) \leq 64(2+\epsilon) \sigma_\numCls(\bar{\vPhi})^{\invv{2}} \epsilon_{Agg}^2$.
\end{proof}

\section{Lumpable Clustering --- Proof for Theorem \ref{thrm_MJSReduction_LumpCase}} \label{appendix_lumpable}
We first provide a supporting result regarding the perturbation of stationary distribution of Markov chains.
\begin{lemma}[Section 3.6 in \cite{cho2001comparison}]\label{lemma_MJSReduction_MCStationaryDistPerturbation}
	For two Markov matrices $\vT, \vT_0 {\in} \dm{\numSys}{\numSys}$ and their stationary distributions $\vpi,\vpi_0 {\in} \dm{\numSys}{1}$, 
	we have
	$
	\norm{\vpi {-} \vpi_0}_1 {\leq} \gamma_1 \norm{\vT{-}\vT_0}_\infty,
	$
	where 	
	$\gamma_1 {:=} \sum_{i=2}^{\numSys} $	
	$\frac{1}{1-\lambda_i(\vT)}$.
\end{lemma}
When the difference $\norm{\vT - \vT_0}$ is small, we can further have the following corollary.
\begin{corollary}\label{corollary_MJSReduction_MCStationaryDistPerturbation}
	In Lemma \ref{lemma_MJSReduction_MCStationaryDistPerturbation}, let $\pi_{\min}:= \min_i \vpi(i), \pi_{\max}:= \max_i \vpi(i)$. Suppose $\norm{\vT - \vT_0}_\infty \leq \frac{\pi_{\min}}{\gamma_1}$, then we have
	\begin{gather}
		\max_i |\vpi(i) - \vpi_0(i)| \leq \frac{\pi_{\min}}{2}, \quad
		\label{eq_MJSReduction_29_1}\\
		\min_i \vpi_0(i) \geq \frac{\pi_{\min}}{2}, \quad \max_i \vpi_0(i) \leq \pi_{\max} + \frac{\pi_{\min}}{2} \label{eq_MJSReduction_29}\\
		\hspace{-0.4em}\max_i |\vpi(i)^{\minus \frac{1}{2}} {-} \vpi_0(i)^{\minus\frac{1}{2}}|  \leq (\sqrt{2}{-}1) \gamma_1 \pi_{\min}^{\minus\frac{3}{2}} \norm{\vT {-} \vT_0}_\infty \label{eq_MJSReduction_30}\\
		\max_i |\vpi(i)^{\frac{1}{2}} {-} \vpi_0(i)^{\frac{1}{2}}|  \leq	(1-\frac{\sqrt{2}}{2}) \gamma_1 \pi_{\min}^{\invv{\frac{1}{2}}} \norm{\vT {-} \vT_0}_\infty \label{eq_MJSReduction_31}.
	\end{gather}	
\end{corollary}
\begin{proof}
	Since $\onevec^\T \vpi = \onevec^\T \vpi_0 = 1$, we have $\max_i |\vpi(i) - \vpi_0(i)| \leq \frac{1}{2} \norm{\vpi - \vpi_0}_1 \leq \frac{\gamma_1}{2} \norm{\vT - \vT_0}_\infty \leq \frac{\pi_{\min}}{2} $. Then using triangle inequality, we can show \eqref{eq_MJSReduction_29_1} and \eqref{eq_MJSReduction_29}. Note that the LHS of \eqref{eq_MJSReduction_30} is equivalent to $\max_i \frac{|\vpi_0(i) - \vpi(i)|}{\sqrt{\vpi(i) \vpi_0(i)} (\sqrt{\vpi(i)} + \sqrt{\vpi_0(i)}) }$, then plugging in \eqref{eq_MJSReduction_29} gives \eqref{eq_MJSReduction_30}. 
	And \eqref{eq_MJSReduction_31} follows similarly.
\end{proof}

When the lumpability perturbation $\epsilon_\vT \neq 0$, matrix $\vS_\numCls$ in Algorithm \ref{Alg_MJSReduction} Line \ref{algline_MJSReduction_1} no longer has the row identity pattern as discussed in Lemma \ref{lemma_MJSReduction_lumpableInformativeSpectrum_shorVer}. The next result measures this effect.

\begin{lemma}\label{lemma_MJSReduction_SminusSbar}
	Consider an ergodic Markov matrix $\vT \in \dm{\numSys}{\numSys}$ with stationary distribution $\vpi$ and a partition $\Omega_{1:\numCls}$ such that it is approximately lumpable as in \eqref{eq_MJSReduction_approxLump} with perturbation $\epsilon_\vT$. 
	Consider the neighborhood of $\vT$ given by $\Lcal(\vT, \Omega_{1:\numCls}, \epsilon_\vT)$ defined in \eqref{eq_NeighborhoodT}.
	Assume there exists an ergodic and reversible $\vT_0 \in \Lcal(\vT, \Omega_{1:\numCls}, \epsilon_\vT)$ that has informative spectrum. Construct $\vS_\numCls \in \dm{\numSys}{\numCls}$ with $\vT$ and $\vpi$ as in Algorithm \ref{Alg_MJSReduction} Line \ref{algline_MJSReduction_1}.
	Construct $\vSbar_\numCls \in \dm{\numSys}{\numCls}$ such that for any $i \in [\numSys]$ (suppose $i \in \Omega_k$), $\vSbar_\numCls(i,:) = \frac{1}{|\Omega_k|} \sum_{i' \in \Omega_k} \vS_\numCls(i',:)$.
	Let $\pi_{\min}:=\max_i \vpi(i)$, $\pi_{\max}:=\min_i \vpi(i)$, $\gamma_1 := \sum_{i=2}^{\numSys} \frac{1}{1-\lambda_i(\vT)}$, $\gamma_2:= \min\curlybrackets{ \sigma_\numCls(\vH) - \sigma_{\numCls+1}(\vH), 1}$, and $\gamma_3 := \frac{16 \gamma_1 \sqrt{\numCls \pi_{\max}} \norm{\vT}_\fro}{\gamma_2 \pi_{\min}^2}$ where $\vH$ is defined in Algorithm \ref{Alg_MJSReduction}.
	Then, when perturbation $\epsilon_\vT \leq \frac{\pi_{\min}}{\gamma_1}$, we have
	$
		\norm{ \vS_\numCls - \vSbar_{\numCls} }_\fro \leq \gamma_3 \epsilon_\vT.
	$\
\end{lemma}
\begin{proof}
	We will start with analyzing $\vT_0$ and use it as a bridge to prove the claim. Let $\vpi_0 \in \dm{\numSys}{1}$ denote the stationary distribution of $\vT_0$. Since $\vT_0$ is ergodic, we know $\vpi_0$ is strictly positive. By definition of reversibility, we know $\diag(\vpi_0) \vT_0 = \vT_0^\T \diag(\vpi_0)$, and this further gives $\diag(\vpi_0)^\frac{1}{2} \vT_0 \diag(\vpi_0)^{\invv{\frac{1}{2}}} = \diag(\vpi_0)^{\invv{\frac{1}{2}}} \vT_0^\T \diag(\vpi_0)^\frac{1}{2}$. Let $\vH_0 := \diag(\vpi_0)^\frac{1}{2} \vT_0 \diag(\vpi_0)^{\invv{\frac{1}{2}}}$, then we see $\vH_0$ is symmetric. Let $\vW_{0, \numCls} \in \dm{\numSys}{\numCls}$ denote the top $\numCls$ left singular vectors of $\vH_0$, by spectral theorem, we know the columns of $\vW_{0, \numCls}$ also serve as the top $\numCls$ eigenvectors of $\vH_0$. Let $\vS_{0,\numCls}: = \diag(\vpi_0)^{\invv{\frac{1}{2}}} \vW_{0, \numCls}$, by definition of $\vH_0$, it is easy to see that the columns of $\vS_{0,\numCls}$ are also the top $\numCls$ eigenvectors of $\vT_0$. Then, by Lemma \ref{lemma_MJSReduction_lumpableInformativeSpectrum_shorVer} and the definition of informative spectrum, for any $i, i' \in \Omega_k$, we have $\vS_{0,\numCls}(i,:) = \vS_{0,\numCls}(i',:)$.
	
	Recall in Algorithm \ref{Alg_MJSReduction}, $\vW_\numCls$ denotes the top $\numCls$ left singular vectors of $\vH:= \diag(\vpi)^\frac{1}{2} \vT \diag(\vpi)^{ \invv{\frac{1}{2}}}$ and let $\vS_\numCls = \diag(\vpi)^{\invv{\frac{1}{2}}} \vW_\numCls$. Let $\vO^\star := \min_{\vO \in \mathcal{O}(\numCls)} \norm{\vW_{0, \numCls} \vO - \vW_\numCls}_\fro$, where $\mathcal{O}(\numCls)$ is the set of all $\numCls \times \numCls$ orthonormal matrices. Then, for any $i, i' \in \Omega_k$, we have $[\vS_{0,\numCls} \vO^\star] (i,:) = [\vS_{0,\numCls} \vO^\star] (i',:)$. Using this, for any $i \in [\numSys]$ (suppose $i \in \Omega_k$), we have
	\begin{equation}\label{eq_MJSReduction_32}
		\begin{split}
			&\vS_\numCls(i,:) - \vSbar_{\numCls}(i,:) \\
			=& \frac{|\Omega_k|-1}{|\Omega_k|} \vS_\numCls(i,:) -  \frac{1}{|\Omega_k|} \sum_{i': i' \in \Omega_k, i'\neq i} \vS_\numCls(i',:)\\
			\leq &\frac{|\Omega_k|-1}{|\Omega_k|} (\vS_\numCls(i,:) - [\vS_{0,\numCls} \vO^\star] (i,:)) \\
			&+ \frac{1}{|\Omega_k|} \sum_{i': i' \in \Omega_k, i'\neq i} \big( [\vS_{0,\numCls} \vO^\star] (i',:) - \vS_\numCls(i',:) \big).
		\end{split}
	\end{equation}
	WLOG, assume $\curlybrackets{1, \dots, |\Omega_1|} = \Omega_1$, $\curlybrackets{|\Omega_1|+1, \dots, |\Omega_1|+|\Omega_2|} = \Omega_2, \cdots$ and define block diagonal matrices $\vD, \vP \in \dm{\numSys}{\numSys}$ both with $\numCls$ diagonal blocks such that their $k$-th diagonal blocks $[\vD]_k, [\vP]_k \in \dm{|\Omega_k|}{|\Omega_k|}$ are given by
	\begin{equation}
		[\vD]_k {=} \frac{|\Omega_k|{-}1}{|\Omega_k|} \vI_{|\Omega_k|}, \ [\vP]_k {=} \frac{1}{|\Omega_k|} (\onevec_{|\Omega_k|} \onevec_{|\Omega_k|}^\T {-} \vI_{|\Omega_k|}).
	\end{equation}
	Then, stacking \eqref{eq_MJSReduction_32} for all $i$, one can verify that
	$
	\vS_\numCls- \vSbar_{\numCls}
	= \vD (\vS_\numCls - \vS_{0,\numCls} \vO^\star) + \vP (\vS_{0,\numCls} \vO^\star - \vS_\numCls)
	$.
	Note that for an arbitrary matrix $\vE \in \dm{\numSys}{\numSys}$, 	
	we have $\norm{\vP \vE}_\fro^2 = \tr(\vP^\T \vP \vE \vE^\T) \leq \tr(\vD^\T$
	$ \cdot \vD \vE \vE^\T) = \norm{\vD \vE}_\fro^2$ where the inequality holds since for each diagonal block we have
	$[\vP]_k^\T [\vP]_k \preceq [\vD]_k^\T [\vD]_k$. 
	Therefore, 	
	$
	\norm{ \vS_\numCls {-} \vSbar_{\numCls} }_\fro
			\leq 2\norm{\vD (\vS_\numCls {-} \vS_{0,\numCls} \vO^\star)}_\fro 
			\leq
			2 \max_k \frac{|\Omega_k|{-}1}{|\Omega_k|} \|\vS_\numCls - 
	$
	$  \vS_{0,\numCls} \vO^\star|_\fro 
			\leq 2 \norm{\vS_\numCls - \vS_{0,\numCls} \vO^\star}_\fro.
	$
	To complete the proof, it suffices to study	$\norm{\vS_\numCls - \vS_{0,\numCls} \vO^\star}_\fro$.	
	\begin{equation}\label{eq_MJSReduction_24}{\hspace*{-0.5em}
		\begin{aligned}
			&\norm{\vS_\numCls - \vS_{0,\numCls} \vO^\star}_\fro \\
			=& \| \diag(\vpi)^{\invv{\frac{1}{2}}} (\vW_\numCls - \vW_{0, \numCls} \vO^\star) \\
			& + (\diag(\vpi)^{\invv{\frac{1}{2}}} - \diag(\vpi_0)^{\invv{\frac{1}{2}}})\vW_{0, \numCls}\vO^\star \|_\fro \\
			\leq& \pi_{\min}^{\invv{0.5}} \norm{\vW_\numCls {-} \vW_{0, \numCls} \vO^\star}_\fro 
			{+} \sqrt{\numCls} \max_i |\vpi(i)^{\invv{\frac{1}{2}}} {-} \vpi_0(i)^{\invv{\frac{1}{2}}}|.
		\end{aligned}}
	\end{equation}
	According to Lemma \ref{lemma_MJSReduction_singvecPerturb}, we know $\norm{\vW_\numCls - \vW_{0, \numCls} \vO^\star}_\fro 
	\leq \frac{2 \sqrt{2}}{ \sigma_\numCls(\vH) - \sigma_{\numCls+1}(\vH)} \norm{\vH - \vH_0}_\fro$. This together with the upper bound for $\max_i |\vpi(i)^{\invv{\frac{1}{2}}} - \vpi_0(i)^{\invv{\frac{1}{2}}}|$ in \eqref{eq_MJSReduction_30} gives
	\begin{multline}\label{eq_MJSReduction_28}
		\norm{\vS_\numCls - \vS_{0,\numCls} \vO^\star}_\fro
		\leq \frac{2 \sqrt{2}}{(\sigma_\numCls(\vH) - \sigma_{\numCls+1}(\vH)) \pi_{\min}^{0.5}} \norm{\vH - \vH_0}_\fro \\
		+ \frac{ (\sqrt{2} - 1 ) \gamma_1 \sqrt{\numCls} }{\pi_{\min}^{1.5} } \epsilon_\vT.
	\end{multline}	
	By the definitions of $\vH$ and $\vH_0$, we have
	$
	\norm{\vH {-} \vH_0}_\fro \leq $
	$	\|(\diag(\vpi)^{\frac{1}{2}} {-} \diag(\vpi_0)^\frac{1}{2}) \cdot \vT \cdot\diag(\vpi)^{\invv{\frac{1}{2}}}\|_\fro + 
	\|\diag(\vpi_0)^\frac{1}{2} $	
	$\cdot \vT \cdot (\diag(\vpi)^{\invv{\frac{1}{2}}} {-} \diag(\vpi_0)^{-\frac{1}{2})} \|_\fro
	+
	\|\diag(\vpi_0)^\frac{1}{2} \cdot (\vT {-} \vT_0) 
	$	
	$
	 \cdot \diag(\vpi_0)^{\invv{\frac{1}{2}}} \|_\fro.
	$	
	Applying Corollary \ref{corollary_MJSReduction_MCStationaryDistPerturbation} gives
	$
	\norm{\vH - \vH_0}_\fro
	\leq
	2.56 \gamma_1 \pi_{\max}^{0.5} \pi_{\min}^{-1.5} \norm{\vT}_\fro \epsilon_\vT.
	$
    Plugging this into \eqref{eq_MJSReduction_28}, we have
    $
    \norm{\vS_\numCls - \vS_{0,\numCls} \vO^\star}_\fro
    \leq
    \frac{8 \gamma_1 \sqrt{\numCls} \sqrt{\pi_{\max}} \norm{\vT}_\fro}{\gamma_2 \pi_{\min}^2} \epsilon_\vT,
    $
	where $\gamma_2:= \min\curlybrackets{ \sigma_\numCls(\vH) - \sigma_{\numCls+1}(\vH), 1}$. This concludes the proof as we showed that $\norm{ \vS_\numCls - \vSbar_{\numCls} }_\fro \leq 2 \norm{\vS_\numCls - \vS_{0,\numCls} \vO^\star}_\fro$.
\end{proof}

\begin{proof}[Main Proof for Theorem \ref{thrm_MJSReduction_LumpCase}]
	Consider $\vPhi$ in Algorithm \ref{Alg_MJSReduction} Line \ref{algline_MJSReduction_13} and its averaged version $\bar{\vPhi}$ defined in Section \ref{subsec_MJSReduction_theoryClustering}. 
	Then, by definition, we have   
	$
		\norm{\bar{\vPhi} {-} \vPhi}_\fro^2
		=
		\alpha_\vT^2 \cdot \norm{\vS_\numCls {-} \vSbar_\numCls}_\fro^2 $
		$ 
		+ \alpha_\vA^2 \cdot \sum_{k \in [\numCls]} \sum_{i \in \Omega_k}  \|\vA_{i} - |\Omega_k|^\inv
		\cdot \sum_{i' \in \Omega_k} \vA_{i'} \|_\fro^2 
		$
		$
		+ \alpha_\vB^2 \cdot \sum_{k \in [\numCls]} \sum_{i \in \Omega_k}  \| \vB_{i} - |\Omega_k|^\inv \sum_{i' \in \Omega_k} \vB_{i'} \|_\fro^2 .
	$
	where $\vSbar_\numCls$ is defined in Lemma \ref{lemma_MJSReduction_SminusSbar}. 
	By Lemma \ref{lemma_MJSReduction_SminusSbar} and the definitions of $\epsilon_\vA$ and $\epsilon_\vB$ in Problem \ref{problem_lumpableCase}, we have
	$
		\norm{\bar{\vPhi} - \vPhi}_\fro $ $
		\leq \epsilon_{Lmp}
	$
	where $\epsilon_{Lmp}: = \sqrt{\alpha_\vA^2 \epsilon_\vA^2 + \alpha_\vB^2 \epsilon_\vB^2 + \alpha_\vT^2 \gamma_3^2 \epsilon_\vT^2}$.
	By construction, in $\bar{\vPhi}$, rows that belong to the same cluster have the same rows, thus we can apply Lemma \ref{lemma_MJSReduction_kmeans2} to $\curlybrackets{\bar{\vPhi}, \vPhi}$ and obtain that when $\epsilon_{Lmp} {\leq} \frac{\sigma_\numCls(\bar{\vPhi}) \sqrt{|\Omega_{(\numCls)}| + |\Omega_{(1)}|}}{8 \sqrt{(2+\epsilon) |\Omega_{(1)}|}}$, we have
	$\textup{MR}(\hat{\Omega}_{1:\numCls}) {\leq} 64(2+\epsilon) \sigma_\numCls(\bar{\vPhi})^{\invv{2}} \epsilon_{Lmp}^2$.
\end{proof}

\subsection{Non-emptiness of \texorpdfstring{$\Lcal(\vT, \Omega_{1:\numCls}, \epsilon_\vT)$}{L}}
Note that both Lemma \ref{lemma_MJSReduction_SminusSbar} and Theorem \ref{thrm_MJSReduction_LumpCase} require the set $\Lcal(\vT, \Omega_{1:\numCls}, \epsilon_\vT)$, a neighborhood of $\vT$. Now, we show it is non-empty under the approximate lumpability condition \eqref{eq_MJSReduction_approxLump}.

Let $\vT_0:=\vT + \vDelta$ for $\vDelta \in \Dcal$ where
\begin{align}
\Dcal:=  \Big\{\vDelta{\in} \dm{\numSys}{\numSys} :  \forall k,l {\in} [\numCls], \forall i {\in} \Omega_k, \hspace{7em} \nonumber \\
-\vT(i,j) \leq \vDelta(i,j) \leq 1-\vT(i,j) \quad \forall j \in [\numSys], \label{eq_nonempty_1}\\
\sum_{j \in \Omega_l} \vDelta(i,j) = - \sum_{j \in \Omega_l} \vT(i,j) + |\Omega_k|^\inv \sum_{\substack{i' \in \Omega_k\\j \in \Omega_l}} \vT(i', j), \label{eq_nonempty_2}\\
\norm{\vDelta}_\fro \leq \epsilon_\vT, \quad \norm{\vDelta}_\infty \leq \epsilon_\vT. \hspace{1em}\Big\}  \hspace{0em} \nonumber
\end{align}
Then, we see to show there exists $\vT_0 \in \Lcal(\vT, \Omega_{1:\numCls}, \epsilon_\vT)$, i.e. $\Lcal(\vT, \Omega_{1:\numCls}, \epsilon_\vT)$ is non-empty, it is equivalent to show there exists $\vDelta \in \Dcal$. 

Note that \eqref{eq_nonempty_1} gives that for all $i \in \Omega_k, l \in [\numCls]$, $- \sum_{j \in \Omega_l} \vT(i,j) \leq \sum_{j \in \Omega_l} \vDelta(i,j) \leq |\Omega_l| - \sum_{j \in \Omega_l} \vT(i,j)$. This together with \eqref{eq_nonempty_1} and \eqref{eq_nonempty_2} imply that there exists $\vDelta$ satisfying both \eqref{eq_nonempty_1} and \eqref{eq_nonempty_2} such that among its elements $\curlybrackets{\vDelta(i,j)}_{j \in \Omega_l}$, the nonzero ones have the same signs as the RHS of \eqref{eq_nonempty_2}. Then, for all $i \in \Omega_k, l \in [\numCls]$, we have $\sum_{j \in \Omega_l} |\vDelta(i,j)| = |\sum_{j \in \Omega_l} \vDelta(i,j)| = |\text{RHS of \eqref{eq_nonempty_2}}| \leq |\Omega_k|^\inv \sum_{i' \in \Omega_k} |\sum_{j \in \Omega_l} \vT(i,j) - \sum_{j \in \Omega_l} \vT(i', j)|$. This further gives $\norm{\vDelta}_\fro \leq \sum_{k, l\in [\numCls]} \sum_{i \in \Omega_k, j \in \Omega_l} |\vDelta(i,j)| \leq |\Omega_k|^\inv \sum_{k,l {\in} [\numCls]} \sum_{i, i' {\in} \Omega_k} |\sum_{j {\in} \Omega_l} \vT(i,j) - \sum_{j \in \Omega_l} \vT(i', j)|$ $ \leq |\Omega_k|^\inv \epsilon_\vT$, where the last inequality follows from \eqref{eq_MJSReduction_approxLump}. These steps also show $\norm{\vDelta}_\infty \leq \epsilon_\vT$. We have shown $\vDelta \in \Dcal$, i.e. $\Dcal$ is non-empty, and so is $\Lcal(\vT, \Omega_{1:\numCls}, \epsilon_\vT)$.

\section{Approximation with MSS --- Proof for Theorem \ref{thrm_MJSReduction_stateDistSyncMSS}}
We first provide several supporting results regarding the perturbation of matrix product.
\begin{lemma}\label{lemma_MJSReduction_JSRPerturb}
 	Consider two sets of matrices $\vA_1, \dots, \vA_\numSys$ and $\vAhat_1, \dots, \vAhat_\numSys$ with $\norm{\vA_i - \vAhat_i} \leq \epsilon_\vA$ for all $i \in [\numSys]$. Assume there exists a pair $\curlybrackets{\xi, \kappa}$ such that for all $t \in \mathbb{N}$, $\max_{\sigma_{1:t} \in [\numSys]^{t}}\norm{\vA_{\sigma_1} \cdots \vA_{\sigma_t}}^{\frac{1}{t}} \leq \kappa \cdot \xi^t$. Then, for all $t$ and any sequence $\sigma_{1:t} \in [\numSys]^{t}$, we have (i) $\norm{\prod_{h=1}^t \vAhat_{\sigma_h}} \leq \kappa (\kappa \epsilon_\vA + \xi)^t$; (ii) $\norm{\prod_{h=1}^t \vAhat_{\sigma_h} - \prod_{h=1}^t \vA_{\sigma_h}} \leq \kappa^2 t (\kappa \epsilon_\vA + \xi)^{t-1} \epsilon_\vA$.
\end{lemma}
\begin{proof}
 	Let $\vE_i {:=} \vAhat_i - \vA_i$, then we see $\norm{\vE_i} {\leq} \epsilon_\vA$ and $\prod_{h=1}^{t} \vAhat_{\sigma_h} = $
 	$\prod_{h=1}^{t}(\vA_{\sigma_h} + \vE_{\sigma_h})$. In the expansion of $\prod_{h=1}^{t}(\vA_{\sigma_h} + \vE_{\sigma_h})$,  	
 	for each $i=0, 1, \dots, t$, there are $\binom{t}{i}$ terms, each of which is a product where $\vE_{\sigma_h}$ has degree $i$ and $\vA_{\sigma_h}$ has degree $t-i$. We let $\vF_{i,j}$ with $i = 0,1,\dots, t$ and $j \in [\binom{t}{i}]$ to index these expansion terms. Note that $\norm{\vF_{i,j}} \leq \kappa^{i+1} \xi^{t-i} \epsilon_\vA^i$. Then, we have
 	$\norm{\prod_{h=1}^t \vAhat_{\sigma_h}}
		 			\leq \sum_{i=0}^t \sum_{j \in [\binom{t}{i}]} \norm{\vF_{i,j}}
		 			\leq \sum_{i=0}^t \binom{t}{i} \kappa^{i+1} \xi^{t-i} \epsilon_\vA^i
		 			\leq \kappa (\kappa \epsilon_\vA + \xi)^t
 	$.
 	
 	Similarly,
 	$
 	\norm{\prod_{h=1}^t \vAhat_{\sigma_h} {-} \prod_{h=1}^t \vA_{\sigma_h}}
 		= \|\sum_{i=1}^t \sum_{j \in [\binom{t}{i}]} $ 		
 		$\vF_{i,j}\|
 		\leq \sum_{i=0}^t \sum_{j \in [\binom{t}{i}]} \norm{\vF_{i,j}} - \norm{\vF_{0,1}}
 		\leq \kappa (\kappa \epsilon_\vA + \xi)^t - \kappa \xi^t
 		\leq \kappa^2 t (\kappa \epsilon_\vA + \xi)^{t-1} \epsilon_\vA,
 	$
 	where the last line follows from the fact that for function $f(x) {:=} x^t$ and $x,a {\geq} 0$, $f(x) {\geq} f(x+a) - a\cdot f'(x+a)$.
\end{proof}
Based on Lemma \ref{lemma_MJSReduction_JSRPerturb}, we have the following corollaries, which will be used in different settings in later derivations.
\begin{corollary}\label{corollary_MJSReduction_singleMatrixPerturbation}
 	Consider two matrices $\vAcal$  and $\vAcalbar$ with $\norm{\vAcal - \vAcalbar} \leq \epsilon_{\vAcal}$. 
    Suppose there exists a pair $\{\rho, \tau\}$ such that for all $k \in \mathbb{N}$, $\norm{\vAcal^k} \leq \tau \rho^k$.
 	Then, we have
 	$
		\norm{\vAcalbar^t} \leq \tau \parenthesesbig{\tau \epsilon_{\vAcal} + \rho}^t
	$
	and
	$
		\norm{\vAcalbar^t - \vAcal^t} \leq \tau^2 t \parenthesesbig{\tau \epsilon_{\vAcal} + \rho}^{t-1} \epsilon_{\vAcal}.
	$
\end{corollary}
\begin{corollary}\label{corollary_MJSReduction_scalarMultiplicationPerturbation}
 	Consider two sets of scalars $a_1, \dots, a_s$ and $\ahat_1, \dots, \ahat_s$ with $|a_i {-} \hat{a}_i| {<} \epsilon_a$ and $|a_i| {<} \bar{a}$ for all $i {\in} [\numSys]$. Then, for all $t$ and any sequence $\sigma_{1:t} {\in} [\numSys]^{t}$, we have
 	$
	 		|\prod_{h=1}^t \ahat_{\sigma_h}| \leq (\epsilon_a + \bar{a})^t
	$
	and
	$
	 		|\prod_{h=1}^t \ahat_{\sigma_h} - \prod_{h=1}^t a_{\sigma_h}| \leq t (\epsilon_a + \bar{a})^{t-1} \epsilon_a.
	$
\end{corollary}

The next result considers the evolution of state $\vx_t$ in the mean-square sense for autonomous MJSs.
\begin{lemma}[Lemma 9 in \cite{sattar2021identification}]\label{lemma_MJSReduction_MJSCovDynamics}
	Consider \textup{MJS}$(\vA_{1:\numSys}, 0, \vT)$ and define matrix $\vAcal {\in} \dm{\numSys \dimSt^2}{\numSys \dimSt^2}$ with its $ij$-th $\dimSt^2 {\times} \dimSt^2$ block given by $[\vAcal]_{ij} {:=} \vT(j,i) {\cdot} \vA_j {\otimes} \vA_j$. 
	Let $\vSigma_t^{(i)} {:=} \expctn[\vx_t \vx_t^\T \indicator{\omega_t=i}]$ and $\vs_t {:=} [\vek(\vSigma_t^{(1)})^\T, \dots, \vek(\vSigma_t^{(\numSys)})^\T]^\T$. Then, 
	$
		\vs_t = \vAcal^t \vs_0.
	$
\end{lemma}

Recall in Section \ref{sec_MJSReduction_approxMetrics}, for $\Sigma$, i.e., MJS($\vA_{1:\numSys}, \vB_{1:\numSys}, \vT$), we define the augmented state matrix $\vAcal \in \dm{\numSys \dimSt^2}{\numSys \dimSt^2}$ with its $ij$-th $\dimSt^2 \times \dimSt^2$ block given by $[\vAcal]_{ij}:= \vT(j,i) \cdot \vA_j \otimes \vA_j$; and for any $\rho \geq \rho(\vAcal)$ and all $k \in \mathbb{N}$, we have
$
\norm{\vAcal^k} \leq \tau \rho^k
$.
The next lemma is regarding the augmentation of two MJS with the same $\vA$ matrix.
\begin{lemma}\label{lemma_MJSReduction_concatenateTwoMJS}
	Construct matrix $\vAcalcheck \in \dm{4 \numSys \dimSt^2}{4 \numSys \dimSt^2}$ with its $ij$-th $4\dimSt^2 \times 4\dimSt^2$ block given by $[\vAcalcheck]_{ij}:= \vT(j,i) \cdot 
	\begin{bmatrix}
		\vA_j & \\ 
		& \vA_j
	\end{bmatrix} 
	\otimes
	\begin{bmatrix}
		\vA_j	& \\ 
		& \vA_j
	\end{bmatrix} 
	$. 
    Then, for all $k \in \mathbb{N}$, $\norm{\vAcalcheck^k} \leq \tau \rho^k$.
\end{lemma}
To see this result, first notice that there exists a permutation matrix $\vP$ such that $\vP \vAcalcheck \vP^\T = \vI_4 \otimes \vAcal$, where $\vI_4$ denotes the $4 \times 4$ identity matrix. This gives $\norm{\vAcalcheck^k} = \norm{\vAcal^k}$ and shows the claim.

To prove the result in Theorem \ref{thrm_MJSReduction_stateDistSyncMSS}, we first consider the simplified autonomous case but with potentially different initial states $\vx_0$ and $\vxhat_0$.
\begin{proposition}\label{prop_MJSReduction_stateDistSyncMSS_autonomous}
	Consider the setup in Theorem \ref{thrm_MJSReduction_stateDistSyncMSS} except that $\vu_t=0$ for all $t$, and $\vx_0$ and $\vxhat_0$ can be different such that $\norm{\vx_0 - \vxhat_0} \leq \epsilon_0$ for some $\epsilon_0 \geq 0$.
	For perturbation, assume $\epsilon_\vA \leq \min \curlybrackets{\Abar, \frac{1-\rho}{6\tau \overset{{} }{\Abar} \norm{\vT} }}$.
	Then, $\expctn[\norm{\vx_t - \vxhat_t}] \leq 4 \sqrt{\dimSt \sqrt{\numSys}} \tau \rho_0^{\frac{t-1}{2}} \big( \norm{\vx_0}\sqrt{t \Abar \norm{\vT} \epsilon_\vA} + \sqrt{(\norm{\vx_0} + \epsilon_0) \epsilon_0} \, \big).$
\end{proposition}
\begin{proof}
	First, we construct two autonomous switched systems: 
	\begin{equation}
		\check{\Pi}:=
		\left\{\begin{matrix}
			\check{\vx}_{t+1} = \check{\vA}_{\check{\omega}_t}\check{\vx}_t \\
			\check{\omega}_t = \omega_t, \
		\end{matrix}\right.
		, \quad
		\bar{\Pi}:=
		\left\{\begin{matrix}
			\bar{\vx}_{t+1} = \bar{\vA}_{\bar{\omega}_t}\bar{\vx}_t \\
			\bar{\omega}_t = \omega_t, \
		\end{matrix}\right.	
	\end{equation}
	where 
	for $i {\in} [\numSys]$ (suppose $i {\in} \Omega_k$),
	$
		\check{\vA}_i := \begin{bmatrix} \vA_i	& \\ & \vA_i \end{bmatrix}, 
		\bar{\vA}_i := \begin{bmatrix} \vA_i	& \\ & \vAhat_k \end{bmatrix}.
	$
	Since $\omega_t$ of $\Sigma$ follows Markov chain $\vT$, systems $\check{\Pi}$ and $\bar{\Pi}$ can be viewed as MJS$(\check{\vA}_{1:\numSys}, 0, \check{\vT})$ and MJS$(\bar{\vA}_{1:\numSys}, 0, \bar{\vT})$ respectively with $\check{\vT} = \bar{\vT} = \vT$.
	We then define observations for $\check{\Pi}$ and $\bar{\Pi}$: $\check{\vy}_t = \check{\vC} \check{\vx}_t$ and $\bar{\vy}_t = \bar{\vC} \bar{\vx}_t$ where $\check{\vC} = \bar{\vC} = [\vI_n, - \vI_n]$. We set their initial states as $\check{\vx}_0 = [\vx_0^\T, \vx_0^\T]^\T$, $\bar{\vx}_0 = [\vx_0^\T, \vxhat_0^\T]^\T$ where $\vx_0$ and $\vxhat_0$ are the initial states of $\Sigma$ and $\hat{\Sigma}$ respectively. 
	
	By construction, we have, for all t, $\check{\vx}_t {=} [\vx_t^\T, \vx_t^\T]^\T$ and $\bar{\vx}_t {=} [\vx_t^\T, \vxhat_t^\T]^\T$, thus $\check{\vy}_t {=} 0$ and $\bar{\vy}_t {=} \vx_t {-} \vxhat_t$. Define $\check{\vSigma}_t := \expctn[\check{\vx}_t \check{\vx}_t^\T]$ and $\bar{\vSigma}_t := \expctn[\bar{\vx}_t \bar{\vx}_t^\T]$, then we have
	$\expctn[\norm{\vx_t - \vxhat_t}^2] 
	{=}
	$
	$
	\expctn[\bar{\vy}_t\bar{\vy}_t^\T] 
	{=} \expctn[\bar{\vy}_t\bar{\vy}_t^\T] {-} \expctn[ \check{\vy}_t\check{\vy}_t^\T] 
	{=} \tr(\bar{\vC}^\T \bar{\vC} \bar{\vSigma}_t) {-} \tr(\check{\vC}^\T \check{\vC} \check{\vSigma}_t)
	{=} $	
	$\tr(\bar{\vC}^\T \bar{\vC} (\bar{\vSigma}_t - \check{\vSigma}_t)).
	$
	Since $\bar{\vC}^\T \bar{\vC} \succeq 0$, we further have
	\begin{equation}\label{eq_MJSReduction_22}
		\hspace*{-0.5em}
		\expctn[\norm{\vx_t {-} \vxhat_t}^2] \leq \tr(\bar{\vC}^\T \bar{\vC}) \norm{\bar{\vSigma}_t {-} \check{\vSigma}_t} = 2 \dimSt \norm{\bar{\vSigma}_t {-} \check{\vSigma}_t}.
	\end{equation}
	Let $\check{\vSigma}_{t}^{(i)} := \expctn[\check{\vx}_t \check{\vx}_t^\T \indicator{\check{\omega}_t=i}]$, $\bar{\vSigma}_{t}^{(i)} := \expctn[\bar{\vx}_t \bar{\vx}_t^\T \indicator{\bar{\omega}_t=i}]$,  $\check{\vs}_t := [\vek(\check{\vSigma}_{t}^{(1)})^\T, \dots, \vek(\check{\vSigma}_{t}^{(\numSys)})^\T]^\T$ and $\bar{\vs}_t := [\vek(\bar{\vSigma}_{t}^{(1)})^\T, \dots $ $, \vek(\bar{\vSigma}_{t}^{(\numSys)})^\T]^\T$.
	Note that $\vek(\check{\vSigma}_t) = [\vI_{4\dimSt^2}, \dots, \vI_{4\dimSt^2}] \check{\vs}_t $ and $ \vek(\bar{\vSigma}_t) = [\vI_{4\dimSt^2}, \dots, \vI_{4\dimSt^2}] \bar{\vs}_t$, thus we have	
	$
	\norm{\check{\vSigma}_t - \bar{\vSigma}_t}
	\leq \norm{\check{\vSigma}_t - \bar{\vSigma}_t}_\fro
	= \norm{\vek(\check{\vSigma}_t - \bar{\vSigma}_t)}
	\leq \sqrt{\numSys} \norm{\check{\vs}_t - \bar{\vs}_t}.
	$		
	Plugging this into \eqref{eq_MJSReduction_22}, we have
	\begin{equation}
		\expctn[\norm{\vx_t - \vxhat_t}^2] \leq 2 \dimSt \sqrt{\numSys} \norm{\check{\vs}_t - \bar{\vs}_t}.
	\end{equation}
	By Lemma \ref{lemma_MJSReduction_MJSCovDynamics}, we have $\check{\vs}_t = \vAcalcheck^t \check{\vs}_0$ and $\bar{\vs}_t = \vAcalbar^t \bar{\vs}_0$, where $\vAcalcheck {\in} \dm{4 \numSys \dimSt^2}{4 \numSys \dimSt^2}$ is constructed such that its $ij$-th $4\dimSt^2 \times 4\dimSt^2$ block given by $[\vAcalcheck]_{ij} = \check{\vT}(j,i) \check{\vA}_j \otimes \check{\vA}_j$, and $\vAcalbar$ is constructed similarly. By triangle inequality, we further have
	\begin{equation} \label{eq_MJSReduction_35}
		\hspace*{-0.6em}
		\expctn[\norm{\vx_t {-} \vxhat_t}^2] {\leq} 2 \dimSt \sqrt{\numSys} (
		\norm{\vAcalcheck^t {-} \vAcalbar^t} \norm{\check{\vs}_0} {+} \norm{\vAcalbar^t} \norm{\check{\vs}_0 {-} \bar{\vs}_0}
		)
	\end{equation}
	To bound $\expctn[\norm{\vx_t {-} \vxhat_t}^2]$, we seek to bound the terms on the RHS individually.	
	Since $\check{\vs}_0 = [\vek(\check{\vx}_0 \check{\vx}_0^\T)^\T \cdot \prob(\omega_t=1), \dots, \vek(\check{\vx}_0 \check{\vx}_0^\T)^\T \cdot \prob(\omega_t=\numSys)]^\T$, we have $\norm{\check{\vs}_0} = \norm{\check{\vx}_0 \check{\vx}_0^\T}_\fro \cdot (\sum_{i \in [\numSys]} \prob(\omega_t=i)^2)^\frac{1}{2} \leq \norm{\check{\vx}_0 \check{\vx}_0^\T}_\fro = 2 \norm{\vx_0}^2$. Similarly, we have $ \norm{\check{\vs}_0 - \bar{\vs}_0} \leq \norm{\bar{\vx}_0 \bar{\vx}_0^\T - \check{\vx}_0 \check{\vx}_0^\T }_\fro \leq \norm{\bar{\vx}_0 (\bar{\vx}_0 - \check{\vx}_0)^\T}_\fro + \norm{ (\bar{\vx}_0 - \check{\vx}_0) \check{\vx}_0^\T }_\fro \leq \sqrt{2}(\sqrt{3} \norm{\vx_0}+\epsilon_0) \epsilon_0 $. 
	
	To bound $\norm{\vAcalbar^t}$ and $\norm{\vAcalcheck^t {-} \vAcalbar^t}$, we first evaluate $\norm{\vAcalcheck - \vAcalbar}$. Define $\vDelta_i := \check{\vA}_i \otimes \check{\vA}_i - \bar{\vA}_i \otimes \bar{\vA}_i$ for all $i$, and block diagonal matrix $\vDelta \in \dm{\numSys \dimSt^2}{\numSys \dimSt^2}$ such that the $i$th $\dimSt^2 {\times} \dimSt^2$ block is given by $\vDelta_i$. Then one can verify that $\vAcalcheck - \vAcalbar = (\vT \otimes \vI_{\dimSt^2}) \vDelta$, which gives $\norm{\vAcalcheck - \vAcalbar} \leq \norm{\vT} \max_i \norm{\vDelta_i}$. For $\norm{\vDelta_i}$, we have $\norm{\vDelta_i} \leq \norm{\check{\vA}_i} \norm{\check{\vA}_i - \bar{\vA}_i} +  \norm{\check{\vA}_i - \bar{\vA}_i} \norm{\bar{\vA}_i}$. It is easy to see $\norm{\check{\vA}_i} \leq \Abar$, $\norm{\check{\vA}_i - \bar{\vA}_i} \leq \epsilon_\vA$, and $\norm{\bar{\vA}_i} \leq$ $ \Abar + \epsilon_\vA \leq 2 \Abar$. These give $\norm{\vAcalcheck - \vAcalbar} \leq 3 \Abar \norm{\vT} \epsilon_\vA$.
	From Lemma \ref{lemma_MJSReduction_concatenateTwoMJS}, we know for all $k \in \mathbb{N}$, $\norm{\vAcalcheck^k} \leq \tau \rho^k$.
	Then, according to Corollary \ref{corollary_MJSReduction_singleMatrixPerturbation}, we have $\norm{\vAcalbar^t} \leq \tau (3 \tau \Abar \norm{\vT} \epsilon_\vA + \rho)^t \leq \tau \rho_0^t$ and 
	$\norm{\vAcalcheck^t - \vAcalbar^t} \leq 3 t (3 \tau \Abar \norm{\vT} \epsilon_\vA + \rho)^{t-1} \tau^2  \Abar \norm{\vT} \epsilon_\vA \leq 3t \rho_0^{t-1} \tau^2  \Abar \norm{\vT}  \epsilon_\vA$, where the premise $\epsilon_\vA {\leq} \frac{1-\rho}{6 \tau \Abar \norm{\vT} }$ and notation $\rho_0 {:=} \frac{1+\rho}{2}$ are used. 
	
	Finally, plugging in the bounds we just derived for each term on the RHS of \eqref{eq_MJSReduction_35} back, we have $\expctn[\norm{\vx_t - \vxhat_t}] \leq \sqrt{\expctn[\norm{\vx_t - \vxhat_t}^2]} \leq 4 \sqrt{\dimSt \sqrt{\numSys}} \tau \rho_0^{\frac{t-1}{2}} \big( \norm{\vx_0}\sqrt{t \Abar \norm{\vT} \epsilon_\vA} + \sqrt{(\norm{\vx_0} + \epsilon_0) \epsilon_0} \, \big)$, which concludes the proof.
\end{proof}

\begin{proof}[Main Proof for Theorem \ref{thrm_MJSReduction_stateDistSyncMSS}]
	We first decompose $\vx_t$ in terms of the contribution from $\vx_0$, $\vu_{0:t-1}$: define $\vx_t^{(0')} {:=} \big(\prod_{h=0}^{t-1} \vA_{\omega_h} \big) \vx_0$, for $l{=}0,\dots,t-2$, 	
	define $\vx_t^{(l)} {:=} \big( \prod_{h=l+1}^{t-1} \vA_{\omegahat_h} \big) \vB_{\omega_{l}} \vu_l$, and $\vx_t^{(t-1)} {:=}\vB_{\omega_{t-1}} $
	$\vu_{t-1} $. Then it is easy to see $\vx_t {=} \vx_t^{(0')} {+} \sum_{l=0}^{t-1} \vx_t^{(l)}$. Similarly, we define $\vxhat_t^{(0')}$ and $\vxhat_t^{(l)}$ for $\vxhat_t$ such that $\vxhat_t = \vxhat_t^{(0')} + \sum_{l=0}^{t-1} \vxhat_t^{(l)}$. According to Proposition \ref{prop_MJSReduction_stateDistSyncMSS_autonomous}, we have
	\begin{equation}\label{eq_MJSReduction_36}
    		\expctn[\norm{\vx_t^{(0')} {-} \vxhat_t^{(0')}}] {\leq} 4 {\textstyle\sqrt{\dimSt\sqrt{\numSys}}} \tau \rho_0^{\frac{t-1}{2}} {\textstyle\sqrt{t \Abar \norm{\vT} \epsilon_\vA}} \norm{\vx_0} . 
	\end{equation}
	Note that $\vx_t^{(l)}$ and $\vxhat_t^{(l)}$ can be viewed as the states at time $t-l$ with respective initial states $\vB_{\omega_l} \vu_l$ and $\vBhat_{\omega_l} \vu_l$ and zero inputs. Therefore, applying Proposition \ref{prop_MJSReduction_stateDistSyncMSS_autonomous} again, we have
	\begin{multline}\label{eq_MJSReduction_37}
		\expctn[\norm{\vx_t^{(l)} - \vxhat_t^{(l)}}] \leq 4 {\textstyle\sqrt{\dimSt \sqrt{\numSys}}} \tau \rho_0^{\frac{t-l-1}{2}} {\textstyle\sqrt{\Bbar}} \bar{u} \\\cdot \big({\textstyle\sqrt{(t-l-1) \Abar \norm{\vT} \epsilon_\vA}} + \sqrt{2 \epsilon_\vB}\big),
	\end{multline}
	where the premise $\epsilon_\vB {\leq} \Bbar$ is applied. 	
	With \eqref{eq_MJSReduction_36} and \eqref{eq_MJSReduction_37}, we have	
	$\expctn[\norm{\vx_t - \vxhat_t}] \leq \expctn[\norm{\vx_t^{(0')} - \vxhat_t^{(0')}}] + \sum_{l=0}^{t-1}$	
	$  \expctn[\norm{\vx_t^{(l)} {-} \vxhat_t^{(l)}}] {\leq}
	4 \sqrt{\dimSt \sqrt{\numSys}} \tau \rho_0^{\frac{t-1}{2}} \sqrt{t \Abar \norm{\vT} \epsilon_\vA} \norm{\vx_0} + 4 \sqrt{\dimSt \sqrt{\numSys} \Bbar}\tau  \bar{u} (\frac{\sqrt{\rho_0}}{(1-\sqrt{\rho_0})^2} \sqrt{\Abar \norm{\vT} \epsilon_\vA} + \frac{\sqrt{2}}{1-\sqrt{\rho_0}} \sqrt{\epsilon_\vB})$, which concludes the proof.
\end{proof}

\section{Approximation with Unif. Stability --- Proof for Theorem \ref{thrm_MJSReduction_bisimUS}}
\begin{proof}[Proof for Theorem \ref{thrm_MJSReduction_bisimUS} \ref{enum_bisimUS_stateDistSync}]
$\vx_t$ and $\vxhat_t$ can be decomposed as:
$
\vx_t 
	= \big(\prod_{h=0}^{t-1} \vA_{\omega_h} \big)\vx_0 + 
$
$
	 \sum_{t'=0}^{t-2} \big( \prod_{h=t'+1}^{t-1} \vA_{\omega_h} \big) \vB_{\omega_{t'}} \vu_{t'} {+} \vB_{\omega_{t-1}} \vu_{t-1}, 
$
$
\vxhat_t
{=} \big(\prod_{h=0}^{t-1} 
$
$
 \vAhat_{\omegahat_h} \big)\vxhat_0  + \sum_{t'=0}^{t-2} \big( \prod_{h=t'+1}^{t-1} \vAhat_{\omegahat_h} \big) \vBhat_{\omegahat_{t'}} \vuhat_{t'} + \vBhat_{\omegahat_{t-1}} \vuhat_{t-1}.
$
Since in Algorithm \ref{Alg_MJSReduction} we let $\vAhat_k {=} |\hat{\Omega}_k|^\inv \sum_{i \in \hat{\Omega}_k} \vA_i$, and the premise gives $\hat{\Omega}_{1:\numCls} {=} \Omega_{1:\numCls}$, we have
$\norm{\vAhat_k - \vA_i} {\leq}$
$ \epsilon_\vA$ for all $i \in \Omega_k$. Based on the mode synchrony Setup \ref{setup_modeSynchrony}, i.e., $\omega_t \in \Omega_{\omegahat_t}$, we further have $\norm{\vAhat_{\omegahat_t} {-} \vA_{\omega_t}} \leq \epsilon_\vA$. Similarly, we obtain $\norm{\vBhat_{\omegahat_t} {-} \vB_{\omega_t}} \leq \epsilon_\vB$. Then, by Lemma \ref{lemma_MJSReduction_JSRPerturb}
: (i) $\norm{\prod_{h=t'+1}^{t-1} \vAhat_{\omegahat_h}} {\leq} \kappa(\kappa \epsilon_\vA + \xi)^{t-t'-1}$
and (ii) $
\|\prod_{h=t'+1}^{t-1} \vA_{\omega_h} $
${-} \prod_{h=t'+1}^{t-1} \vAhat_{\omegahat_h}\| {\leq} \kappa^2 (t{-}t'{-}1) (\kappa \epsilon_\vA + \xi)^{t-t'-2} \epsilon_\vA.
$

With (i) and (ii), and the fact that 
$\prod_{h=t'+1}^{t-1} \vA_{\omega_h} \vB_{\omega_{t'}} - $
$\prod_{h=t'+1}^{t-1} \vAhat_{\omegahat_h} \vBhat_{\omegahat_{t'}} {=}$ $ \big( \prod_{h=t'+1}^{t-1} \vA_{\omega_h} {-} \prod_{h=t'+1}^{t-1} \vAhat_{\omegahat_h} \big) \vBhat_{\omegahat_{t'}} $
$ - \big( \prod_{h=t'+1}^{t-1} \vA_{\omega_h}\big)$ $ (\vBhat_{\omegahat_{t'}} - \vB_{\omega_{t'}})$, we have
\begin{multline}\label{eq_MJSReduction_17}
	{\textstyle\norm{\prod_{h=t'+1}^{t-1} \vA_{\omega_h} \vB_{\omega_{t'}} {-} \prod_{h=t'+1}^{t-1} \vAhat_{\omegahat_h} \vBhat_{\omegahat_{t'}}}}
	\leq \kappa \xi^{t-t'-1} \epsilon_\vB \\
	+ \kappa^2 (t-t'-1) (\kappa \epsilon_\vA + \xi)^{t-t'-2} (\Bbar + \epsilon_\vB) \epsilon_\vA.
\end{multline}
According to Setup \ref{setup_InitializationExcitation}, $\Sigma$ and $\hat{\Sigma}$ have the same initial states and inputs. Then, applying triangle inequality to the difference $\norm{\vx_t - \vxhat_t}$, we have
\begin{multline}
    \norm{\vx_t - \vxhat_t} \leq \kappa^2 t (\kappa \epsilon_\vA + \xi)^{t-1} \norm{\vx_0} \epsilon_\vA \\
    + \kappa^2 \frac{1 + t (\kappa \epsilon_\vA + \xi)^t}{1 - \kappa \epsilon_\vA - \xi} (\Bbar + \epsilon_\vB) \bar{u} \epsilon_\vA + \frac{\kappa}{1 - \xi} \bar{u} \epsilon_\vB,
\end{multline}
where the following facts are implicitly used: (i) $\kappa \geq 1$ by definition; (ii) $\kappa \epsilon_\vA + \xi < 1$ according to the premise. Finally, note that we assume perturbation $\epsilon_\vA \leq \frac{1-\xi}{2 \kappa}$ and $\epsilon_\vB \leq \Bbar$, we have
$
		\norm{\vx_t - \vxhat_t}
		\leq
		t \xi_0^{t-1} \kappa^2 \norm{\vx_0} \epsilon_\vA
		+ \frac{2(1 + t \xi_0^t) \kappa^2 \Bbar \bar{u}}{1 - \xi_0}  \epsilon_\vA
		+ \frac{\kappa \bar{u}}{1 - \xi} \epsilon_\vB,
$
which concludes the proof.
\end{proof}

\subsection{Proof for Theorem \ref{thrm_MJSReduction_bisimUS} \ref{enum_bisimUS_WasDisAutonomousSys}}\label{appendix_MJSReduction_bisimUS_WasDisZeroEpsT}
To ease the proof exposition, we first define a few notations and concepts.
For the original system $\Sigma$, fixing the initial state $\vx_0$ and input sequence $\vu_{0:t-1}$, there can be at most $\numSys^t$ possible $\vx_t$, each of which correspond to one possible mode switching sequence $\omega_{0:t-1} \in [\numSys]^t$. We use $g \in [\numSys^t]$ to index these states and mode sequences, i.e., mode sequence $\omega_{0:t-1}^{(g)}$ generates state $\vx_t^{(g)}$. 
Then, the reachable set $\Xcal_t$ defined in Section \ref{sec_MJSReduction_approxMetrics} satisfies $\Xcal_t = \bigcup_{g \in [\numSys^t]} \curlybrackets{\vx_t^{(g)}}$. Define probability measure $q_t(g):= \prob(\omega_{0:t-1} = \omega_{0:t-1}^{(g)})$, then we see $p_t(\vx)$ defined in Section \ref{sec_MJSReduction_approxMetrics} satisfies $p_t(\vx) = \sum_{g: \vx_t^{(g)} = \vx} q_t(g)$. 
For the reduced $\hat{\Sigma}$ and for all $\ghat \in [\numCls^t]$, we similarly define notations $\omegahat_{0:t-1}^{(\ghat)}$ for the mode sequence,  $\vxhat_t^{(\ghat)}$ for the state, and $\qhat_t(\ghat):= \prob(\omegahat_{0:t-1} = \omegahat_{0:t-1}^{(\ghat)})$ for the measure. Then, the following holds: $\hat{\Xcal}_t = \bigcup_{\ghat \in [\numCls^t]} \curlybrackets{\vxhat_t^{(\ghat)}}$ and $\phat_t(\vxhat) = \sum_{\ghat: \vxhat_t^{(\ghat)}=\vxhat} \qhat(\ghat)$.
Next, we introduce the following relation regarding mode sequences between $\Sigma$ and $\hat{\Sigma}$.
\begin{definition}[Mode Sequence Synchrony]
 	For any $g \in [\numSys^t], \ghat \in [\numCls^t]$, we say $\omega_{0:t-1}^{(g)}$ is synchronous to $\omegahat_{0:t-1}^{(\ghat)}$ (denoted by $g \triangleright \ghat$) if $\omega_h \in \Omega_{h}$ for all $h = 0, 1, \dots, t-1$. 
\end{definition}
Note that the synchrony definition here coincides with the mode synchrony in Setup \ref{setup_modeSynchrony}. With this synchrony relation, we first present a preliminary result.
\begin{lemma}\label{eq_MJSReduction_PtwiseDist_3}
    For any $\ghat \in [\numCls]^t$, we have
    $
    	\big| \qhat_t(\ghat)  - \sum_{g: g \triangleright \ghat} q_t(g)  \big| 
 		\leq 
 		(t-1) (\bar{\Tcal} + \epsilon_\vT)^{t-2}  \epsilon_\vT.
    $
\end{lemma}
\begin{proof}
    Recall $\zeta_t$ indexes the active cluster of $\Sigma$ at time $t$, i.e., $\zeta_t{=}k$ if and only if $\omega_t {\in} \Omega_k$.
    First observe that $\sum_{g: g \triangleright \ghat'} q_t(g) = $
    $\sum_{g: g \triangleright \ghat} \prob(\omega_{0:t-1} {=} \omega_{0:t-1}^{(g)}) {=} \prob(\omega_{t-1} {\in} \Omega_{\omegahat_{t-1}^{(\ghat)}}{,} \dots{,} \omega_{0} {\in} \Omega_{\omegahat_{0}^{(\ghat)}}) {=}$     
    $\prob(\zeta_{0:t-1} {=} \omegahat_{0:t-1}^{(\ghat)})$. Also note that $\qhat_t(\ghat) {=} \prob(\omegahat_{0:t-1} {=} \omegahat_{0:t-1}^{(\ghat)})$.    
    So, to show the claim, it suffices to show for any $\sigma_{0:t} {\in} [\numCls]^t$,
 	\begin{equation}\label{eq_MJSReduction_1}
	 		|\prob( \omegahat_{0:t} = \sigma_{0:t} ) - \prob(\zeta_{0:t} = \sigma_{0:t})|
	 		\leq t (\bar{\Tcal} + \epsilon_\vT)^{t-1} \epsilon_\vT.
	 	\end{equation}
	For the LHS of \eqref{eq_MJSReduction_1}, we have
	\begin{align}
	    \prob( \omegahat_{0:t} &= \sigma_{0:t} ) = \prob(\omegahat_0 = \sigma_0) \cdot {\textstyle \prod_{h=1}^{t}} \vThat(\sigma_{h-1}, \sigma_h) \label{eq_MJSReduction_4} \\
	    \prob(\zeta_{0:t} &= \sigma_{0:t})= \prob(\omega_0 \in \Omega_{\sigma_0}) \cdot {\textstyle \prod_{h=1}^{t}} \Ttil_h \label{eq_MJSReduction_5}
	\end{align}
 	where $\Ttil_h := \prob(\omega_h \in \Omega_{\sigma_h} \mid \omega_{h-1} \in \Omega_{\sigma_{h-1}}, \dots, \omega_0 \in \Omega_{\sigma_0})$.
    Note that $\zeta_{0:t}$ may not be a Markov process when $\epsilon_\vT \neq 0$, so we cannot drop the past conditional events in \eqref{eq_MJSReduction_5}. 
    
    Let $\alpha_i := \prob(\omega_{h-1}{=}i \mid \omega_{h-2} {\in} \Omega_{\sigma_{h-2}}, \dots, \omega_0 {\in} \Omega_{\sigma_0})$, then    
    $
    \Ttil_h {=} \sum_{i {\in} \Omega_{\sigma_{h-1}}} \big[ \prob(\omega_h {\in} \Omega_{\sigma_h} \mid \omega_{h-1}{=}i) \cdot \prob(\omega_{h-1}{=}i \mid \omega_{h-2} {\in} $        
    $\Omega_{\sigma_{h-2}}, \dots, \omega_0 {\in} \Omega_{\sigma_0}) \big]
    = \sum_{i \in \Omega_{\sigma_{h-1}}} \big[\big(\sum_{j \in \Omega_{\sigma_h}} \vT(i,j)\big) \alpha_i \big]$.

    Let $\beta_i:= |\Omega_{\sigma_{h-1}}|^\inv$.
    For $\vThat(\sigma_{h-1}, \sigma_h)$, by definition in Algorithm \ref{Alg_MJSReduction} and the assumption $\hat{\Omega}_{1:\numCls} = \Omega_{1:\numCls}$, we know
    $
    \vThat(\sigma_{h-1}, \sigma_h) 
    =
    |\Omega_{\sigma_{h-1}}|^\inv $    
    $\cdot \sum_{i \in \Omega_{\sigma_{h-1}}} ( \sum_{j \in \Omega_{\sigma_h}} \vT(i,j) )
    =$
    $
    \sum_{i \in \Omega_{\sigma_{h-1}}} \big[  \big( \sum_{j \in \Omega_{\sigma_h}} \vT(i,j) \big) \beta_i \big].
    $
    
    Then, it follows that the difference
    $
    |\Ttil_h  - \vThat(\sigma_{h-1}, \sigma_h)|
    $    
    $= \big| \sum_{i, i' \in \Omega_{\sigma_{h-1}}} \big[ \big( \sum_{j \in \Omega_{\sigma_h}} \vT(i,j) \big) \alpha_i \beta_{i'} \big]
    - \sum_{i, i' \in \Omega_{\sigma_{h-1}}}
    $    
    $
    \big[ \hspace{-0.1em} \big( \hspace{-0.1em} \sum_{j \in \Omega_{\sigma_h}} \hspace{-0.3em} \vT(i',j) \big) \alpha_i \beta_{i'} \hspace{-0em}  \big] \hspace{-0em} \big|
    \leq
    \sum_{i, i' \in \Omega_{\sigma_{h-1}}} \hspace{-0.2em} \big[ \big| \sum_{j \in \Omega_{\sigma_h}}  \vT(i, j) $    
    $  - \sum_{j \in \Omega_{\sigma_h}} \vT(i', j) \big| \alpha_{i} \beta_{i'} \big] \leq \epsilon_\vT,
    $
    where the first inequality follows from triangle inequality on the absolute values;     
    the second inequality holds since the definition of perturbation $\epsilon_\vT$ in either Problem \ref{problem_lumpableCase} or \ref{problem_aggregatableCase} gives $| \sum_{j \in \Omega_{\sigma_h}} \vT(i, j)  - \sum_{j \in \Omega_{\sigma_h}} \vT(i', j) | \leq \epsilon_\vT$ for any $i, i' \in \Omega_{\sigma_{h-1}}$.
    
    We have established upper bounds for the differences between each multiplier in \eqref{eq_MJSReduction_4} and \eqref{eq_MJSReduction_5}, by Corollary \ref{corollary_MJSReduction_scalarMultiplicationPerturbation}, we obtain
    $|\prob( \omegahat_{0:t} = \sigma_{0:t} ) - \prob(\zeta_{0:t} = \sigma_{0:t})|
 	\leq
 	t (\bar{\Tcal} + \epsilon_\vT)^{t-1} \epsilon_\vT$
 	which shows \eqref{eq_MJSReduction_1} and concludes the proof.
\end{proof}

\begin{proof}[Main Proof for Theorem \ref{thrm_MJSReduction_bisimUS} \ref{enum_bisimUS_WasDisAutonomousSys}]
To lower bound the Wasserstein distance $W_\ell(p_{t}, \phat_{t})$ defined in the mass transportation problem \eqref{eq_MJSReduction_WassDistDef}, we consider the objective value given by a constrained mass transportation scheme. Recall with measures $q_t$ and $\qhat_t$, we have
$p_t(\vx) = \sum_{g: \vx_t^{(g)} = \vx} q_t(g)$ and
$\phat_t(\vxhat) = \sum_{\ghat: \vxhat_t^{(\ghat)} = \vxhat} \qhat_t(\ghat)$. With these relations, we consider the following transportation scheme in terms of $q_t$ and $\qhat_t$:
for all the mass $q_{t}(g)$ with mode sequence $\omega_t^{(g)}$ synchronous to mode sequence $\omegahat_t^{(\ghat)}$, it is prioritized to be moved to location $\ghat$; if there is surplus, i.e., $\sum_{g: g \triangleright \ghat} q_{t}(g) > \qhat_{t}(\ghat)$, we move the surplus portion $\sum_{g: g \triangleright \ghat} q_{t}(g) - \qhat_{t}(\ghat)$ elsewhere. 

Under this moving scheme, let $\Wbar_\ell(q_{t}, \qhat_{t})$ denote the optimal objective value of the mass transportation problem \eqref{eq_MJSReduction_WassDistDef}. 
Let $\hat{\Gcal}_1 {=} \curlybrackets{\ghat : \ghat {\in} [\numCls^t], \sum_{g: g \triangleright \ghat} q_{t}(g) {\leq} \qhat_{t}(\ghat)}$, $\hat{\Gcal}_2 = [\numCls^t] \backslash \hat{\Gcal}_1$. Then, $\Wbar_\ell(q_{t}, \qhat_{t})$ can be viewed as the optimal objective of the following problem:
\begin{align}
		\min_{f\geq0} \quad & \big( {\textstyle\sum_{g \in [\numSys^t], \ghat \in [\numCls^t]}} \ f(g, \ghat) \norm{\vx_t^{(g)} - \vxhat_t^{(\ghat)}}^{\ell} \big)^{1/\ell} \label{eq_MJSReduction_8} \\
		\text{s.t.} \quad 
		& {\textstyle\sum_{ g \in [\numSys^t] }} \  f(g, \ghat) = \qhat_{t}(\ghat), \forall \ \ghat \nonumber \\
		& {\textstyle\sum_{\ghat \in [\numCls^t]}} \ f(g, \ghat) = q_{t}(g), \forall \ g. \nonumber \\
		& f(g, \ghat) = q_{t}(g), \qquad \quad \; \; \forall \ g \triangleright \ghat, \forall \ \ghat \in \hat{\Gcal}_1 \label{eq_MJSReduction_10} \\
		& {\textstyle\sum_{g \triangleright \ghat}} \ f(g, \ghat) = \qhat_{t}(\ghat), \; \; \; \forall \  \ghat \in \hat{\Gcal}_2 \label{eq_MJSReduction_11}
\end{align}
where constraints \eqref{eq_MJSReduction_10} and \eqref{eq_MJSReduction_11} characterize the moving scheme outlined above. Without them, the problem reduces to \eqref{eq_MJSReduction_WassDistDef}, thus $W_\ell(p_{t}, \phat_{t}) \leq \Wbar_\ell(q_{t}, \qhat_{t})$. To prove the main claim, it suffices to show
\begin{multline}\label{eq_MJSReduction_14}
	\Wbar_\ell (q_{t}, \qhat_{t})
	\leq 
	t \xi_0^{t-1} \kappa^2 \norm{\vx_0} \epsilon_\vA \\
	+
	2 \numCls^2 t \kappa \norm{\vx_0} \numCls^t (\kappa \epsilon_\vA + \xi)^t 
	(\bar{\Tcal} + \epsilon_\vT)^{\frac{t-2}{\ell}} \epsilon_\vT^{\frac{1}{\ell}}.
\end{multline}
For all $\ghat {\in} [\numCls^t]$, define its synchrony set $\Scal(\ghat) {:=} \{g : g {\in} [\numSys^t], $
$g \triangleright \ghat \}$ 
and the asynchrony set $\Scal^c(\ghat) {:=} [\numSys^t] \backslash \Scal(\ghat)$. 
For the synchrony flow, define total flow $F_s {:=} \sum_{\ghat \in [\numCls^t], g \in \Scal(\ghat)} f(g, \ghat)$
and maximum travel distance 
$D_s {:=} \max_{\ghat \in [\numCls^t], g \in \Scal(\ghat)}\|\vx_t^{(g)} $
$- \vxhat_t^{(\ghat)}\|$. 
For the asynchrony flow, similarly define $F_a := $
$\sum_{\ghat \in [\numCls^t], g \in \Scal^c(\ghat)} f(g, \ghat)$ and
$D_a := \max_{\ghat \in [\numCls^t], g \in \Scal^c(\ghat)}$
$ \|\vx_t^{(g)} $
$- \vxhat_t^{(\ghat)} \|$.
Then, we have
$\Wbar_\ell(q_{t}, \qhat_{t}) {\leq} (F_s D_s^\ell + F_a D_a^\ell)^{\frac{1}{\ell}} \leq F_s^{\frac{1}{\ell}} D_s + F_a^{\frac{1}{\ell}} D_a$. We next bound $F_s, D_s, F_a, D_a$ separately.

For the synchrony maximum travel distance $D_s$, since 
$g \triangleright \ghat$, by Theorem \ref{thrm_MJSReduction_bisimUS} \ref{enum_bisimUS_stateDistSync}, we know 
$D_s {\leq} t \xi_0^{t-1} \kappa^2 \norm{\vx_0} \epsilon_\vA$.
For the synchrony total flow $F_s$, we simply bound it with $F_s {\leq} 1$.

Now we consider the asynchrony maximum travel distance $D_a$. First note that for any $g$ and $\ghat$, we have
$
\norm{\vx_t^{(g)} {-} \vxhat_t^{(\ghat)}}
{=}$
$
\|\prod_{h=0}^{t-1} \vA_{\omega_h^{(g)}} \vx_0 {-} \prod_{h=0}^{t-1} \vAhat_{\omegahat_h^{(\ghat)}} \vxhat_0\| \leq$
$
\norm{\prod_{h=0}^{t-1} \vA_{\omega_h^{(g)}} \vx_0} {+}$
$ \norm{\prod_{h=0}^{t-1} \vAhat_{\omegahat_h^{(\ghat)}} \vxhat_0}
\leq
2 \kappa (\kappa \epsilon_\vA {+} \xi)^t \norm{\vx_0}
$, 
where the second inequality follows from Lemma \ref{lemma_MJSReduction_JSRPerturb}.
Then, it follows that $D_a {\leq} 2 \kappa (\kappa \epsilon_\vA + \xi)^t \norm{\vx_0}$.

For the asynchrony total flow $F_a$, define $F_{a,\ghat} {:=} \sum_{g \in \Scal^c(\ghat)}$
$ f(g, \ghat)$, then $F_a = \sum_{\ghat \in [\numCls^t]} F_{a,\ghat}$. By constraints \eqref{eq_MJSReduction_10} and \eqref{eq_MJSReduction_11}, 
$
F_{a,\ghat} = 
\sum_{g \in \Scal^c(\ghat)} f(g, \ghat)
=
\qhat_{t}(\ghat) {-} \sum_{g: g \in \Scal(\ghat)}f(g, \ghat)
$
$=\qhat_{t}(\ghat) {-} \sum_{g: g \triangleright \ghat} f(g, \ghat).
$
Thus,
if $\ghat \in \hat{\Gcal}_2$, $F_{a, \ghat}=0$;
and if $\ghat \in \hat{\Gcal}_1$,
$
F_{a, \ghat}
=
\qhat_{t}(\ghat) - \sum_{g: g \triangleright \ghat} q_{t}(g)> 0$.
For the latter case, according to Lemma \ref{eq_MJSReduction_PtwiseDist_3}, we have $F_{a, \ghat} = |\qhat_{t}(\ghat) - \sum_{g: g \triangleright \ghat} q_{t}(g)| 
\leq 
(t-1) (\bar{\Tcal} + \epsilon_\vT)^{t-2} \epsilon_\vT$, which further implies that $F_a = \sum_{\ghat \in [\numCls^t]} F_{a,\ghat} \leq \numCls^2 t (\numCls(\bar{\Tcal} + \epsilon_\vT))^{t-2} \epsilon_\vT$.

Finally, \eqref{eq_MJSReduction_14} can be shown by plugging the upper bounds for $F_s, D_s, F_a, D_a$ into the relation that 
$\Wbar_\ell(q_{t}, \qhat_{t}) {\leq} F_s^{\frac{1}{\ell}} D_s + F_a^{\frac{1}{\ell}} D_a$, which concludes the proof.	
\end{proof}

\bibliographystyle{IEEEtran}
\bibliography{MJSReduction.bib} 

\begin{thebibliography}{10}
\providecommand{\url}[1]{#1}
\csname url@samestyle\endcsname
\providecommand{\newblock}{\relax}
\providecommand{\bibinfo}[2]{#2}
\providecommand{\BIBentrySTDinterwordspacing}{\spaceskip=0pt\relax}
\providecommand{\BIBentryALTinterwordstretchfactor}{4}
\providecommand{\BIBentryALTinterwordspacing}{\spaceskip=\fontdimen2\font plus
\BIBentryALTinterwordstretchfactor\fontdimen3\font minus
  \fontdimen4\font\relax}
\providecommand{\BIBforeignlanguage}[2]{{%
\expandafter\ifx\csname l@#1\endcsname\relax
\typeout{** WARNING: IEEEtran.bst: No hyphenation pattern has been}%
\typeout{** loaded for the language `#1'. Using the pattern for}%
\typeout{** the default language instead.}%
\else
\language=\csname l@#1\endcsname
\fi
#2}}
\providecommand{\BIBdecl}{\relax}
\BIBdecl

\bibitem{blackmore2005combining}
L.~Blackmore, S.~Funiak, and B.~C. Williams, ``Combining stochastic and greedy
  search in hybrid estimation,'' in \emph{AAAI}, 2005, pp. 282--287.

\bibitem{cajueiro2002stochastic}
D.~Cajueiro, ``Stochastic optimal control of jumping markov parameter processes
  with applications to finance,'' Ph.D. dissertation, PhD thesis, 2002,
  Instituto Tecnol{\'o}gico de Aeron{\'a}utica-ITA, Brazil, 2002.

\bibitem{loparo1990probabilistic}
K.~Loparo and F.~Abdel-Malek, ``A probabilistic approach to dynamic power
  system security,'' \emph{IEEE transactions on circuits and systems}, vol.~37,
  no.~6, pp. 787--798, 1990.

\bibitem{svensson2008optimal}
L.~E. Svensson, N.~Williams \emph{et~al.}, ``Optimal monetary policy under
  uncertainty: a markov jump-linear-quadratic approach,'' \emph{Federal Reserve
  Bank of St. Louis Review}, vol.~90, no.~4, pp. 275--293, 2008.

\bibitem{ugrinovskii2005decentralized}
V.~Ugrinovskii and H.~R. Pota, ``Decentralized control of power systems via
  robust control of uncertain markov jump parameter systems,''
  \emph{International Journal of Control}, vol.~78, no.~9, pp. 662--677, 2005.

\bibitem{sinopoli2005lqg}
B.~Sinopoli, L.~Schenato, M.~Franceschetti, K.~Poolla, and S.~Sastry, ``An lqg
  optimal linear controller for control systems with packet losses,'' in
  \emph{Proceedings of the 44th IEEE Conference on Decision and Control}.\hskip
  1em plus 0.5em minus 0.4em\relax IEEE, 2005, pp. 458--463.

\bibitem{truong2021analysis}
T.~H. Truong, P.~Seiler, and L.~E. Linderman, ``Analysis of networked
  structural control with packet loss,'' \emph{IEEE Transactions on Control
  Systems Technology}, 2021.

\bibitem{zhang2003h}
L.~Zhang, B.~Huang, and J.~Lam, ``$\mathcal{H}_\infty$ model reduction of
  markovian jump linear systems,'' \emph{Systems \& Control Letters}, vol.~50,
  no.~2, pp. 103--118, 2003.

\bibitem{zamani2014approximately}
M.~Zamani and A.~Abate, ``Approximately bisimilar symbolic models for randomly
  switched stochastic systems,'' \emph{Systems \& Control Letters}, vol.~69,
  pp. 38--46, 2014.

\bibitem{zhe2021clustering}
Z.~Du, N.~Ozay, and L.~Balzano, ``Clustering-based mode reduction for markov
  jump systems,'' in \emph{Learning for Dynamics and Control Conference}.\hskip
  1em plus 0.5em minus 0.4em\relax PMLR, 2022, pp. 689--701.

\bibitem{sattar2021identification}
Y.~Sattar, Z.~Du, D.~A. Tarzanagh, L.~Balzano, N.~Ozay, and S.~Oymak,
  ``Identification and adaptive control of markov jump systems: Sample
  complexity and regret bounds,'' \emph{arXiv preprint arXiv:2111.07018}, 2021.

\bibitem{desharnais2004metrics}
J.~Desharnais, V.~Gupta, R.~Jagadeesan, and P.~Panangaden, ``Metrics for
  labelled markov processes,'' \emph{Theoretical computer science}, vol. 318,
  no.~3, pp. 323--354, 2004.

\bibitem{tkachev2014approximation}
I.~Tkachev and A.~Abate, ``On approximation metrics for linear temporal
  model-checking of stochastic systems,'' in \emph{Proceedings of the 17th
  international conference on Hybrid systems: computation and control}, 2014,
  pp. 193--202.

\bibitem{bian2017relationship}
G.~Bian and A.~Abate, ``On the relationship between bisimulation and trace
  equivalence in an approximate probabilistic context,'' in \emph{International
  Conference on Foundations of Software Science and Computation
  Structures}.\hskip 1em plus 0.5em minus 0.4em\relax Springer, 2017, pp.
  321--337.

\bibitem{lun2018approximate}
Y.~Z. Lun, J.~Wheatley, A.~D’Innocenzo, and A.~Abate, ``Approximate
  abstractions of markov chains with interval decision processes,''
  \emph{IFAC-PapersOnLine}, vol.~51, no.~16, pp. 91--96, 2018.

\bibitem{zhang2018state}
A.~Zhang and M.~Wang, ``Spectral state compression of markov processes,''
  \emph{IEEE transactions on information theory}, vol.~66, no.~5, pp.
  3202--3231, 2019.

\bibitem{du2019mode}
Z.~Du, N.~Ozay, and L.~Balzano, ``Mode clustering for markov jump systems,''
  \emph{arXiv preprint arXiv:1910.02193}, 2019.

\bibitem{bittracher2021probabilistic}
A.~Bittracher and C.~Sch{\"u}tte, ``A probabilistic algorithm for aggregating
  vastly undersampled large markov chains,'' \emph{Physica D: Nonlinear
  Phenomena}, vol. 416, p. 132799, 2021.

\bibitem{abate2011approximate}
A.~Abate, A.~D'Innocenzo, and M.~D. Di~Benedetto, ``Approximate abstractions of
  stochastic hybrid systems,'' \emph{IEEE Transactions on Automatic Control},
  vol.~56, no.~11, pp. 2688--2694, 2011.

\bibitem{soudjani2011adaptive}
S.~E.~Z. Soudjani and A.~Abate, ``Adaptive gridding for abstraction and
  verification of stochastic hybrid systems,'' in \emph{2011 Eighth
  International Conference on Quantitative Evaluation of SysTems}.\hskip 1em
  plus 0.5em minus 0.4em\relax IEEE, 2011, pp. 59--68.

\bibitem{julius2006approximate}
A.~A. Julius, A.~Girard, and G.~J. Pappas, ``Approximate bisimulation for a
  class of stochastic hybrid systems,'' in \emph{2006 American Control
  Conference}.\hskip 1em plus 0.5em minus 0.4em\relax IEEE, 2006, pp.
  4724--4729.

\bibitem{julius2009approximations}
A.~A. Julius and G.~J. Pappas, ``Approximations of stochastic hybrid systems,''
  \emph{IEEE Transactions on Automatic Control}, vol.~54, no.~6, pp.
  1193--1203, 2009.

\bibitem{zamani2016approximations}
M.~Zamani, M.~Rungger, and P.~M. Esfahani, ``Approximations of stochastic
  hybrid systems: A compositional approach,'' \emph{IEEE Transactions on
  Automatic Control}, vol.~62, no.~6, pp. 2838--2853, 2016.

\bibitem{shen2019model}
Y.~Shen, Z.-G. Wu, P.~Shi, and C.~K. Ahn, ``Model reduction of markovian jump
  systems with uncertain probabilities,'' \emph{IEEE Transactions on Automatic
  Control}, vol.~65, no.~1, pp. 382--388, 2019.

\bibitem{kotsalis2010balanced}
G.~Kotsalis and A.~Rantzer, ``Balanced truncation for discrete time markov jump
  linear systems,'' \emph{IEEE Transactions on Automatic Control}, vol.~55,
  no.~11, pp. 2606--2611, 2010.

\bibitem{sun2016model}
M.~Sun and J.~Lam, ``Model reduction of discrete markovian jump systems with
  time-weighted $\text{H}_2$ performance,'' \emph{International Journal of
  Robust and Nonlinear Control}, vol.~26, no.~3, pp. 401--425, 2016.

\bibitem{larsen1991bisimulation}
K.~G. Larsen and A.~Skou, ``Bisimulation through probabilistic testing,''
  \emph{Information and computation}, vol.~94, no.~1, pp. 1--28, 1991.

\bibitem{desharnais2002bisimulation}
J.~Desharnais, A.~Edalat, and P.~Panangaden, ``Bisimulation for labelled markov
  processes,'' \emph{Information and Computation}, vol. 179, no.~2, pp.
  163--193, 2002.

\bibitem{abate2013approximation}
A.~Abate, ``Approximation metrics based on probabilistic bisimulations for
  general state-space markov processes: a survey,'' \emph{Electronic Notes in
  Theoretical Computer Science}, vol. 297, pp. 3--25, 2013.

\bibitem{girard2007approximation}
A.~Girard and G.~J. Pappas, ``Approximation metrics for discrete and continuous
  systems,'' \emph{IEEE Transactions on Automatic Control}, vol.~52, no.~5, pp.
  782--798, 2007.

\bibitem{alur2000discrete}
R.~Alur, T.~A. Henzinger, G.~Lafferriere, and G.~J. Pappas, ``Discrete
  abstractions of hybrid systems,'' \emph{Proceedings of the IEEE}, vol.~88,
  no.~7, pp. 971--984, 2000.

\bibitem{clarke2018model}
E.~M. Clarke~Jr, O.~Grumberg, D.~Kroening, D.~Peled, and H.~Veith, \emph{Model
  checking}.\hskip 1em plus 0.5em minus 0.4em\relax MIT press, 2018.

\bibitem{kurshan2014computer}
R.~P. Kurshan, \emph{Computer-aided verification of coordinating processes: the
  automata-theoretic approach}.\hskip 1em plus 0.5em minus 0.4em\relax
  Princeton university press, 2014, vol. 302.

\bibitem{maler1995synthesis}
O.~Maler, A.~Pnueli, and J.~Sifakis, ``On the synthesis of discrete controllers
  for timed systems,'' in \emph{Annual symposium on theoretical aspects of
  computer science}.\hskip 1em plus 0.5em minus 0.4em\relax Springer, 1995, pp.
  229--242.

\bibitem{abate2010approximate}
A.~Abate, J.-P. Katoen, J.~Lygeros, and M.~Prandini, ``Approximate model
  checking of stochastic hybrid systems,'' \emph{European Journal of Control},
  vol.~16, no.~6, pp. 624--641, 2010.

\bibitem{gugercin2004survey}
S.~Gugercin and A.~C. Antoulas, ``A survey of model reduction by balanced
  truncation and some new results,'' \emph{International Journal of Control},
  vol.~77, no.~8, pp. 748--766, 2004.

\bibitem{buchholz1994exact}
P.~Buchholz, ``Exact and ordinary lumpability in finite markov chains,''
  \emph{Journal of applied probability}, vol.~31, no.~1, pp. 59--75, 1994.

\bibitem{hoffmann2009bounding}
K.~H. Hoffmann and P.~Salamon, ``Bounding the lumping error in markov chain
  dynamics,'' \emph{Applied mathematics letters}, vol.~22, no.~9, pp.
  1471--1475, 2009.

\bibitem{schulman2001coarse}
L.~Schulman and B.~Gaveau, ``Coarse grains: The emergence of space and order,''
  \emph{Foundations of Physics}, vol.~31, no.~4, pp. 713--731, 2001.

\bibitem{gaveau2005dynamical}
B.~Gaveau and L.~Schulman, ``Dynamical distance: coarse grains, pattern
  recognition, and network analysis,'' \emph{Bulletin des sciences
  mathematiques}, vol. 129, no.~8, pp. 631--642, 2005.

\bibitem{meila2001random}
M.~Meil{\u{a}} and J.~Shi, ``A random walks view of spectral segmentation,'' in
  \emph{International Workshop on Artificial Intelligence and
  Statistics}.\hskip 1em plus 0.5em minus 0.4em\relax PMLR, 2001, pp. 203--208.

\bibitem{bottou1994convergence}
L.~Bottou and Y.~Bengio, ``Convergence properties of the k-means algorithms,''
  \emph{Advances in neural information processing systems}, vol.~7, 1994.

\bibitem{gonzalez1985clustering}
T.~F. Gonzalez, ``Clustering to minimize the maximum intercluster distance,''
  \emph{Theoretical computer science}, vol.~38, pp. 293--306, 1985.

\bibitem{kumar2004simple}
A.~Kumar, Y.~Sabharwal, and S.~Sen, ``A simple linear time (1+/spl
  epsiv/)-approximation algorithm for k-means clustering in any dimensions,''
  in \emph{45th Annual IEEE Symposium on Foundations of Computer
  Science}.\hskip 1em plus 0.5em minus 0.4em\relax IEEE, 2004, pp. 454--462.

\bibitem{song2010fast}
M.~Song and S.~Rajasekaran, ``Fast algorithms for constant approximation
  k-means clustering.'' \emph{Trans. Mach. Learn. Data Min.}, vol.~3, no.~2,
  pp. 67--79, 2010.

\bibitem{costa2006discrete}
O.~L.~V. Costa, M.~D. Fragoso, and R.~P. Marques, \emph{Discrete-time Markov
  jump linear systems}.\hskip 1em plus 0.5em minus 0.4em\relax Springer, 2006.

\bibitem{kuhn2019wasserstein}
D.~Kuhn, P.~M. Esfahani, V.~A. Nguyen, and S.~Shafieezadeh-Abadeh,
  ``Wasserstein distributionally robust optimization: Theory and applications
  in machine learning,'' in \emph{Operations research \& management science in
  the age of analytics}.\hskip 1em plus 0.5em minus 0.4em\relax Informs, 2019,
  pp. 130--166.

\bibitem{chen2015optimal}
Y.~Chen, T.~T. Georgiou, and M.~Pavon, ``Optimal steering of a linear
  stochastic system to a final probability distribution, part i,'' \emph{IEEE
  Transactions on Automatic Control}, vol.~61, no.~5, pp. 1158--1169, 2015.

\bibitem{goldshtein2017finite}
M.~Goldshtein and P.~Tsiotras, ``Finite-horizon covariance control of linear
  time-varying systems,'' in \emph{2017 IEEE 56th Annual Conference on Decision
  and Control (CDC)}.\hskip 1em plus 0.5em minus 0.4em\relax IEEE, 2017, pp.
  3606--3611.

\bibitem{okamoto2018optimal}
K.~Okamoto, M.~Goldshtein, and P.~Tsiotras, ``Optimal covariance control for
  stochastic systems under chance constraints,'' \emph{IEEE Control Systems
  Letters}, vol.~2, no.~2, pp. 266--271, 2018.

\bibitem{li2015ensemble}
J.-S. Li and J.~Qi, ``Ensemble control of time-invariant linear systems with
  linear parameter variation,'' \emph{IEEE Transactions on Automatic Control},
  vol.~61, no.~10, pp. 2808--2820, 2015.

\bibitem{jungers2009joint}
R.~Jungers, \emph{The joint spectral radius: theory and applications}.\hskip
  1em plus 0.5em minus 0.4em\relax Springer Science \& Business Media, 2009,
  vol. 385.

\bibitem{parrilo2008approximation}
P.~A. Parrilo and A.~Jadbabaie, ``Approximation of the joint spectral radius
  using sum of squares,'' \emph{Linear Algebra and its Applications}, vol. 428,
  no.~10, pp. 2385--2402, 2008.

\bibitem{du2021certainty}
Y.~Sattar, Z.~Du, D.~A. Tarzanagh, S.~Oymak, L.~Balzano, and N.~Ozay,
  ``Certainty equivalent quadratic control for markov jump systems,'' in
  \emph{2022 American Control Conference (ACC)}, 2022, pp. 2871--2878.

\bibitem{lei2015consistency}
J.~Lei, A.~Rinaldo \emph{et~al.}, ``Consistency of spectral clustering in
  stochastic block models,'' \emph{The Annals of Statistics}, vol.~43, no.~1,
  pp. 215--237, 2015.

\bibitem{cho2001comparison}
G.~E. Cho and C.~D. Meyer, ``Comparison of perturbation bounds for the
  stationary distribution of a markov chain,'' \emph{Linear Algebra and its
  Applications}, vol. 335, no. 1-3, pp. 137--150, 2001.

\end{thebibliography}
\end{document}